\documentclass[journal,draftcls,draftclsnofoot,onecolumn,12pt]{IEEEtranTCOM} 

\usepackage[colorlinks=false,linkcolor=black, bookmarksnumbered]{hyperref}

\usepackage{amsmath, amsthm, amsfonts, amssymb, amsbsy}
\usepackage{setspace}
\usepackage{graphicx}
\usepackage{tikz}
\usepackage{cite}
\usepackage{mathtools}

\usetikzlibrary{fit,matrix} 

\usepackage{arydshln} 

\usepackage{algpseudocode}
\usepackage{algorithm}

\usepackage{multirow, multicol}
\usepackage{subfigure}
\usepackage{verbatim}
\usepackage{epstopdf}

\usepackage{bm}

\newtheorem{theorem}{Theorem}

\newtheorem{corollary}{Corollary}

\newtheorem{example}{Example}

\newcommand\bH{{\bf H}}
\newcommand\bM{{\bf M}}

\newcommand\bO{{\bf O}}

\newcommand\ba{{\bf a}}

\newcommand\bG{{\bf G}}

\newcommand\bF{{\bf F}}

\newcommand\bB{{\bf B}}
\newcommand\bb{{\bf b}}

\newcommand\bV{{\bf V}}
\newcommand\bv{{\bf v}}
\newcommand\bg{{\bf g}}

\newcommand\bz{{\bf z}}

\newcommand\bh{{\bf h}}
\newcommand\bp{{\bf p}}

\newcommand\bu{{\bf u}}
\newcommand\bU{{\bf U}}

\newcommand{\hl}[1]{\textit{#1}} 

\newcommand{\cred}[1]{{\color{red}{#1}}} 

\newcommand\GCD{\text{GCD}}

\newcommand\convol{\text{convol}}
\newcommand\local{\text{local}}
\newcommand\free{\text{free}}

\makeatletter

\newcommand{\Rmnum}[1]{\expandafter\@slowromancap\romannumeral #1@}
\makeatother

\begin{document}

\sloppy
%
\title{Convolutional Codes from Cyclic Codes with Guaranteed Free and Local Minimum Distances}
\author{Khaled Abdel-Ghaffar,~\IEEEmembership{Member,~IEEE,} Daniel J. Costello Jr.,~\IEEEmembership{Life~Fellow,~IEEE,} Juane Li, ~\IEEEmembership{Member,~IEEE,} Shu Lin, ~\IEEEmembership{Life~Fellow,~IEEE,} 




} 

\maketitle

\begin{abstract}
This paper presents an algebraic method to construct convolutional codes with guaranteed \emph{free and local minimum distances} without limit based on cyclic codes of odd lengths. The constructions are simple but effective, and no computer search is needed. For any two positive integers $r$ and $t$ with $1 \leq r < t$, a rate-$r/t$ convolutional code $\mathcal{C}_{\text{convol}}$ can be constructed by using a chain of $r$ cyclic codes $\mathcal{C}_0, \mathcal{C}_1, \ldots, \mathcal{C}_{r-1}$ of the same length $n$ which satisfy the inclusion condition, $\mathcal{C}_0 \supset \mathcal{C}_1 \supset \ldots \supset \mathcal{C}_{r-1}$. Such a convolutional code $\mathcal{C}_{\text{convol}}$ is composed of a \emph{semi-infinite chain of identical local codes} confined in a diagonal band of width $n$. Each local code $\mathcal{C}_{\text{local}}$ of $\mathcal{C}_{\text{convol}}$ is formed from the $r$ cyclic codes in the code chain and is a specially localized subcode of the \emph{mother code} $\mathcal{C}_0$ in the code chain. The minimum distance $d_{\text{local}}$ of each local code of $\mathcal{C}_{\text{convol}}$ is lower bounded by the minimum distance $d_0$ of the mother code $\mathcal{C}_0$ in the code chain. The local structure of $\mathcal{C}_{\text{convol}}$ allows it to be decoded based on a designed parity-check matrix of the mother code $\mathcal{C}_0$ using a sliding window decoding scheme.
\end{abstract}



\section{Introduction}
Convolutional codes were introduced by Elias~\cite{elias1955} in 1955 as an alternative to block codes. Since their introduction, extensive research has been conducted into their structural properties, code construction, encoding methods, decoding methods, and applications. Convolutional codes have been widely used in deep space, satellite, broadcast and wireless communications since late 1960’s. The algebraic structure of these codes was first investigated by Forney~\cite{forney1970,forney1973}. Convolutional codes are covered in varying amount of detail in many books on coding theory and digital communications~\cite{dholakia1994,lee1997,johannesson1999,lin2004}. The book by Johannesson and Zigangirov~\cite{johannesson1999} is the most extensive and comprehensive one.

A common measure of the strength of a convolutional code is its free distance. Most of the works on constructing convolutional codes with good free distances depend on computationally expensive search algorithms, and such algorithms have severe limitations as to how large a free distance can be guaranteed. Most convolutional codes that have been constructed with good free distances are codes of rates $1/2$, $1/3$, $1/4$, $2/3$ and $3/4$ with relatively short constraint lengths~\cite{heller1968,larsen1973,paaske1974,justesen1975,johannesson1977,chambers1992,chang2006}.

This paper presents a new category of convolutional codes constructed based on \emph{chains of cyclic codes}. For any two positive integers $r$ and $t$ with $1 \leq r < t$, a rate-$r/t$ convolutional code $\mathcal{C}_{\text{convol}}$ is constructed using a chain of $r$ cyclic codes, $\mathcal{C}_0, \mathcal{C}_1, \ldots, \mathcal{C}_{r-1}$, of the same odd length $n$ which satisfy the inclusion condition, $\mathcal{C}_0 \supset \mathcal{C}_1 \supset \ldots \supset \mathcal{C}_{r-1}$. Such a convolutional code $\mathcal{C}_{\text{convol}}$ is composed of a semi-infinite chain of identical local codes confined in a diagonal band of width $n$. Each local code $\mathcal{C}_{\text{local}}$ of $\mathcal{C}_{\text{convol}}$ is formed from the $r$ cyclic codes in the code chain and is a \emph{specially localized subcode} of $\mathcal{C}_0$. The minimum distance $d_{\text{local}}$ of each local code of $\mathcal{C}_{\text{convol}}$, called \emph{the local minimum distance}, is lower bounded by the minimum distance $d_0$ of the mother code $\mathcal{C}_0$ in the code chain. The construction is simple and effective, and no computer search is needed. The paper presents a bridge to construct convolutional codes from cyclic block codes with distinctive algebraic and geometric structures.

The local structure and guaranteed local minimum distance of such a convolutional code allows the code to be decoded locally with a \emph{sliding window successive cancellation (SWSC)} decoding algorithm. The type of local decoding is determined by the nature of the cyclic codes in the code chain used in the construction. With local decoding, the local minimum distance $d_{\text{local}}$ of the convolutional code $\mathcal{C}_{\text{convol}}$ will determine its performance. Alternatively, $\mathcal{C}_{\text{convol}}$ can be decoded globally using typical algorithms for a convolutional code, such as Viterbi decoding, MAP decoding (for codes of short constraint lengths)~\cite{viterbi1967,bahl1974} or sequential decoding for codes of long constraint lengths~\cite{wozencraft1961}. In this case the free distance $d_{\text{free}}$ of $\mathcal{C}_{\text{convol}}$ will determine its performance. If the mother code $\mathcal{C}_0$ in the code chain $\mathcal{C}_0 \supset \mathcal{C}_1 \supset \ldots \supset \mathcal{C}_{r-1}$ is one-step majority-logic decodable, the convolutional code $\mathcal{C}_{\text{convol}}$ is self-orthogonal and can be decoded with simple majority-logic decoding~\cite{massey1963,lin2004}. If the mother code $\mathcal{C}_0$ of the code chain is a cyclic LDPC code, the local codes of the convolutional code can be decoded with an iterative decoding algorithm based on the principle of belief-propagation~\cite{gallager1962,mackay1999}. The rate-$r/t$ convolutional code $\mathcal{C}_{\text{convol}}$ constructed based on the code chain $\mathcal{C}_0 \supset \mathcal{C}_1 \supset \ldots \supset \mathcal{C}_{r-1}$ is the \emph{direct-sum} of $r$ rate-$1/t$ convolutional codes, each constructed based on a code in the code chain $\mathcal{C}_0 \supset \mathcal{C}_1 \supset \ldots \supset \mathcal{C}_{r-1}$.

There have been other works that have explored the relationship between block and convolutional codes~\cite{massey1973,justesen1973,solomon1979,tanner1987,esmaeili1998,rosenthal1999,smarandache2001,pusane2011}. These works use block codes to construct convolutional codes with good free distances, and they provide interesting links between the theories of block codes and convolutional codes. The constructions of convolutional codes in these works use quite different techniques and have different perspectives in contrast to the work in this paper. The paper~\cite{massey1973} uses the minimum distance properties of a cyclic code to provide a lower bound on the free distance of an associated convolutional code. It shows that if the minimum distances of a cyclic code and its dual code are $d$ and $d_h$, respectively, the free distance $d_{\text{free}}$ of the rate-$1/2$ convolutional code associated with the cyclic code is lower bounded by $\min\{d, 2d_h\}$. The paper~\cite{justesen1973} gives a construction of a convolutional code based on a cyclic code with free distance $d_{\text{free}}$ lower bounded by the minimum distance $d$ of the cyclic code, but it involves a rather complicated condition of the roots of the generator polynomial of the cyclic code. The papers~\cite{solomon1979,tanner1987,esmaeili1998} provide links between quasi-cyclic codes and convolutional codes. The paper~\cite{rosenthal1999} presents a construction of convolutional codes which is similar to the construction of classical BCH codes. The construction results in high rate codes with good free distances, but the construction requires fields of large sizes. The paper~\cite{smarandache2001} presents a construction of a class of nonbinary MDS-convolutional codes based on Reed-Solomon codes. The paper~\cite{pusane2011} shows how LDPC convolutional codes can be constructed from LDPC block codes.

The rest of the paper is organized as follows. Section \ref{sect2:fundamentals_cyclic_code} first gives a brief review of cyclic codes and their fundamental structural properties. Section \ref{sect3:chain_from_cyclic_code} presents cyclic code chains with inclusion structure. Cyclic code chains with inclusion structure will be used for constructing convolutional codes of various rates and constraint lengths with guaranteed local minimum distances. In Section \ref{sect4:convol_from_cyclic_code}, a new method is introduced to construct a category of convolutional codes based on chains of cyclic codes with inclusion structure. Each convolutional code in this category has a distinctive local structure. Also in this section, an SWSC decoding algorithm is presented. Sections \ref{sect5:convol_from_BCH} to \ref{sect8:convol_from_DTI}  present a wide variety of convolutional codes constructed based on well-known classes of cyclic codes, such as BCH codes, cyclic Reed-Muller (RM) codes, finite geometry (FG) LDPC codes, and doubly transitive invariant (DTI) codes. Local minimum distances and free distances of convolutional codes in these categories are analyzed. Section \ref{sect9:conclusion} concludes the paper with some remarks.

\section{Fundamental Structural Properties of a Cyclic Code} \label{sect2:fundamentals_cyclic_code}

Let $\mathcal{C}$ be a binary $(n, k)$ cyclic code over $\mathrm{GF}(2)$ of odd length $n$ with dimension $k$ and minimum distance $d$ generated by a monic polynomial of degree $n-k$ over $\mathrm{GF}(2)$~\cite{lin2004}:
\begin{equation}
\bg(X) = g_0 + g_1 X + \cdots + g_{n-k-1} X^{n-k-1} + g_{n-k} X^{n-k}
\end{equation}
with $g_0 = g_{n-k} = 1$. The polynomial $\bg(X)$ is called the generator polynomial of $\mathcal{C}$ and is a factor of $X^n + 1$. Each code polynomial in $\mathcal{C}$,
\[
\bv(X) = v_0 + v_1 X + v_2 X^2 + \cdots + v_{n-1} X^{n-1},
\]
is a polynomial of degree $n-1$ or less over $\mathrm{GF}(2)$ which is a multiple of $\bg(X)$. The generator polynomial $\bg(X)$ is the unique code polynomial in $\mathcal{C}$ with the smallest degree $n-k$. The $n$-tuple $\mathbf{v} = (v_0, v_1, v_2, \ldots, v_{n-1})$ over $\mathrm{GF}(2)$ associated with the code polynomial $\bv(X)$ is called a codeword in $\mathcal{C}$.

The codeword associated with the generator polynomial $\bg(X)$ is an $n$-tuple over $\mathrm{GF}(2)$:
\begin{equation}
\bg = (\underbracket{g_0, g_1, g_2, \ldots, g_{n-k-1}, g_{n-k}}_{n-k+1}, \underbracket{0, 0, \ldots, 0}_{k-1}),
\end{equation}
where $g_0 = g_{n-k} = 1$. Labeling the components of $\mathbf{g}$ from $0$ to $n-1$, we see that the nonzero components of $\mathbf{g}$ are confined to the first $n-k+1$ consecutive positions, and the last $k-1$ components of $\mathbf{g}$ are zeros which form a \emph{zero-span} of length $k-1$. A zero-span of length $l$ in an $n$-tuple is defined as a sequence of $l$ consecutive zeros confined between two 1-components (including the end-around case). The length of the longest zero-span in $\mathbf{g}$ is $k-1$, which is the length of the \emph{end-around zero-span} of $\mathbf{g}$~\cite{lin2004}. The codeword $\mathbf{g}$ associated with the generator polynomial $\bg(X)$ of the $(n, k)$ cyclic code $\mathcal{C}$ is called the \hl{generator vector} (or codeword) of $\mathcal{C}$. The weight $w(\mathbf{g})$ of $\mathbf{g}$ is at least the minimum distance $d$ of $\mathcal{C}$. If $w(\mathbf{g}) = d$, then $\mathbf{g}$ is a minimum weight codeword and $\bg(X)$ is a minimum weight code polynomial. The $(n-k+1)$-tuple formed by the first $n-k+1$ components of $\mathbf{g}$ (the $n-k+1$ coefficients of $\bg(X)$),
\begin{equation}
(g_0, g_1, \ldots, g_{n-k})
\end{equation}
is called the \hl{generator sequence} of $\mathcal{C}$.

Using the generator vector $\mathbf{g}$ and its $k-1$ cyclic-shifts to the right, one position at a time, we obtain the following \hl{full rank} $k \times n$ matrix over $\mathrm{GF}(2)$:
\begin{equation} \label{eqn:G_matrix_cyclic_code}
\mathbf{G} = 
\begin{bmatrix}
1 & g_1 & g_2 & \cdots & g_{n-k-1} & 1 & 0 & \cdots & 0 \\
0 & 1 & g_1 & g_2 & \cdots & g_{n-k-1} & 1 & 0 & \cdots \\
\vdots & & & \ddots & & & & & \\
0 & 0 & \cdots & 1 & g_1 & g_2 & \cdots & g_{n-k-1} & 1
\end{bmatrix}.
\end{equation}
The rows of $\mathbf{G}$ are linearly independent, and the row space of $\mathbf{G}$ gives the $(n, k)$ cyclic code $\mathcal{C}$. The matrix $\mathbf{G}$ is called a \hl{generator matrix} of $\mathcal{C}$. The nonzero entries in $\mathbf{G}$ are confined to a diagonal band of width $n-k+1$ in which all the entries lying on the left and right diagonal borders are 1's. Each row within the diagonal band is identical to the generator sequence $(1, g_1, g_2, \ldots, g_{n-k-1}, 1)$. The generator matrix $\mathbf{G}$ of $\mathcal{C}$ in the form of (\ref{eqn:G_matrix_cyclic_code}) is referred to as the \hl{diagonal} (or \hl{cyclic}) generator matrix of $\mathcal{C}$. The generator matrix of $\mathcal{C}$ can be expressed in polynomial form as follows:
\begin{equation}
\mathbf{G}(X) = [\bg(X), X\bg(X), X^2\bg(X), \ldots, X^{k-1}\bg(X)]^T.
\end{equation}

Let $\mathbf{H}$ be a parity-check matrix of $\mathcal{C}$. Then, $\mathbf{G}\mathbf{H}^T = \mathbf{O}$. The null space of $\mathbf{H}$ gives the $(n, k)$ cyclic code $\mathcal{C}$. The row space of $\mathbf{H}$ gives an $(n, n-k)$ cyclic code $\mathcal{C}_h$ which is the dual code of $\mathcal{C}$. For every codeword $\mathbf{v}$ in $\mathcal{C}$, the condition $\mathbf{v}\mathbf{H}^T = \mathbf{o}$ holds. In general, encoding of a cyclic code $\mathcal{C}$ is based on its generator matrix $\mathbf{G}$ (or its generator polynomial $\bg(X)$) and decoding is based on a parity-check matrix $\mathbf{H}$, which is designed for implementation of a specific decoding algorithm.

Since $\bg(X)$ is a factor of $X^n + 1$ with degree $n-k$, $X^n + 1$ can be factored as the product of $\bg(X)$ and a polynomial $\bh(X)$ of degree $k$ as follows,
\begin{equation}
X^n + 1 = \bg(X) \bh(X)
\end{equation}
where
\begin{equation}
\bh(X) = h_0 + h_1 X + \cdots + h_{k-1} X^{k-1} + h_k X^k
\end{equation}
with $h_0 = h_k = 1$. The reciprocal of $\bh(X)$ is defined as
\begin{equation}
\bg_h(X) = X^k \bh(X^{-1}) = h_k + h_{k-1} X + \cdots + h_1 X^{k-1} + h_0 X^k
\end{equation}
which is also a factor of $X^n + 1$. The polynomial $\bg_h(X)$ is the generator polynomial of the $(n, n-k)$ dual code $\mathcal{C}_h$ of the $(n, k)$ cyclic code $\mathcal{C}$ generated by $\bg(X)$. Hence, $\bg(X)$ and $\bg_h(X)$ generate $\mathcal{C}$ and its dual code $\mathcal{C}_h$, respectively.

Let $\bb(X)$ be a factor of $X^n + 1$ and with degree $\ell$ which is not contained in $\bg(X)$, i.e., $\bb(X)$ and $\bg(X)$ are relatively prime. Then $\bg(X) \bb(X)$ generates an $(n, k-\ell)$ cyclic subcode of $\mathcal{C}$ with minimum distance at least $d$, the minimum distance of $\mathcal{C}$. If $\bb(X) = 1 + X$ which is not a factor of $\bg(X)$, the code generated by $\bg(X)(1 + X)$ is the $(n, k-1)$ even weight subcode of $\mathcal{C}$.

\section{Chains and Localized Subcodes of Cyclic Codes}\label{sect3:chain_from_cyclic_code}

In this section, we present a method to construct a chain of cyclic codes of the same odd length $n$ which satisfies an inclusion condition: (1) each code in the chain is a proper subcode of the code preceding it in the chain; (2) the first code in the chain, called the \hl{mother code}, contains the rest of the codes in the chain as subcodes; and (3) the end code in the chain is a subcode of all the other codes in the chain.

Using the code chain, a convolutional code composing of a \hl{semi-infinite chain} of \hl{identical local codes} can be constructed. Each local code of the convolutional code is formed from the cyclic codes in the code chain and is a \hl{specially localized subcode} of the mother code in the code chain. The minimum distance of each local code of the convolutional code is lower bounded by the minimum distance of the mother code in the code chain. Using cyclic code chains, a large class of convolutional codes with distinct local structure of various rates and constraint lengths, short to very long, can be constructed based on various types of cyclic codes.

\subsection{Construction of Chains of Cyclic Codes}

Let $\mathcal{C}_0$ be an $(n, k_0)$ cyclic code over $\mathrm{GF}(2)$ of length $n$ with dimension $k_0$ and minimum distance $d_0$ generated by the polynomial
\[
\bg_0(X) = g_{0,0} + g_{0,1} X + \cdots + g_{0,n-k_0} X^{n-k_0}
\]
with $g_{0,0} = g_{0,n-k_0} = 1$. Let $r$ be a positive integer and let ${\bf f}_1(X), \mathbf{f}_2(X), \ldots, \mathbf{f}_{r-1}(X)$ be $r-1$ polynomials over $\mathrm{GF}(2)$ with degrees $\ell_1, \ell_2, \ldots, \ell_{r-1}$, respectively, and $1 \leq \ell_1 < \ell_2 < \cdots < \ell_{r-1} < k_0$, which have the following properties: (1) they are factors of $X^n + 1$; (2) for $1 \leq i < r$, $\mathbf{f}_i(X)$ and $\bg_0(X)$ are relatively prime, i.e., $\GCD(\mathbf{f}_i(X), \bg_0(X)) = 1$; and (3) for $1 \leq i < r-1$, $\mathbf{f}_i(X)$ divides $\mathbf{f}_{i+1}(X)$ with $\mathbf{f}_{i+1}(X) = \mathbf{b}_{i+1}(X) \mathbf{f}_i(X)$ and $\mathbf{b}_{i+1}(X) \neq 1$.

Set $\mathbf{f}_0(X) = 1$. Using $\bg_0(X), \mathbf{f}_0(X), \mathbf{f}_1(X), \ldots, \mathbf{f}_{r-1}(X)$, we form $r$ polynomials
\begin{equation} \label{eq:r_genenator_polynomials}
\begin{array}{c}
\bg_0(X) = \bg_0(X) {\bf f}_0(X), \\
\bg_1(X) = \bg_0(X) {\bf f}_1(X), \\
\bg_2(X) = \bg_0(X) {\bf f}_2(X), \\
\vdots \\
\bg_{r-1}(X) = \bg_0(X) {\bf f}_{r-1}(X). \\
\end{array}
\end{equation}
Each of the polynomials $\bg_0(X), \bg_1(X), \ldots, \bg_{r-1}(X)$ is a factor of $X^n + 1$, and for $0 \leq i < r-1$, $\bg_{i+1}(X)$ is a multiple of $\bg_i(X)$, i.e., $\bg_i(X)$ divides $\bg_{i+1}(X)$. The $r-1$ cyclic codes $\mathcal{C}_1, \mathcal{C}_2, \ldots, \mathcal{C}_{r-1}$ generated by $\bg_1(X), \bg_2(X), \ldots, \bg_{r-1}(X)$ are proper subcodes of $\mathcal{C}_0$, and for $0 \leq i < r-1$, $\mathcal{C}_{i+1}$ is a subcode of $\mathcal{C}_i$. Hence, the $r$ cyclic codes $\mathcal{C}_0, \mathcal{C}_1, \ldots, \mathcal{C}_{r-1}$ form an inclusion code chain
\begin{equation}
\mathcal{C}_0 \supset \mathcal{C}_1 \supset \cdots \supset \mathcal{C}_{r-1}.
\end{equation}
The first code $\mathcal{C}_0$ in the code chain is called the \hl{mother code} and the codes $\mathcal{C}_1, \ldots, \mathcal{C}_{r-1}$ in the code chain are called the \hl{descendant codes} of $\mathcal{C}_0$. The parameter $r$ is referred to as the length of the code chain. The polynomials $\mathbf{f}_0(X) = 1, \mathbf{f}_1(X), \ldots, \mathbf{f}_{r-1}(X)$ are called \hl{generator multipliers} for the code chain. The dimensions of $\mathcal{C}_0, \mathcal{C}_1, \ldots, \mathcal{C}_{r-1}$ are $k_0, k_1 = k_0 - \ell_1, k_2 = k_0 - \ell_2, \ldots, k_{r-1} = k_0 - \ell_{r-1}$, respectively. Let $d_0, d_1, \ldots, d_{r-1}$ be the minimum distances of $\mathcal{C}_0, \mathcal{C}_1, \ldots, \mathcal{C}_{r-1}$. Then, for $1 \leq i < r$, $d_{i-1} \leq d_i$.

For $0 \leq i < r$, the generator polynomial of the $i$-th cyclic code $\mathcal{C}_i$ in the code chain is a polynomial of degree $n - k_i$:
\begin{equation}
\bg_i(X) = g_{i,0} + g_{i,1} X + \cdots + g_{i,n-k_i} X^{n-k_i}
\end{equation}
with $g_{i,0} = g_{i,n-k_i} = 1$. The generator vector and the generator sequence of $\mathcal{C}_i$ are
\begin{equation}
\bg_i = (\underbracket{g_{i,0}, g_{i,1}, g_{i,2}, \ldots, g_{i,n-k_i-1}, g_{i,n-k_i}}_{n-k_i+1}, \underbracket{0, 0, \ldots, 0}_{k_i-1}),
\end{equation}
and
\begin{equation} \label{eqn:cyclic_code_generator_sequence}
(g_{i,0}, g_{i,1}, \ldots, g_{i,n-k_i})
\end{equation}
respectively. The length of the end-around zero-span of $\mathbf{g}_i$ is $k_i - 1$.

From the (\ref{eq:r_genenator_polynomials}) to (\ref{eqn:cyclic_code_generator_sequence}), we see that: (1) the generator polynomials $\bg_1(X), \bg_2(X), \ldots, \bg_{r-1}(X)$ of $\mathcal{C}_1, \mathcal{C}_2, \ldots, \mathcal{C}_{r-1}$ have the generator polynomial $\bg_0(X)$ of $\mathcal{C}_0$ as a common factor; (2) the ending (or rightmost) 1-components of the generator vectors $\mathbf{g}_0, \mathbf{g}_1, \ldots, \mathbf{g}_{r-1}$ of $\mathcal{C}_0, \mathcal{C}_1, \ldots, \mathcal{C}_{r-1}$ are in $r$ different locations; (3) for $1 \leq i < r$, the nonzero-span of $\mathbf{g}_i$ is longer than the nonzero-span of $\mathbf{g}_{i-1}$; (4) for $1 \leq i < r$, the length of the end-around zero-span of $\mathbf{g}_{i-1}$ is longer than the end-around zero-span of $\mathbf{g}_i$; and (5) $\mathbf{g}_0, \mathbf{g}_1, \ldots, \mathbf{g}_{r-1}$ are linearly independent.

\subsection{Localized Subcodes of Cyclic Codes in a Code Chain}

Let $t$ be a positive integer greater than $r$ and less than or equal to the zero-span $k_{r-1} - 1$ of the generator vector $\mathbf{g}_{r-1}$ of the end code $\mathcal{C}_{r-1}$ of the code chain $\mathcal{C}_0 \supset \mathcal{C}_1 \supset \cdots \supset \mathcal{C}_{r-1}$, i.e., $r < t \leq k_{r-1} - 1$. Dividing $k_{r-1} - 1$ by $t$, we have
\begin{equation}
k_{r-1} - 1 = b t + c
\end{equation}
where $b$ and $c$ are the quotient and remainder, respectively, with $b \geq 1$ and $0 \leq c < t$. Using the generator vectors $\mathbf{g}_0, \mathbf{g}_1, \ldots, \mathbf{g}_{r-1}$ of the $r$ cyclic codes in the code chain $\mathcal{C}_0 \supset \mathcal{C}_1 \supset \cdots \supset \mathcal{C}_{r-1}$, we form the following $r \times n$ matrix over $\mathrm{GF}(2)$:
\begin{equation} \label{eq:M_matrix}
\bM (\bg_0, \bg_1, \ldots, \bg_{r-1}) = \left[ \begin{array}{c} 
\bg_0 \\
\bg_1 \\
\vdots \\
\bg_{r-1}
\end{array} \right] = \left[ \bB(\bg_0, \bg_1, \ldots, \bg_{r-1}) ~ \bO ~ \bO ~ \ldots ~ \bO \right],
\end{equation}  
which consists of an $r \times (n - b t)$ submatrix $\bB(\mathbf{g}_0, \mathbf{g}_1, \ldots, \mathbf{g}_{r-1})$ and $b$ copies of the $r \times t$ zero matrix $\bO$. The submatrix $\bB(\mathbf{g}_0, \mathbf{g}_1, \ldots, \mathbf{g}_{r-1})$ of $\bM(\mathbf{g}_0, \mathbf{g}_1, \ldots, \mathbf{g}_{r-1})$ is given as follows:

\begin{equation} \label{eq:B_matrix}
\begin{array} {l}
\bB(\bg_0, \bg_1, \ldots, \bg_{r-1}) = \\
=\underbracket{\left[
\begin{array}{ccccccccccccc}
1 & g_{0, 1} & g_{0, 2} & \cdots & g_{0, n-k_0-1} & 1 & 0 & 0 & \cdots & 0  \\
1 & g_{1, 1} & g_{1, 2} & \cdots & \cdots & g_{1, n-k_1-1} & 1 & 0 & \cdots & 0  \\
\vdots  & \vdots  & \vdots\\
1 & g_{r-1, 1} & g_{r-1, 2} & \cdots & \cdots & \cdots & \cdots & \cdots & g_{r-1, n-k_{r-1}-1} & 1 \\
\end{array}\right.}_{ n- k_{r-1} + 1}
\underbracket{\left.\begin{array}{ccc}
0 & \cdots & 0 \\
0 & \cdots & 0 \\
\vdots  & \vdots  & \vdots\\
0 & \cdots & 0 \\
\end{array}
\right]}_{c}
\end{array}
\end{equation}

\noindent in which, for $0 \leq i < r$, the $i$-th row consists of the first $n - b t$ components of the generator vector $\mathbf{g}_i$ of the $i$-th cyclic code $\mathcal{C}_i$ in the code chain. The nonzero entries in $\bB(\mathbf{g}_0, \mathbf{g}_1, \ldots, \mathbf{g}_{r-1})$ are confined in the first $n - k_{r-1} + 1$ columns and, for $c \neq 0$, the last $c$ columns of $\bB(\mathbf{g}_0, \mathbf{g}_1, \ldots, \mathbf{g}_{r-1})$ are zero columns. For $1 \leq i < r$, the top $i$ entries of the column in $\bB(\mathbf{g}_0, \mathbf{g}_1, \ldots, \mathbf{g}_{r-1})$ that contains the rightmost 1-component of the $i$-th row are all zeros. From this, we see that all the rows of $\bB(\mathbf{g}_0, \mathbf{g}_1, \ldots, \mathbf{g}_{r-1})$ are linearly independent.

Next, we form the following $(b+1)r \times n$ matrix over $\mathrm{GF}(2)$ by cyclically shifting the rows of $\bM(\mathbf{g}_0, \mathbf{g}_1, \ldots, \mathbf{g}_{r-1})$ together to the right $b$ times, $t$ positions at a time, called the $t$-position-cyclic-shift ($t$-PCS),
\begin{equation}\label{eq:G_local_matrix}
\bG_{\local} (\bg_0, \bg_1, \ldots, \bg_{r-1}) = \left[ \begin{array}{ccccc}
\bB(\bg_0, \bg_1, \ldots, \bg_{r-1}) & \bO & \bO  & \cdots & \bO \\
\bO & \bB(\bg_0, \bg_1, \ldots, \bg_{r-1}) & \bO  & \cdots & \bO \\
\vdots & \vdots & \vdots \\
\bO & \bO & \cdots & \bO & \bB(\bg_0, \bg_1, \ldots, \bg_{r-1})
\end{array} \right], 
\end{equation}
in which the last $c$ columns, for $c \neq 0$, are zero columns.

The row space of $\bG_{\text{local}}(\mathbf{g}_0, \mathbf{g}_1, \ldots, \mathbf{g}_{r-1})$ gives a linear subcode, denoted by $\mathcal{C}_{\text{local}}$, of $\mathcal{C}_0$ and is called the \hl{$(r, t)$-PCS-local subcode} of $\mathcal{C}_0$. For $c \neq 0$, the last $c$ components of each nonzero codeword in $\mathcal{C}_{\text{local}}$ are zeros. The minimum distance $d_{\text{local}}$ of $\mathcal{C}_{\text{local}}$ is lower bounded by the minimum distance $d_0$ of the mother code $\mathcal{C}_0$ in the code chain and it is called the \hl{$(r, t)$-PCS-local minimum distance} (simply the local minimum distance) of $\mathcal{C}_0$.

From (\ref{eq:G_local_matrix}), we see that the generator matrix $\bG_{\text{local}}(\mathbf{g}_0, \mathbf{g}_1, \ldots, \mathbf{g}_{r-1})$ of $\mathcal{C}_{\text{local}}$ also has the diagonal-band structure with $b+1$ copies of the matrix $\bB(\mathbf{g}_0, \mathbf{g}_1, \ldots, \mathbf{g}_{r-1})$ lying on its main diagonal. We call $\bB(\mathbf{g}_0, \mathbf{g}_1, \ldots, \mathbf{g}_{r-1})$ the \hl{belt generator matrix} of $\mathcal{C}_{\text{local}}$ and the number $r$ of rows in $\bB(\mathbf{g}_0, \mathbf{g}_1, \ldots, \mathbf{g}_{r-1})$ is referred to as the \hl{width of the belt}. The integer $t$ is called the \hl{shifting factor}. The matrix $\bM(\mathbf{g}_0, \mathbf{g}_1, \ldots, \mathbf{g}_{r-1})$ is called the \hl{bend matrix} of $\bG_{\text{local}}(\mathbf{g}_0, \mathbf{g}_1, \ldots, \mathbf{g}_{r-1})$ and the integer $b$ is called the \hl{shifting span}.

From (\ref{eq:B_matrix}) and (\ref{eq:G_local_matrix}), we see that $\bG_{\text{local}}(\mathbf{g}_0, \mathbf{g}_1, \ldots, \mathbf{g}_{r-1})$ is composed of selected rows from the generator matrices $\bG_0, \bG_1, \ldots, \bG_{r-1}$ of the $r$ cyclic codes $\mathcal{C}_0, \mathcal{C}_1, \ldots, \mathcal{C}_{r-1}$ in the code chain, $b+1$ rows from each generator matrix. The $b+1$ selected rows from $\bG_i$ are the generator vector $\mathbf{g}_i$ of $\mathcal{C}_i$ and its $b$ $t$-PCSs and they are linearly independent. These $b+1$ rows of $\bG_i$ form a $(b+1) \times n$ submatrix, denoted by $\bG_{i,\text{local}}(\mathbf{g}_i)$, of $\bG_i$ which is called the $(1, t)$-PCS-local submatrix of $\bG_i$. The row space of $\bG_{i,\text{local}}(\mathbf{g}_i)$ gives an $(n, b+1)$ $(1, t)$-PCS-local subcode $\mathcal{C}_{i,\text{local}}$ of the $i$-th cyclic code $\mathcal{C}_i$ in the code chain with rate $R_{i,\text{local}} = (b+1)/n = (k_{r-1} + t - c - 1)/nt$ and minimum distance $d_{i,\text{local}}$ at least $d_0$. The $(1, t)$-PCS-local subcodes $\mathcal{C}_{0,\text{local}}, \mathcal{C}_{1,\text{local}}, \ldots, \mathcal{C}_{r-1,\text{local}}$ of the $r$ cyclic codes in the code chain $\mathcal{C}_0 \supset \mathcal{C}_1 \supset \cdots \supset \mathcal{C}_{r-1}$ are subcodes of the $(r, t)$-PCS-local subcode $\mathcal{C}_{\text{local}}$ of the mother code $\mathcal{C}_0$ in the code chain, and they are referred to as the \hl{constituent $(1, t)$-PCS-local subcodes} of $\mathcal{C}_{\text{local}}$.

\subsection{Direct-Sum Structure of the $(r, t)$-PCS-Local Subcode}

If the $(b+1)r$ rows of $\bG_{\text{local}}(\mathbf{g}_0, \mathbf{g}_1, \ldots, \mathbf{g}_{r-1})$ are linearly independent, then $\mathcal{C}_{\text{local}}$ is the \emph{direct-sum} of $\mathcal{C}_{0,\text{local}}, \mathcal{C}_{1,\text{local}}, \ldots, \mathcal{C}_{r-1,\text{local}}$. Conversely, if $\mathcal{C}_{\text{local}}$ is the direct-sum of its $r$ constituent $(1, t)$-PCS-local subcodes, the $(b+1)r$ rows of $\bG_{\text{local}}(\mathbf{g}_0, \mathbf{g}_1, \ldots, \mathbf{g}_{r-1})$ are linearly independent. In this case, $\mathcal{C}_{\text{local}}$ is an $(n, (b+1)r)$ code with rate
\[
R_{\text{local}} = \frac{(b+1)r}{n} = \frac{(k_{r-1} + t - c - 1)r}{nt}
\]
and minimum distance $d_{\text{local}}$ lower bounded by the minimum distance $d_0$ of the mother code $\mathcal{C}_0$ in the code chain. For $\mathcal{C}_{\text{local}}$ to be the direct-sum of its $r$ constituent $(1, t)$-PCS-local subcodes, no two of its constituent $(1, t)$-PCS-local subcodes can have a nonzero codeword in common.

A condition for choosing the generator multipliers $\mathbf{f}_0(X) = 1, \mathbf{f}_1(X), \ldots, \mathbf{f}_{r-1}(X)$ for the code chain such that no two constituent $(1, t)$-PCS-local subcodes of $\mathcal{C}_{\text{local}}$ have a nonzero codeword in common is given in Theorem~\ref{Thm1}.

\begin{theorem} \label{Thm1}
For $0 \leq i < j < r$, the two constituent subcodes $\mathcal{C}_{i,\text{local}}$ and $\mathcal{C}_{j,\text{local}}$ do not have a nonzero codeword in common if and only if the ratio $\mathbf{f}_j(X)/\mathbf{f}_i(X)$ of the generator multipliers $\mathbf{f}_j(X)$ and $\mathbf{f}_i(X)$ is not a polynomial in $X^t$.
\end{theorem}

\begin{proof}
Any nonzero code polynomial in $\mathcal{C}_{i,\text{local}}$ can be written as $\ba_i(X^t)\bg_i(X) = \ba_i(X^t)\mathbf{f}_i(X)\bg_0(X)$ for some nonzero polynomial $\ba_i(X^t)$ of degree at most $b$. Similarly, any nonzero code polynomial in $\mathcal{C}_{j,\text{local}}$ can be written as $\ba_j(X^t) \bg_j(X) = \ba_j(X^t) \mathbf{f}_j(X)\bg_0(X)$ for some nonzero polynomial $\ba_j(X^t)$. Since, for $i < j$, $\mathbf{f}_j(X)$ is divisible by $\mathbf{f}_i(X)$, we can write $\mathbf{f}_j(X) = \bb_j(X)\mathbf{f}_i(X)$ for some nonzero polynomial $\bb_j(X)$. Hence, a nonzero code polynomial $\bv(X)$ in $\mathcal{C}_{j,\text{local}}$ can be written as $\ba_j(X^t) \bb_j(X)\mathbf{f}_i(X)\bg_0(X)$. If this is a nonzero code polynomial in $\mathcal{C}_{i,\text{local}}$, then $\ba_j(X^t) \bb_j(X)\mathbf{f}_i(X)\bg_0(X) = \ba_i(X^t)\mathbf{f}_i(X)\bg_0(X)$ for some nonzero polynomial $\ba_i(X^t)$. This says that $\bb_j(X) = \ba_i(X^t)/\ba_j(X^t)$ is a nonzero polynomial in $X^t$ and the ratio $\mathbf{f}_j(X)/\mathbf{f}_i(X) = \bb_j(X) = \ba_i(X^t)/\ba_j(X^t)$ is a nonzero polynomial in $X^t$. This says that if $\mathcal{C}_{i,\text{local}}$ and $\mathcal{C}_{j,\text{local}}$ with $0 \leq i < j < r$ have a nonzero code polynomial in common, the ratio $\mathbf{f}_j(X)/\mathbf{f}_i(X)$ must be a polynomial in $X^t$. Conversely, this implies that $\mathcal{C}_{i,\text{local}}$ and $\mathcal{C}_{j,\text{local}}$ do not have a nonzero code polynomial in common if and only if the ratio $\mathbf{f}_j(X)/\mathbf{f}_i(X)$ is not a polynomial in $X^t$ for $0 \leq i < j < r$. This proves the theorem. 
\end{proof}

Therefore, in the design of a code chain $\mathcal{C}_0 \supset \mathcal{C}_1 \supset \cdots \supset \mathcal{C}_{r-1}$ of length $r$ for which the $(r, t)$-PCS-local subcode $\mathcal{C}_{\text{local}}$ of the mother code $\mathcal{C}_0$ in the code chain is the direct-sum of its $r$ constituent $(1, t)$-PCS-local subcodes, we must choose the generator multipliers $\mathbf{f}_0(X) = 1, \mathbf{f}_1(X), \ldots, \mathbf{f}_{r-1}(X)$ for the code chain to satisfy the ratio condition given in Theorem~\ref{Thm1}. Note that the multiplier ratio condition does not guarantee that $\mathcal{C}_{\text{local}}$ is the direct-sum of its $r$ constituent $(1, t)$-PCS-local subcodes. That is, it is a necessary condition but not a sufficient condition. In the following, we derive a sufficient condition under which $\mathcal{C}_{\text{local}}$ is the direct-sum of its $r$ constituent $(1, t)$-PCS-local subcodes.

Given a set $\Omega$ of $r$ generator multipliers $\mathbf{f}_0(X) = 1, \mathbf{f}_1(X), \ldots, \mathbf{f}_{r-1}(X)$ for a code chain $\mathcal{C}_0 \supset \mathcal{C}_1 \supset \cdots \supset \mathcal{C}_{r-1}$ of length $r$ which satisfy the ratio condition given by Theorem \ref{Thm1}, we next derive the necessary and sufficient condition on the set $\Omega$ of $r$ generator multipliers such that the $(r, t)$-PCS-local subcode $\mathcal{C}_{\text{local}}$ of the mother code $\mathcal{C}_0$ in the code chain is the direct-sum of its $r$ constituent $(1, t)$-PCS-local subcodes. This is equivalent to deriving the necessary and sufficient condition on the set $\Omega$ of $r$ generator multipliers such that the rows of the generator matrix $\bG_{\text{local}}(\mathbf{g}_0, \mathbf{g}_1, \ldots, \mathbf{g}_{r-1})$ of $\mathcal{C}_{\text{local}}$ are linearly independent.

To derive the necessary and sufficient condition for the rows in $\bG_{\text{local}}(\mathbf{g}_0, \mathbf{g}_1, \ldots, \mathbf{g}_{r-1})$ to be linearly independent, we re-arrange the rows in $\bG_{\text{local}}(\mathbf{g}_0, \mathbf{g}_1, \ldots, \mathbf{g}_{r-1})$ and express it in polynomial form as follows:
\begin{equation} \label{eq:G_local_matrix_2}
\bG_{\local}(\bg_0, \bg_1, \ldots, \bg_{r-1})(X) = \left[ \begin{array}{c}
\bG_{0, \local}( \bg_0(X) ) \\
\bG_{1, \local}( \bg_1(X) ) \\
\vdots \\
\bG_{r-1, \local}( \bg_{r-1}(X) ) \\
\end{array} \right]
\end{equation}
in which $\bG_{i,\text{local}}(\bg_i(X))$ is composed of $\bg_i(X)$ and its $b$ $t$-PCSs in order. The $j$-th row in $\bG_{i,\text{local}}(\bg_i(X))$ is $X^{jt}\bg_i(X)$, where $0 \leq i < r$ and $0 \leq j \leq b$. A linear combination of the rows of the $(b+1)r \times n$ matrix $\bG_{\text{local}}(\mathbf{g}_0, \mathbf{g}_1, \ldots, \mathbf{g}_{r-1})(X)$ in which the $j$-th row in $\bG_{i,\text{local}}(\bg_i(X))$ is multiplied by $a_{i,j}$ for $0 \leq i < r$ and $0 \leq j \leq b$, gives a polynomial
\begin{equation}
\ba_0 (X^t) \bg_0(X) + \ba_1(X^t) \bg_1(X) + \ldots + \ba_{r-1}(X^t) \bg_{r-1}(X),
\end{equation}
where
\begin{equation}
\ba_i(X) = a_{i, 0} + a_{i, 1} X + a_{i, 2} X^2 + \ldots + a_{i, b} X^b
\end{equation}
is a polynomial over $\mathrm{GF}(2)$ of degree at most $b$.

Therefore, the question about whether or not the $(b + 1)r$ rows of $\bG_{\local}(\bg_0, \bg_1, \ldots, \bg_{r-1})(X)$ (or $\bG_{\local}(\bg_0, \bg_1, \ldots, \bg_{r-1})$) given in (\ref{eq:G_local_matrix})) are linearly independent is equivalent to the question whether or not there are polynomials $\ba_0(X)$, $\ba_1(X)$, $\ldots$, $\ba_{r-1}(X)$, not all zeros, each of degree at most $b$, such that 

\begin{equation}\label{eq:linear_condition}
\ba_0 (X^t) \bg_0(X) + \ba_1(X^t) \bg_1(X) + \ldots + \ba_{r-1}(X^t) \bg_{r-1}(X) = 0.
\end{equation}

Since $\bg_i(X) = {\bf f}_i (X) \bg_0(X)$ for $0 \le i <r$, Equation (\ref{eq:linear_condition}) reduces to 

\begin{equation}\label{eq:linear_condition2}
\ba_0 (X^t) {\bf f}_0(X) + \ba_1(X^t) {\bf f}_1(X) + \ldots + \ba_{r-1}(X^t) {\bf f}_{r-1}(X) = 0.
\end{equation}

For each $0 \le i < r$, we can put ${\bf f}_i(X)$ in the following form:

\begin{equation} \label{eq:linear_condition3}
{\bf f}_i(X) = {\bf f}_{i,0}(X^t) + X {\bf f}_{i,1}(X^t) + \ldots + X^{t-1} {\bf f}_{i,t-1} (X^t)
\end{equation}

\noindent where for $0 \le j < t$, ${\bf f}_{i,j}(X^t)$ is a polynomial in $X^t$. Thus, Equation (\ref{eq:linear_condition2}) is equivalent to the system of equations

\begin{equation} \label{eq:linear_condition4}
\ba_0 (X^t) {\bf f}_{0,j}(X^t) + \ba_1(X^t) {\bf f}_{1,j}(X^t) + \ldots + \ba_{r-1}(X^t) {\bf f}_{r-1, j}(X^t) = 0,
\end{equation}

\noindent for $0 \le j < t$. For simplification, we can replace $X^t$ by $X$ to write this system of equations as

\begin{equation} \label{eq:linear_condition5}
\ba_0 (X) {\bf f}_{0,j}(X) + \ba_1(X) {\bf f}_{1,j}(X) + \ldots + \ba_{r-1}(X) {\bf f}_{r-1, j}(X) = 0,
\end{equation}

\noindent for $0 \le j < t$. Then the system of equations given by (\ref{eq:linear_condition5}) reduces to 

\begin{equation} \label{eq:linear_condition6}
(\ba_0(X), \ba_1(X), \ldots, \ba_{r-1}(X)) \cdot \left[ \begin{array}{cccc}
{\bf f}_{0,0}(X) & {\bf f}_{0,1}(X) & \cdots & {\bf f}_{0,t-1}(X) \\
{\bf f}_{1,0}(X) & {\bf f}_{1,1}(X) & \cdots & {\bf f}_{1,t-1}(X) \\
\vdots & \vdots & \vdots & \vdots \\
{\bf f}_{r-1,0}(X) & {\bf f}_{r-1,1}(X) & \cdots & {\bf f}_{r-1,t-1}(X) \\
\end{array} \right] = 0.
\end{equation}

Thus, the rows of the $(b+1)r \times n$ matrix $\bG_{\text{local}}(\mathbf{g}_0, \mathbf{g}_1, \ldots, \mathbf{g}_{r-1})(X)$ are linearly independent if and only if the following holds: If the above matrix equation holds, where $\ba_0(X), \ba_1(X), \ldots, \ba_{r-1}(X)$ are polynomials of degree at most $b$, then $\ba_0(X) = \ba_1(X) = \cdots = \ba_{r-1}(X) = 0$.

By dropping the condition that $\ba_0(X), \ba_1(X), \ldots, \ba_{r-1}(X)$ are polynomials of degree at most $b$ and allowing them to be Laurent series~\cite{piret1988}, i.e., of the form 
\[\sum_{i=r}^{\infty} a_i X^i
\]
where $r$ is an integer (negative or nonnegative), we can get a sufficient condition for linear independence. Since Laurent series is a field, the system of equations holds over this field if and only if the $r \times t$ matrix
\begin{equation} \label{eq:linear_condition7}
\bF(X) = \left[ \begin{array}{cccc}
{\bf f}_{0,0}(X) & {\bf f}_{0,1}(X) & \cdots & {\bf f}_{0,t-1}(X) \\
{\bf f}_{1,0}(X) & {\bf f}_{1,1}(X) & \cdots & {\bf f}_{1,t-1}(X) \\
\vdots & \vdots & \vdots & \vdots \\
{\bf f}_{r-1,0}(X) & {\bf f}_{r-1,1}(X) & \cdots & {\bf f}_{r-1,t-1}(X) \\
\end{array} \right]
\end{equation}
has full rank. The condition that the matrix $\bF(X)$ has full rank over the field of Laurent series is \hl{equivalent to having an $r \times r$ submatrix of nonzero determinant}.

For $0 \le i < r$, the $i$-row $({\bf f}_{i,0}(X), {\bf f}_{i,1}(X), \cdots, {\bf f}_{i,t-1}(X))$ of $\bF(X)$ is called the \hl{$t$-fold decomposition} of the generator multiplier ${\bf f}_i(X$) of the $i$-th code $C_i$ in the code chain $C_0 \supset C_1  \supset \ldots \supset C_{r-1}$. Since ${\bf f}_0(X) = $1, the $t$ entries in the first row of $\bF(X)$ are ${\bf f}_{0,0}(X) = 1$ and ${\bf f}_{0,1}(X ) =  \ldots  = {\bf f}_{0,t-1}(X) = 0$.   We call $\bF(X)$ the \hl{$t$-fold generator multiplier decomposition} (GMD) matrix ($t$-fold GMD-matrix). The polynomials $\mathbf{f}_{i,0}(X), \mathbf{f}_{i,1}(X), \ldots, \mathbf{f}_{i,t-1}(X)$ are called the $t$-fold decomposed polynomials of the generator multiplier $\mathbf{f}_i(X)$.

Summarizing the above developments, we have Theorem \ref{Thm2}.

\begin{theorem}\label{Thm2}
Let $\mathcal{C}_0 \supset \mathcal{C}_1 \supset \cdots \supset \mathcal{C}_{r-1}$ be a code chain of length $r$ in which the codes $\mathcal{C}_0, \mathcal{C}_1, \ldots, \mathcal{C}_{r-1}$ are generated by the polynomials $\bg_0(X), \bg_1(X), \ldots, \bg_{r-1}(X)$ with $\bg_i(X) = \mathbf{f}_i(X)\bg_0(X)$ for $0 \leq i < r$ and $\mathbf{f}_0(X) = 1$. The $(b+1)r$ rows of the $(b+1)r \times n$ generator matrix $\bG_{\text{local}}(\mathbf{g}_0, \mathbf{g}_1, \ldots, \mathbf{g}_{r-1})$ of the $(r, t)$-PCS-local subcode $\mathcal{C}_{\text{local}}$ of the mother code $\mathcal{C}_0$ in the code chain are linearly independent if and only if the $r \times t$ $t$-fold GMD-matrix $\bF(X)$ formed by the $t$-fold decompositions of the $r$ generator multipliers $\mathbf{f}_0(X) = 1, \mathbf{f}_1(X), \ldots, \mathbf{f}_{r-1}(X)$ for the code chain has an $r \times r$ submatrix of nonzero determinant, i.e., $\bF(X)$ is a full rank matrix. 
\end{theorem}  

\begin{corollary}\label{cor1}
If the $r \times t$ $t$-fold GMD-matrix $\bF(X)$ formed by the $t$-fold decompositions of the $r$ generator multipliers $\mathbf{f}_0(X) = 1, \mathbf{f}_1(X), \ldots, \mathbf{f}_{r-1}(X)$ for a code chain $\mathcal{C}_0 \supset \mathcal{C}_1 \supset \cdots \supset \mathcal{C}_{r-1}$ is a full rank matrix, the $(r, t)$-PCS-local subcode $\mathcal{C}_{\text{local}}$ of the mother code $\mathcal{C}_0$ in the code chain is the direct-sum of its $r$ constituent $(1, t)$-PCS-local subcodes. 
\end{corollary}

Therefore, in the design of a code chain $\mathcal{C}_0 \supset \mathcal{C}_1 \supset \cdots \supset \mathcal{C}_{r-1}$ of length $r$ for which the $(r, t)$-PCS-local subcode $\mathcal{C}_{\text{local}}$ of the mother code $\mathcal{C}_0$ in the code chain is the direct-sum of its constituent $(1, t)$-PCS-local subcodes, we must first choose the generator multipliers $\mathbf{f}_0(X) = 1, \mathbf{f}_1(X), \ldots, \mathbf{f}_{r-1}(X)$ for the code chain to satisfy the ratio condition given in Theorem~\ref{Thm1} and then check whether the $t$-fold GMD-matrix $\bF(X)$ formed by the $t$-fold decompositions of the generator multipliers satisfies the necessary and sufficient condition given in Theorem~\ref{Thm2}, i.e., to have an $r \times r$ submatrix of nonzero determinant. A code chain for which its $t$-fold GMD-matrix $\bF(X)$ is a full rank matrix is called a \hl{$t$-fold full rank code chain}.

Let $\mathcal{C}_0 \supset \mathcal{C}_1 \supset \cdots \supset \mathcal{C}_{r-1}$ be a $t$-fold full rank code chain of length $r$ constructed based on the mother code $\mathcal{C}_0$ and the generator multipliers $\mathbf{f}_0(X) = 1, \mathbf{f}_1(X), \ldots, \mathbf{f}_{r-1}(X)$ for the $r$ codes in the code chain. Let $s$ be a positive integer with $1 \leq s < r$. A subchain of the code chain $\mathcal{C}_0 \supset \mathcal{C}_1 \supset \cdots \supset \mathcal{C}_{r-1}$ that consists of $s$ codes in the code chain is also a $t$-fold full rank code chain. It can be used to construct the $(s, t)$-PCS-local subcode of the mother code in the subchain. Let $\mathcal{C}_i$ be the mother code of the subchain and $\mathbf{f}_i(X)$ be its generator multiplier associated with the code chain $\mathcal{C}_0 \supset \mathcal{C}_1 \supset \cdots \supset \mathcal{C}_{r-1}$. Then, in the construction of the $(s, t)$-PCS-local subcode $\mathcal{C}^{(s)}_{\text{local}}$ of the mother code $\mathcal{C}_i$ in the subchain, the generator multiplier $\mathbf{f}^*_j(X)$ for a descendant code $\mathcal{C}_j$ of $\mathcal{C}_i$ in the subchain is adjusted to $\mathbf{f}_j(X)/\mathbf{f}_i(X)$. For the mother code $\mathcal{C}_i$, the generator multiplier is $\mathbf{f}^*_i(X) = \mathbf{f}_i(X)/\mathbf{f}_i(X) = 1$. Clearly, the $(s, t)$-PCS-local subcode $\mathcal{C}^{(s)}_{\text{local}}$ of $\mathcal{C}_i$ is a subcode of the $(r, t)$-PCS-local subcode $\mathcal{C}_{\text{local}}$ of the mother code $\mathcal{C}_0$ in the full code chain $\mathcal{C}_0 \supset \mathcal{C}_1 \supset \cdots \supset \mathcal{C}_{r-1}$.

A special case is that a code chain consists of a single $(n, k_0)$ code $\mathcal{C}_0$ with minimum distance $d_0$. In this case, the $(1, t)$-PCS-local code $\mathcal{C}_{\text{local}}$ of $\mathcal{C}_0$ with shifting span $b = \left\lfloor \frac{k_0 - 1}{t} \right\rfloor$ is an $(n, \left\lfloor \frac{k_0 - 1}{t} \right\rfloor + 1)$ code with minimum distance at least $d_0$. For $t = 2$, if $k_0$ is an odd integer, $\mathcal{C}_{\text{local}}$ is an $(n, (k_0 + 1)/2)$ code.

Full rank code chains will be used to construct convolutional codes in the next sections based on various classes of cyclic codes. The convolutional code constructed based on a $t$-fold full rank code chain of length $r$ is a rate $r/t$ convolutional code of overall constraint length $\nu$ upper bounded by
\[
\left\lceil \frac{n - k_0}{t} \right\rceil + \left\lceil \frac{n - k_1}{t} \right\rceil + \cdots + \left\lceil \frac{n - k_{r-1}}{t} \right\rceil
\]
with local minimum distance lower bounded by the minimum distance $d_0$ of the mother code in the code chain. It is the direct-sum of $r$ rate-$1/t$ constituent convolutional codes, each constructed based on a code in the code chain, with memory orders upper bounded by $\left\lceil \frac{n - k_0}{t} \right\rceil, \ldots, \left\lceil \frac{n - k_{r-1}}{t} \right\rceil$, respectively.

\subsection*{Example 1}
\begin{example} \label{eg1}
In this example, we design a $4$-fold full rank code chain $\mathcal{C}_0 \supset \mathcal{C}_1 \supset \mathcal{C}_2$ of length $r=3$ in which the $3$ codes are cyclic codes of length $63$. The code chain is designed to construct a $(3,4)$-PCS-local subcode of the mother code $\mathcal{C}_0$ in the code chain with shifting factor $t=4$. The desired minimum distance of the $(3,4)$-PCS-local subcode of $\mathcal{C}_0$ is at least $11$.

In the design, we choose the $(63,36)$ BCH code of length $n=63$, dimension $k_0=36$ and minimum distance $d_0=11$ as the mother code $\mathcal{C}_0$ in the code chain. The $(63,36)$ BCH mother code $\mathcal{C}_0$ is constructed based on the Galois field $\mathrm{GF}(2^6)$. Let $\alpha$ be a primitive element in $\mathrm{GF}(2^6)$. The generator polynomial $\bg_0(X)$ of the $(63,36)$ BCH mother code $\mathcal{C}_0$ has $\alpha, \alpha^2, \ldots, \alpha^{10}$ and their conjugates as roots and it is a polynomial of degree $27$~\cite{lin2004}:
\[
\bg_0(X) = 1 + X + X^4 + X^8 + X^{15} + X^{17} + X^{18} + X^{19} + X^{21} + X^{22} + X^{27}.
\]

To construct the two descendant codes $\mathcal{C}_1$ and $\mathcal{C}_2$ of the mother code $\mathcal{C}_0$, we choose $\mathbf{f}_1(X) = 1 + X$ and $\mathbf{f}_2(X) = \mathbf{f}_1(X)(1 + X + X^2) = 1 + X^3$ as the generator multipliers for $\mathcal{C}_1$ and $\mathcal{C}_2$, respectively, where $1 + X$ and $(1 + X + X^2)$ are the minimal polynomials of the element $\alpha^0 = 1$ and $\alpha^{21}$ in $\mathrm{GF}(2^6)$, which are not roots of $\bg_0(X)$. Then, the generator polynomials of $\mathcal{C}_1$ and $\mathcal{C}_2$ are:
\begin{eqnarray*}
\bg_1(X) &=& {\bf f}_1(X) \bg_0(X)  = ( 1 + X) \bg_0(X)\\
 &=& 1 + X^2 + X^4 + X^5 + X^8 + X^9 + X^{15} + X^{16} + X^{17} + X^{20} + X^{21} + X^{23} + X^{27} + X^{28},
\end{eqnarray*}
and
\begin{eqnarray*}
\bg_2(X) &=& {\bf f}_2(X) \bg_0(X)  = ( 1 + X) (1 + X + X^2) \bg_0(X)\\
 &=& 1 + X + X^3 + X^7 + X^8 + X^{11} + X^{15} + X^{17}   + X^{19} + X^{20} + X^{24} + X^{25} + X^{27} + X^{30}.
\end{eqnarray*}

The two codes $\mathcal{C}_1$ and $\mathcal{C}_2$ generated by $\bg_1(X)$ and $\bg_2(X)$ are $(63,35)$ and $(63,33)$ codes with dimensions $k_1=35$ and $k_2=33$, respectively. $\mathcal{C}_1$ is the even weight subcode of the $(63,36)$ BCH mother code $\mathcal{C}_0$ in the code chain with minimum distance $12$. Since $\mathcal{C}_2$ is a subcode of $\mathcal{C}_1$, its minimum distance is at least $12$.

With $\mathbf{f}_0(X) = 1$ as the generator multiplier for $\mathcal{C}_0$, the ratios of the chosen generator multipliers are $\mathbf{f}_1(X)/\mathbf{f}_0(X) = 1 + X$, $\mathbf{f}_2(X)/\mathbf{f}_0(X) = 1 + X^3$, and $\mathbf{f}_2(X)/\mathbf{f}_1(X) = 1 + X + X^2$, respectively. We see that none of these ratios is a polynomial in $X^4$. Hence, $\mathbf{f}_0(X), \mathbf{f}_1(X)$ and $\mathbf{f}_2(X)$ satisfy the ratio condition given in Theorem~\ref{Thm1}.

Next, we show that the chosen generator multipliers $\mathbf{f}_0(X), \mathbf{f}_1(X)$ and $\mathbf{f}_2(X)$ for the code chain satisfy the full rank code chain condition given by Theorem~\ref{Thm2}. To do that, we decompose the generator multipliers $\mathbf{f}_0(X), \mathbf{f}_1(X)$ and $\mathbf{f}_2(X)$ in $4$-fold as follows:
\begin{eqnarray*}
{\bf f}_0(X) & \longleftrightarrow & {\bf f}_{0,0}(X) = 1, {\bf f}_{0,1}(X) = 0, {\bf f}_{0,2}(X) = 0, {\bf f}_{0,3}(X) = 0,\\
{\bf f}_1(X) & \longleftrightarrow & {\bf f}_{1,0}(X) = 1, {\bf f}_{1,1}(X) = 1, {\bf f}_{1,2}(X) = 0, {\bf f}_{1,3}(X) = 0,\\
{\bf f}_2(X) & \longleftrightarrow & {\bf f}_{2,0}(X) = 1, {\bf f}_{2,1}(X) = 0, {\bf f}_{2,2}(X) = 0, {\bf f}_{2,3}(X) = 1.
\end{eqnarray*}
Using the $4$-fold decompositions, we form the following $3 \times 4$ $4$-fold GMD-matrix:
\[
\bF(X) = \begin{bmatrix}
1 & 0 & 0 & 0 \\
1 & 1 & 0 & 0 \\
1 & 0 & 0 & 1
\end{bmatrix}.
\]
We see that the $0$-th, $1$-st and $3$-rd columns of $\bF(X)$ form a $3 \times 3$ submatrix with nonzero determinant. Hence, $\bF(X)$ is a full rank matrix and the code chain $\mathcal{C}_0 \supset \mathcal{C}_1 \supset \mathcal{C}_2$ is a $4$-fold full rank code chain.

Dividing $k_2 - 1 = 32$ by $4$, we have quotient $b = 8$ and remainder $c = 0$. Using the generator polynomials $\bg_0(X)$, $\bg_1(X)$, and $\bg_2(X)$, shifting factor $t = 4$, and shifting span $b = 8$, we construct the $27 \times 63$ generator matrix $\bG_{\text{local}}(\bg_0, \bg_1, \bg_2)$ of the $(3,4)$-PCS-local subcode $\mathcal{C}_{\text{local}}$ of the mother code $\mathcal{C}_0$ in the code chain $\mathcal{C}_0 \supset \mathcal{C}_1 \supset \mathcal{C}_2$. Since the code chain is a full rank code chain, the $27$ rows of $\bG_{\text{local}}(\bg_0, \bg_1, \bg_2)$ are linearly independent. Hence, the $(3,4)$-PCS-local subcode $\mathcal{C}_{\text{local}}$ of the mother code $\mathcal{C}_0$ is a $(63,27)$ code with minimum distance at least $11$ which is the direct-sum of its $3$ constituent $(1,4)$-PCS-local subcodes $\mathcal{C}_{0,\text{local}}, \mathcal{C}_{1,\text{local}}, \mathcal{C}_{2,\text{local}}$ of the $3$ constituent codes $\mathcal{C}_0, \mathcal{C}_1, \mathcal{C}_2$ in the code chain, each a $(63,9)$ code.

To construct $\bG_{\text{local}}(\bg_0, \bg_1, \bg_2)$, we first form the $3 \times 31$ belt generator matrix of $\mathcal{C}_{\text{local}}$ using the generator sequences of the $3$ codes in the code chain:
{\small
\begin{equation*}
\begin{array}{l}
\bB(\bg_0, \bg_1, \bg_2) = \\
\begin{bmatrix}
1 & 1 & 0 & 0 & 1 & 0 & 0 & 0 & 1 & 0 & 0 & 0 & 0 & 0 & 0 & 1 & 0 & 1 & 1 & 1 & 0 & 1 & 1 & 0 & 0 & 0 & 0 & 1 & 0 & 0 & 0 \\
1 & 0 & 1 & 0 & 1 & 1 & 0 & 0 & 1 & 1 & 0 & 0 & 0 & 0 & 0 & 1 & 1 & 1 & 0 & 0 & 1 & 1 & 0 & 1 & 0 & 0 & 0 & 1 & 1 & 0 & 0 \\
1 & 1 & 0 & 1 & 0 & 0 & 0 & 1 & 1 & 0 & 0 & 1 & 0 & 0 & 0 & 1 & 0 & 1 & 0 & 1 & 1 & 0 & 0 & 0 & 1 & 1 & 0 & 1 & 0 & 0 & 1 \\
\end{bmatrix}
\end{array}
\end{equation*}
}
Appending $b = 8$ zero matrices $O$ of size $3 \times 4$ to $\bB(g_0, g_1, g_2)$, we form the $3 \times 63$ bend matrix of $\mathcal{C}_{\text{local}}$:
\[
\bM(\bg_0, \bg_1, \bg_2) = \left[ \bB(\bg_0, \bg_1, \bg_2), \bO,  \bO, \bO, \bO, \bO, \bO, \bO, \bO \right].
\]
Cyclically shifting the rows of $\bM(\bg_0, \bg_1, \bg_2)$ together to the right $8$ times, $4$ positions at a time, we obtain the $27 \times 63$ generator matrix $\bG_{\text{local}}(\bg_0, \bg_1, \bg_2)$ of the $(63,27)$ $(3,4)$-PCS-local subcode $\mathcal{C}_{\text{local}}$ of the mother code $\mathcal{C}_0$ in the code chain with minimum distance $d_{\text{local}}$ at least $11$.

If we use the subchain $\mathcal{C}_0 \supset \mathcal{C}_2$ of the $4$-fold full rank code chain $\mathcal{C}_0 \supset \mathcal{C}_1 \supset \mathcal{C}_2$, we can construct the $(2,4)$-PCS-local subcode of the mother code $\mathcal{C}_0$ which is a $(63,18)$ code with local minimum distance at least $11$. If we use the subchain $\mathcal{C}_1 \supset \mathcal{C}_2$ of the $4$-fold full rank code chain $\mathcal{C}_0 \supset \mathcal{C}_1 \supset \mathcal{C}_2$ with generator multipliers $\mathbf{f}^*_1(X) = 1$ and $\mathbf{f}^*_2(X) = \mathbf{f}_2(X)/\mathbf{f}_1(X) = (1 + X^3)/(1 + X) = 1 + X + X^2$, we can construct the $(2,4)$-PCS-local subcode of the mother code $\mathcal{C}_1$ in the subchain $\mathcal{C}_1 \supset \mathcal{C}_2$ which is a $(63,18)$ code with local minimum distance at least $12$. If we use the $(63,33)$ code $\mathcal{C}_2$ alone, we can construct a $(63,9)$ $(1,4)$-PCS-local subcode of $\mathcal{C}_2$ (also a subcode of $\mathcal{C}_0$) with minimum distance at least $12$.
\end{example}

\subsection*{Example 2}
\begin{example} \label{eg2}
In this example, we design a $5$-fold full rank code chain $\mathcal{C}_0 \supset \mathcal{C}_1 \supset \mathcal{C}_2 \supset \mathcal{C}_3$ of length $r=4$ with shifting factor $t=5$ to construct a $(4,5)$-PCS-local subcode of the mother code $\mathcal{C}_0$ in the code chain with a relatively large local minimum distance.

In the design, we choose the $(255,115)$ BCH code of length $255$ with dimension $k_0=115$ and designed minimum distance $d_0=43$ as the mother code in the code chain $\mathcal{C}_0 \supset \mathcal{C}_1 \supset \mathcal{C}_2 \supset \mathcal{C}_3$. The BCH code is constructed based on the Galois field $\mathrm{GF}(2^8)$. Let $\alpha$ be a primitive element of $\mathrm{GF}(2^8)$. The generator polynomial $\bg_0(X)$ of $\mathcal{C}_0$ is a polynomial of degree $140$ with $\alpha, \alpha^2, \ldots, \alpha^{42}$ and their conjugates as roots.

To construct the generator polynomials of the three descendant codes $\mathcal{C}_1, \mathcal{C}_2$ and $\mathcal{C}_3$ of $\mathcal{C}_0$ in the code chain, we choose the following $3$ generator multipliers:
\begin{align*}
\mathbf{f}_1(X) &= 1 + X^6 + X^7 + X^8, \\
\mathbf{f}_2(X) &= \mathbf{f}_1(X)(1 + X + X^2) = 1 + X + X^2 + X^6 + X^8 + X^{10}, \\
\mathbf{f}_3(X) &= \mathbf{f}_2(X)(1 + X^3 + X^4) \\
       &= 1 + X + X^2 + X^3 + X^8 + X^9 + X^{11} + X^{13} + X^{14},
\end{align*}
where $1 + X^6 + X^7 + X^8$, $1 + X + X^2$ and $1 + X^3 + X^4$ are the minimal polynomials of the elements $\alpha^{43}$, $\alpha^{85}$ and $\alpha^{119}$ in $\mathrm{GF}(2^8)$, respectively, which are not roots in $\bg_0(X)$. The generator polynomials of the $3$ descendant codes $\mathcal{C}_1, \mathcal{C}_2$ and $\mathcal{C}_3$ of $\mathcal{C}_0$ are
\[
\bg_1(X) = \bg_0(X)\mathbf{f}_1(X), \quad
\bg_2(X) = \bg_0(X)\mathbf{f}_2(X), \quad
\bg_3(X) = \bg_0(X)\mathbf{f}_3(X).
\]
The $4$ codes $\mathcal{C}_0, \mathcal{C}_1, \mathcal{C}_2$ and $\mathcal{C}_3$ generated by $\bg_0(X), \bg_1(X), \bg_2(X)$ and $\bg_3(X)$ form a code chain $\mathcal{C}_0 \supset \mathcal{C}_1 \supset \mathcal{C}_2 \supset \mathcal{C}_3$ of length $4$. Next, we show that the code chain is a full rank code chain.

From the generator multipliers $\mathbf{f}_0(X) = 1, \mathbf{f}_1(X), \mathbf{f}_2(X), \mathbf{f}_3(X)$ given above, we easily check that for $0 \leq i < j < t = 5$, the ratio $\mathbf{f}_j(X)/\mathbf{f}_i(X)$ is not a polynomial in $X^5$. The $5$-fold decompositions of $\mathbf{f}_0(X) = 1, \mathbf{f}_1(X), \mathbf{f}_2(X)$ and $\mathbf{f}_3(X)$ are:
\begin{eqnarray*}
{\bf f}_0(X) & \longleftrightarrow & {\bf f}_{0,0}(X) = 1, {\bf f}_{0,1}(X) = 0, {\bf f}_{0,2}(X) = 0, {\bf f}_{0,3}(X) = 0, {\bf f}_{0,4}(X) = 0, \\
{\bf f}_1(X) & \longleftrightarrow & {\bf f}_{1,0}(X) = 1, {\bf f}_{1,1}(X) = X, {\bf f}_{1,2}(X) = X, {\bf f}_{1,3}(X) = X, {\bf f}_{1,4}(X) = 0, \\
{\bf f}_2(X) & \longleftrightarrow & {\bf f}_{2,0}(X) = 1 + X^2, {\bf f}_{2,1}(X) = 1 + X, {\bf f}_{2,2}(X) = 1, {\bf f}_{2,3}(X) = X, {\bf f}_{2,4}(X) = 0, \\
{\bf f}_3(X) & \longleftrightarrow & {\bf f}_{3,0}(X) = 1, {\bf f}_{3,1}(X) = 1 + X^2, {\bf f}_{3,2}(X) = 1, {\bf f}_{3,3}(X) = 1 + X + X^2, {\bf f}_{3,4}(X) = X + X^2. \\
\end{eqnarray*}

The $5$-fold GMD-matrix $\bF(X)$ of the code chain formed by the $5$-fold decompositions of the generator multipliers $\mathbf{f}_0(X), \mathbf{f}_1(X), \mathbf{f}_2(X), \mathbf{f}_3(X)$ is a $4 \times 5$ matrix,
\[
\bF(X) = \begin{bmatrix}
1 & 0 & 0 & 0 & 0 \\
1 & X & X & X & 0 \\
1 + X^2 & 1 + X & 1 & X & 0 \\
1 & 1 + X^2 & 1 & 1 + X + X^2 & X + X^2
\end{bmatrix}.
\]
The first $4$ columns of $\bF(X)$ form a $4 \times 4$ submatrix with nonzero determinant. Hence, $\bF(X)$ is a full rank matrix and the code chain $\mathcal{C}_0 \supset \mathcal{C}_1 \supset \mathcal{C}_2 \supset \mathcal{C}_3$ constructed by using the $(255,115)$ BCH mother code and the generator multipliers $\mathbf{f}_0(X), \mathbf{f}_1(X), \mathbf{f}_2(X), \mathbf{f}_3(X)$ is a $5$-fold full rank code chain.

The three descendant codes $\mathcal{C}_1, \mathcal{C}_2$ and $\mathcal{C}_3$ of $\mathcal{C}_0$ generated by $\bg_1(X) = \bg_0(X)\mathbf{f}_1(X)$, $\bg_2(X) = \bg_0(X)\mathbf{f}_2(X)$ and $\bg_3(X) = \bg_0(X)\mathbf{f}_3(X)$ are $(255,107)$, $(255,105)$ and $(255,101)$ codes with dimensions $k_1 = 107$, $k_2 = 105$ and $k_3 = 101$, respectively. The $(255,107)$ code $\mathcal{C}_1$ is also a BCH code with designed minimum distance $45$. Since $\mathcal{C}_2$ and $\mathcal{C}_3$ are subcodes of $\mathcal{C}_1$, their minimum distances are at least $45$.

Dividing $k_3 - 1 = 100$ by $5$, the quotient $b$ of the division is $20$. Using the generator polynomials of $4$ codes in the code chain $\mathcal{C}_0 \supset \mathcal{C}_1 \supset \mathcal{C}_2 \supset \mathcal{C}_3$ and setting the shifting factor and span to $5$ and $20$, respectively, we construct the generator matrix $\bG_{\text{local}}(\bg_0, \bg_1, \bg_2, \bg_3)$ of the $(4,5)$-PCS-local subcode $\mathcal{C}_{\text{local}}$ of the mother code $\mathcal{C}_0$ in the code chain which is an $84 \times 255$ matrix. Since the code chain $\mathcal{C}_0 \supset \mathcal{C}_1 \supset \mathcal{C}_2 \supset \mathcal{C}_3$ is a full rank code chain, the $84$ rows of $\bG_{\text{local}}(\bg_0, \bg_1, \bg_2, \bg_3)$ are linearly independent. The row space of $\bG_{\text{local}}(\bg_0, \bg_1, \bg_2, \bg_3)$ gives the $(4,5)$-PCS-local subcode $\mathcal{C}_{\text{local}}$ of $\mathcal{C}_0$ which is a $(255,84)$ code with local minimum distance lower bounded by the designed minimum distance $43$ of the mother code $\mathcal{C}_0$ in the code chain. The code $\mathcal{C}_{\text{local}}$ is the direct-sum of its $4$ constituent $(1,5)$-PCS-local subcodes $\mathcal{C}_{0,\text{local}}, \mathcal{C}_{1,\text{local}}, \mathcal{C}_{2,\text{local}}, \mathcal{C}_{3,\text{local}}$ of the $4$ codes in the code chain, each is a $(255,21)$ code.

Subchains of the $5$-fold code chain $\mathcal{C}_0 \supset \mathcal{C}_1 \supset \mathcal{C}_2 \supset \mathcal{C}_3$ can be used to construct localized subcodes of the mother code in the code chain. Suppose we use the subchain $\mathcal{C}_0 \supset \mathcal{C}_2 \supset \mathcal{C}_3$ of the code chain $\mathcal{C}_0 \supset \mathcal{C}_1 \supset \mathcal{C}_2 \supset \mathcal{C}_3$ for code construction. Then, the $(3,5)$-PCS-local subcode of the mother code $\mathcal{C}_0$ is a $(255,63)$ code with local minimum distance at least $43$.

Suppose we use the subchain $\mathcal{C}_1 \supset \mathcal{C}_2 \supset \mathcal{C}_3$ of the code chain $\mathcal{C}_0 \supset \mathcal{C}_1 \supset \mathcal{C}_2 \supset \mathcal{C}_3$ for code construction. Then, the mother code of the subchain is the $(255,107)$ BCH code $\mathcal{C}_1$. Set the generator multipliers for the subchain to $\mathbf{f}^*_1(X) = 1$, $\mathbf{f}^*_2(X) = \mathbf{f}_2(X)$ and $\mathbf{f}^*_3(X) = \mathbf{f}_3(X)$. In this case, the $5$-fold GMD-matrix $\mathbf{f}_3(X)$ of the subchain $\mathcal{C}_1 \supset \mathcal{C}_2 \supset \mathcal{C}_3$ formed by the $5$-fold decompositions of the generator multipliers $\mathbf{f}^*_1(X), \mathbf{f}^*_2(X), \mathbf{f}^*_3(X)$ is a $3 \times 5$ matrix,
\[
\bF(X) = \begin{bmatrix}
1 & 0 & 0 & 0 & 0 \\
1 + X^2 & 1 + X & 1 & X & 0 \\
1 & 1 + X^2 & 1 & 1 + X + X^2 & X + X^2
\end{bmatrix},
\]
which is a full rank matrix. The $(3,5)$-PCS-local subcode of $\mathcal{C}_1$ is a $(255,63)$ code with local minimum distance at least the minimum distance $45$ of the $(255,107)$ BCH mother code $\mathcal{C}_1$ in the subchain $\mathcal{C}_1 \supset \mathcal{C}_2 \supset \mathcal{C}_3$.

Suppose we reset the shifting factor $t$ to $2$ and the shifting span $b = (115 - 1)/2 = 57$. Using $\mathcal{C}_0$ alone, we can construct a $(1,2)$-PCS-local subcode of $\mathcal{C}_0$ which is a $(255,58)$ code with dimension $56$ and local minimum distance at least $43$.
\end{example}

\section{Construction of Convolutional Codes Based on Full Rank Cyclic Code Chains}\label{sect4:convol_from_cyclic_code}

In this section, we show that using a $t$-fold full rank code chain $\mathcal{C}_0 \supset \mathcal{C}_1 \supset \cdots \supset \mathcal{C}_{r-1}$ of length $r$ with $r < t$, a rate-$r/t$ convolutional code can be constructed based on the $t$-fold decompositions of the generator polynomials $\bg_0(X), \bg_1(X), \ldots, \bg_{r-1}(X)$ of the $r$ constituent codes in the code chain. The rate-$r/t$ convolutional code, denoted by $\mathcal{C}_{\text{convol}}(\bg_0, \bg_1, \ldots, \bg_{r-1})$, consists of an infinite number of copies of an  identical local code. Each local code of $\mathcal{C}_{\text{convol}}(\bg_0, \bg_1, \ldots, \bg_{r-1})$ is the $(r, t)$-PCS-local subcode $\mathcal{C}_{\text{local}}$ of the mother code $\mathcal{C}_0$ in the code chain. The local minimum distance of $\mathcal{C}_{\text{convol}}(\bg_0, \bg_1, \ldots, \bg_{r-1})$ is lower bounded by the minimum distance $d_0$ of the mother code $\mathcal{C}_0$ in the code chain and the overall constraint length $\nu$ of the convolutional code $\mathcal{C}_{\text{convol}}(\bg_0, \bg_1, \ldots, \bg_{r-1})$ is upper bounded by
\[
\left\lceil \frac{n - k_0}{t} \right\rceil + \left\lceil \frac{n - k_1}{t} \right\rceil + \cdots + \left\lceil \frac{n - k_{r-1}}{t} \right\rceil.
\]
The convolutional code $\mathcal{C}_{\text{convol}}(\bg_0, \bg_1, \ldots, \bg_{r-1})$ has a multi-layer structure and is the direct-sum of $r$ rate-$1/t$ convolutional codes which are constructed based on the $r$ cyclic codes in the code chain $\mathcal{C}_0 \supset \mathcal{C}_1 \supset \cdots \supset \mathcal{C}_{r-1}$ separately. The convolutional code $\mathcal{C}_{\text{convol}}(\bg_0, \bg_1, \ldots, \bg_{r-1})$ inherits the structure, algebraic or geometric, of the mother code in the code chain.

\subsection{Code Construction in Time Domain} \label{sect4a:code_construct_in_time_domain}

The code construction starts with a $t$-fold full rank code chain $\mathcal{C}_0 \supset \mathcal{C}_1 \supset \cdots \supset \mathcal{C}_{r-1}$. Using the generator polynomials of the code chain, we form the belt generator matrix $\bB(\bg_0, \bg_1, \ldots, \bg_{r-1})$ of width $r$ in the form of (\ref{eq:B_matrix}), comprised of the generator sequences of the $r$ cyclic codes in the code chain. Next, we form the following semi-infinite matrix by shifting the belt generator matrix $\bB(\bg_0, \bg_1, \ldots, \bg_{r-1})$ to the right and down, $t$ positions for each right shift and $r$ positions for each down shift:

\begin{equation}\label{mat:G_convol}  
\bG_{\convol}(\bg_0, \bg_1, \ldots , \bg_{r-1}) = \left[ \begin{array}{l}
\bB(\bg_0, \bg_1, \ldots , \bg_{r-1}) \\
~~~~~~~~\bB(\bg_0, \bg_1, \ldots , \bg_{r-1})\\
~~~~~~~~~~~~~~~~ \vdots \\
~~~~~~~~~~~~~~~~~~~~~~~~\bB(\bg_0, \bg_1, \ldots , \bg_{r-1})\\
~~~~~~~~~~~~~~~~~~~~~~~~~~~~~~~~\vdots \\
\end{array} \right]
\end{equation}

\noindent where the blank areas are all zeros. The matrix $\bG_{\text{convol}}(\bg_0, \bg_1, \ldots, \bg_{r-1})$ is composed of an infinite number of copies of row-bends. Each row-bend contains a copy of the belt generator matrix $\bB(\bg_0, \bg_1, \ldots, \bg_{r-1})$. For $\ell \geq 0$, the belt generator matrix $\bB(\bg_0, \bg_1, \ldots, \bg_{r-1})$ in the $\ell$-th row-bend of $\bG_{\text{convol}}(\bg_0, \bg_1, \ldots, \bg_{r-1})$ begins at the location $\ell t$ and covers a span of $n - b t$ positions. All copies of the belt generator matrix $\bB(\bg_0, \bg_1, \ldots, \bg_{r-1})$ are confined to a diagonal band of width $n - b t = n - k_{r-1} + 1 \cred{ + c}$.

The columns of $\bB(\bg_0, \bg_1, \ldots, \bg_{r-1})$ can be divided into $\lceil \frac{n - k_{r-1} + 1 \cred{ + c}}{t} \rceil$ blocks. If $n - k_{r-1} + 1 \cred{ + c}$ is divisible by $t$, each block consists of $t$ consecutive columns of $\bB(\bg_0, \bg_1, \ldots, \bg_{r-1})$. If $n - k_{r-1} + 1 \cred{ + c}$ is not divisible by $t$, each of the first $\lceil \frac{n - k_{r-1} + 1 \cred{ + c}}{t} \rceil - 1$ blocks consists of $t$ consecutive columns, and the last block consists of $\tau$ columns, where $\tau$ is the remainder resulting from dividing $n - k_{r-1} + 1 \cred{ + c}$ by $t$.

From (\ref{mat:G_convol}), the structure of the belt generator matrix $\bB(\bg_0, \bg_1, \ldots, \bg_{r-1})$ given by (\ref{eq:B_matrix}), and the full rank structure of the code chain $\mathcal{C}_0 \supset \mathcal{C}_1 \supset \cdots \supset \mathcal{C}_{r-1}$, we see that the rows of $\bG_{\text{convol}}(\bg_0, \bg_1, \ldots, \bg_{r-1})$ are linearly independent. From (\ref{eq:M_matrix}), (\ref{eq:B_matrix}), (\ref{eq:G_local_matrix}), and (\ref{mat:G_convol}), we see that any $(b+1)r \times n$ submatrix of $\bG_{\text{convol}}(\bg_0, \bg_1, \ldots, \bg_{r-1})$ which contains $b+1$ consecutive copies of $\bB(\bg_0, \bg_1, \ldots, \bg_{r-1})$ lying on the main diagonal is identical to the generator matrix $\bG_{\text{local}}(\bg_0, \bg_1, \ldots, \bg_{r-1})$ of the $(r, t)$-PCS-local subcode $\mathcal{C}_{\text{local}}$ of the mother code $\mathcal{C}_0$ in the $t$-fold full rank code chain. Hence, $\bG_{\text{convol}}(\bg_0, \bg_1, \ldots, \bg_{r-1})$ is composed of an infinite number of copies of $\bG_{\text{local}}(\bg_0, \bg_1, \ldots, \bg_{r-1})$ confined to a diagonal band of width $n$. Therefore, $\bG_{\text{local}}(\bg_0, \bg_1, \ldots, \bg_{r-1})$ is a \hl{local submatrix} of $\bG_{\text{convol}}(\bg_0, \bg_1, \ldots, \bg_{r-1})$.

From the diagonal band structure of $\bG_{\text{convol}}(\bg_0, \bg_1, \ldots, \bg_{r-1})$, we see that simultaneously shifting a copy of the local submatrix $\bG_{\text{local}}(\bg_0, \bg_1, \ldots, \bg_{r-1})$ of $\bG_{\text{convol}}(\bg_0, \bg_1, \ldots, \bg_{r-1})$ $t$ positions to the right and $r$ positions down gives another copy of $\bG_{\text{local}}(\bg_0, \bg_1, \ldots, \bg_{r-1})$. This shifting is called $(r, t)$-diagonal sliding. A submatrix of size $(b+1)r \times n$ that covers a local submatrix $\bG_{\text{local}}(\bg_0, \bg_1, \ldots, \bg_{r-1})$ of $\bG_{\text{convol}}(\bg_0, \bg_1, \ldots, \bg_{r-1})$ is called an $(r, t)$-local window. Sliding the first $(r, t)$-local window (the top left corner of $\bG_{\text{convol}}(\bg_0, \bg_1, \ldots, \bg_{r-1})$) diagonally and indefinitely covers all local submatrices of $\bG_{\text{convol}}(\bg_0, \bg_1, \ldots, \bg_{r-1})$. Hence, the matrix $\bG_{\text{convol}}(\bg_0, \bg_1, \ldots, \bg_{r-1})$ is obtained by diagonally sliding the $(r, t)$-local submatrix $\bG_{\text{local}}(\bg_0, \bg_1, \ldots, \bg_{r-1})$ of $\mathcal{C}_{\text{local}}$ indefinitely.

From the sliding window coverage point of view, we could treat the matrix $\bG_{\text{convol}}(\bg_0, \bg_1, \ldots, \bg_{r-1})$ as a global sliding expansion of the generator matrix $\bG_{\text{local}}(\bg_0, \bg_1, \ldots, \bg_{r-1})$ of the $(r, t)$-PCS-local subcode $\mathcal{C}_{\text{local}}$ of $\mathcal{C}_0$. For any nonnegative integer $\ell$, $\ell$ consecutive local windows cover $\ell$ copies of $\bG_{\text{local}}(\bg_0, \bg_1, \ldots, \bg_{r-1})$. Such a group of $\ell$ consecutive local windows is called a \hl{regional window} of multiplicity $\ell$. The matrix covered by a regional window of multiplicity $\ell$ is called a \hl{regional submatrix} $\bG_{\text{regional}}(\ell, \bg_0, \bg_1, \ldots, \bg_{r-1})$ of $\bG_{\text{convol}}(\bg_0, \bg_1, \ldots, \bg_{r-1})$ of multiplicity $\ell$. Two regional submatrices of $\bG_{\text{convol}}(\bg_0, \bg_1, \ldots, \bg_{r-1})$ of the same multiplicity are identical. Thus, based on its structure, we can view the matrix $\bG_{\text{convol}}(\bg_0, \bg_1, \ldots, \bg_{r-1})$ \hl{locally, regionally, or globally}.

The row space of $\bG_{\text{convol}}(\bg_0, \bg_1, \ldots, \bg_{r-1})$ over $\mathrm{GF}(2)$ gives a \hl{binary rate-$r/t$ nonsystematic time-invariant convolutional code} $\mathcal{C}_{\text{convol}}$. It follows from its structure that $\mathcal{C}_{\text{convol}}$ contains an infinite number of copies of the $(r, t)$-PCS-local subcode $\mathcal{C}_{\text{local}}$ of the mother code $\mathcal{C}_0$ in the $t$-fold full rank code chain as local codes. Hence, $\mathcal{C}_{\text{local}}$ is regarded as a local code of the convolutional code $\mathcal{C}_{\text{convol}}$, and the minimum distance $d_{\text{local}}$ of $\mathcal{C}_{\text{local}}$ is called the \hl{local minimum distance of $\mathcal{C}_{\text{convol}}$}. The row space of a regional submatrix $\bG_{\text{regional}}(\ell, \bg_0, \bg_1, \ldots, \bg_{r-1})$ of multiplicity $\ell$ of $\bG_{\text{convol}}(\bg_0, \bg_1, \ldots, \bg_{r-1})$ gives a regional subcode $\mathcal{C}_{\text{regional}}(\ell)$ of $\mathcal{C}_{\text{convol}}$. Clearly,

\begin{equation}
\lim_{\ell \to \infty} \mathcal{C}_{\text{regional}}(\ell) = \mathcal{C}_{\text{convol}}.
\end{equation}
Hence, we can view the convolutional code $\mathcal{C}_{\text{convol}}$ locally, regionally, or globally. We call the code chain $\mathcal{C}_0 \supset \mathcal{C}_1 \supset \cdots \supset \mathcal{C}_{r-1}$ the generator code chain of the convolutional code $\mathcal{C}_{\text{convol}}$.

From (\ref{mat:G_convol}) and the structure of the belt generator matrix $\bB(\bg_0, \bg_1, \ldots, \bg_{r-1})$, we see that the generator matrix $\bG_{\text{convol}}(\bg_0, \bg_1, \ldots, \bg_{r-1})$ of the convolutional code $\mathcal{C}_{\text{convol}}$ consists of $r$ layers. For $0 \leq i < r$, the $i$-th layer of $\bG_{\text{convol}}(\bg_0, \bg_1, \ldots, \bg_{r-1})$, denoted by $\bG_{\text{convol}}(g_i)$, consists of the generator sequence $\bg_i$ of the $i$-th constituent code $\mathcal{C}_i$ in the code chain $\mathcal{C}_1 \supset \cdots \supset \mathcal{C}_{r-1}$ and its $t$-position shifts to the right. The $i$-th layer $\bG_{\text{convol}}(\bg_i)$ of $\bG_{\text{convol}}(\bg_0, \bg_1, \ldots, \bg_{r-1})$ generates a rate-$1/t$ convolutional code $\mathcal{C}_{\text{convol}}(\bg_i)$. Since the code chain $\mathcal{C}_1 \supset \cdots \supset \mathcal{C}_{r-1}$ is a full rank code chain, the rate-$r/t$ convolutional code $\mathcal{C}_{\text{convol}}$ is the direct-sum of the $r$ rate-$1/t$ convolutional codes $\mathcal{C}_{\text{convol}}(\bg_0), \mathcal{C}_{\text{convol}}(\bg_1), \ldots, \mathcal{C}_{\text{convol}}(\bg_{r-1})$ generated by the $r$ layers $\bG_{\text{convol}}(\bg_0), \bG_{\text{convol}}(\bg_1), \ldots, \bG_{\text{convol}}(\bg_{r-1})$ of the generator matrix $\bG_{\text{convol}}(\bg_0, \bg_1, \ldots, \bg_{r-1})$ of the convolutional code $\mathcal{C}_{\text{convol}}$. We call $\mathcal{C}_{\text{convol}}(\bg_0), \mathcal{C}_{\text{convol}}(\bg_1), \ldots, \mathcal{C}_{\text{convol}}(\bg_{r-1})$ the constituent codes of the rate-$r/t$ convolutional code $\mathcal{C}_{\text{convol}}$. Note that for $0 \leq i < r$, there is one-to-one correspondence between the $i$-th constituent code $\mathcal{C}_{\text{convol}}(\bg_i)$ of $\mathcal{C}_{\text{convol}}$ and the $(1, t)$-PCS-local subcode $\mathcal{C}_{i,\text{local}}$ of the $i$-th constituent code $\mathcal{C}_i$ in the code chain $\mathcal{C}_0 \supset \mathcal{C}_1 \supset \cdots \supset \mathcal{C}_{r-1}$.

The local structure of $\mathcal{C}_{\text{convol}}$ allows it to be decoded based on a designed parity-check matrix $\bH_0$ of the mother code $\mathcal{C}_0$ in the code chain using an SWSC decoding scheme (see Section \ref{sect4e:sliding_window_decoding}). If $\bH_0$ is a low-density matrix whose associated Tanner graph has girth at least $6$, the local codes of $\mathcal{C}_{\text{convol}}$ can be decoded based on $\bH_0$ using a soft-decision iterative decoding algorithm based on belief propagation~\cite{gallager1962,mackay1999}, one at a time in a sliding window manner. Cyclic codes with low-density parity-check matrices include finite geometry codes~\cite{lin2004,lin2022,ryan2025} and some doubly transitive invariant codes~\cite{lin2004}, which will be presented in Sections \ref{sect7:convol_from_PaG} and \ref{sect8:convol_from_DTI}.

\subsection{Encoding in Time Domain}

Let $\mathbf{u}$ be a semi-infinite information sequence to be encoded. The sequence $\mathbf{u}$ is segmented into ordered message blocks, each consisting of $r$ information bits, as follows:
\[
\mathbf{u} = (\mathbf{u}_0, \mathbf{u}_1, \ldots, \mathbf{u}_\ell, \ldots)
\]
where for $\ell \geq 0$, the $\ell$-th message block is
\[
\mathbf{u}_\ell = (u_\ell^{(0)}, u_\ell^{(1)}, \ldots, u_\ell^{(r-1)}).
\]

Let $D$ be a shift operator which shifts a sequence one bit-position to the right. For any $j \geq 0$, let $D^j$ denote shifting a sequence $j$ bit-positions to the right. For $\ell \geq 0$, the generator sequence $\mathbf{g}_i$ of the $i$-th cyclic code $\mathcal{C}_i$ in the $\ell$-th row-bend of $\bG_{\text{convol}}(\bg_0, \bg_1, \ldots, \bg_{r-1})$ is denoted by $D^{\ell t}\mathbf{g}_i$.

Then, the code sequence $\mathbf{v}$ for the information sequence $\mathbf{u}$ is given by
\begin{equation} \label{eqn:convol_encode_in_time_domain}
\begin{array}{lll}
\bv & = & \bu \bG_{\convol}(\bg_0, \bg_1, \ldots , \bg_{r-1}) \\
& = & \bu_0 \bB(\bg_0, \bg_1, \ldots , \bg_{r-1}) + \bu_1 D^t \bB(\bg_0, \bg_1, \ldots , \bg_{r-1}) +  \ldots + \bu_{\ell} D^{\ell t} \bB(\bg_0, \bg_1, \ldots , \bg_{r-1}) + \ldots \\
& = & (u_0^{(0)} \bg_0 + u_0^{(1)} \bg_1 + \ldots + u_0^{(r-1)} \bg_{r-1}) +   (u_1^{(0)} D^t \bg_0 + u_1^{(1)} D^t \bg_1 + \ldots + u_1^{(r-1)} D^t \bg_{r-1}) +  \ldots \\
& ~ & + (u_{\ell}^{(0)} D^{\ell t} \bg_0 + u_{\ell}^{(1)} D^{\ell t} \bg_1 + \ldots + u_{\ell}^{(r-1)} D^{\ell t} \bg_{r-1}) +  \ldots 
\end{array}
\end{equation}

The code sequence $\mathbf{v}$ is a sequence of ordered code blocks,
\[
\mathbf{v} = (\mathbf{v}_0, \mathbf{v}_1, \ldots, \mathbf{v}_\ell, \ldots),
\]
each consisting of $t$ code bits, and for $\ell \geq 0$, the $\ell$-th code block is
\[
\mathbf{v}_\ell = (v_\ell^{(0)}, v_\ell^{(1)}, \ldots, v_\ell^{(t-1)}).
\]

Every unit of time, a message block is shifted into the encoder and a code block is produced. The encoder has $r$ inputs and $t$ outputs. At time unit $\ell$, the $\ell$-th code block $\mathbf{v}_\ell$ is formed based on the message block $\mathbf{u}_\ell$ in the current time unit $\ell$ and the message blocks in the preceding $\left\lceil \frac{n - k_{r-1} + 1 \cred{ + c}}{t} \right\rceil - 1$ time units, $\ell - 1, \ell - 2, \ldots, \ell - \left\lceil \frac{n - k_{r-1} + 1 \cred{ + c}}{t} \right\rceil + 1$. For $\ell < i$, the $r$ information symbols in $\mathbf{u}_{\ell-i}$ are set to zero, i.e., $\mathbf{u}_{\ell-i} = \mathbf{0}$. Hence, the encoder needs a memory to store $\left\lceil \frac{n - k_{r-1} + 1 \cred{ + c}}{t} \right\rceil - 1$ previously encoded message blocks. Later in this section, we show that the number $\nu$ of information bits that must be stored is upper bounded by
\[
\nu \leq \left\lceil \frac{n - k_0}{t} \right\rceil + \left\lceil \frac{n - k_1}{t} \right\rceil + \cdots + \left\lceil \frac{n - k_{r-1}}{t} \right\rceil.
\]
The number $\nu$ is referred to as the \hl{overall constraint length} (simply constraint length) of the rate-$r/t$ convolutional code $\mathcal{C}_{\text{convol}}$, and we call $\mathcal{C}_{\text{convol}}($ a rate-$r/t$ $(t, r, \nu)$ convolutional code. From (\ref{eqn:convol_encode_in_time_domain}) and the structure of $\bB(\bg_0, \bg_1, \ldots, \bg_{r-1})$, we can see that the message block $\mathbf{u}_\ell$ in the $\ell$-th time unit can be uniquely retrieved from the code block $\mathbf{v}_\ell$ in the $\ell$-th time unit and the message blocks in the $\left\lceil \frac{n - k_{r-1} + 1}{t} \right\rceil - 1$ preceding time units.

Also, from (\ref{eq:B_matrix}), (\ref{eq:G_local_matrix}), and (\ref{eqn:convol_encode_in_time_domain}), we see that, for $\ell \geq 0$, the sequence formed by $\mathbf{u}_\ell D^{\ell t} \bB(\bg_0, \bg_1, \ldots, \bg_{r-1})$ is a codeword in $\mathcal{C}_0$ with $\ell t$ zeros preceding it and a semi-infinite string of zeros following it. Hence, a code sequence in $\mathcal{C}_{\text{convol}}$ is a linear combination of $t$-bit shifted codewords in $\mathcal{C}_0$. From this perspective, the convolutional code $\mathcal{C}_{\text{convol}}$ may be regarded as a global sliding combination of $r$ cyclic codes in the code chain $\mathcal{C}_0 \supset \mathcal{C}_1 \supset \cdots \supset \mathcal{C}_{r-1}$.

In general, we call a convolutional code $\mathcal{C}_{\text{convol}}$ constructed based on a chain of cyclic codes a cyclic-code-based (CC-based) convolutional code. If the mother code $\mathcal{C}_0$ in the code chain is a BCH code, we call $\mathcal{C}_{\text{convol}}$ a BCH convolutional code. If $\mathcal{C}_0$ is a cyclic finite geometry (FG) LDPC code, we call $\mathcal{C}_{\text{convol}}$ an FG-LDPC convolutional code. CC-based convolutional codes possess algebraic or geometric structure.

\subsection{Code Construction in Transform Domain}

The encoding operation given in (\ref{eqn:convol_encode_in_time_domain}) is performed in the time domain, which requires convolution operations. The convolutional code $\mathcal{C}_{\mathrm{convol}}$ can also be encoded in the transform domain, where the convolution operations are replaced by polynomial multiplications~\cite{lin2004}.

In the polynomial representation, the information sequence
\[
\bu = (\bu_0, \bu_1, \bu_2, \ldots, \bu_\ell, \ldots)
\]
is first demultiplexed into $r$ subsequences $\bu^{(0)}, \bu^{(1)}, \ldots, \bu^{(r-1)}$, where for $0 \leq i < r$,
\begin{equation}
\bu^{(i)} = \left( u_i,\, u_{r+i},\, \ldots,\, u_{\ell r + i} + \ldots \right).
\label{eq:demux}
\end{equation}
Next, we form the following polynomial using the components of $\bu^{(i)}$ as the coefficients:
\begin{equation}
\bu^{(i)}(X) = u_i + u_{r+i} X^{r} + \ldots + u_{\ell r+i} X^{\ell r} + \ldots
\label{eq:poly}
\end{equation}
Then, the entire information polynomial $\bu(X)$ is obtained by multiplexing the $r$ polynomials $\bu^{(0)}(X),\, X \bu^{(1)}(X),\, \ldots,\, X^{r-1} \bu^{(r-1)}(X)$, i.e.,
\begin{equation}
\bu(X) = \bu^{(0)}(X) + X \bu^{(1)}(X) + \ldots + X^{r-1} \bu^{(r-1)}(X).
\label{eq:mux}
\end{equation}

We call the $r$ polynomials $\bu^{(0)}(X), \bu^{(1)}(X), \ldots, \bu^{(r-1)}(X)$ the \emph{decomposed information polynomials} of $\bu(X)$. For encoding, $\bu^{(0)}(X), \bu^{(1)}(X), \ldots, \bu^{(r-1)}(X)$ are applied to the $r$ input terminals of an encoder to produce a code polynomial $\bv(X)$.

For $0 \leq i < r$, let $c_i = \lfloor \frac{n - k_i + 1}{t} \rfloor$. Express the generator polynomial $\bg_i(X)$ of the $i$-th cyclic code $\mathcal{C}_i$ in the code chain $\mathcal{C}_0 \supset \mathcal{C}_1 \supset \cdots \supset \mathcal{C}_{r-1}$ in the following form:
\begin{equation}
\bg_i(X) = \bg_i^{(0)}(X^t) + X\, \bg_i^{(1)}(X^t) + \cdots + X^{t-1} \bg_i^{(t-1)}(X^t).
\label{eq:genpoly}
\end{equation}
where for $0 \leq j < t$,
\begin{equation}
\bg_i^{(j)}(X) = g_{i,j} + g_{i,t+j} X + g_{i,2t+j} X^2 + \cdots + g_{i,c_i t + j} X^{c_i}.
\label{eq:genpolydecomp}
\end{equation}
Note that if $c_i t > n - k_i$, then $g_{i,c_i t + j} = 0$. The degree of $\bg_i^{(j)}(X)$ is upper bounded by $\lceil \frac{n - k_{r-1}}{t} \rceil$. The $t$ polynomials $\bg_i^{(0)}(X), \bg_i^{(1)}(X), \ldots, \bg_i^{(t-1)}(X)$ are said to form a \hl{$t$-fold decomposition of $\bg_i(X)$}. Conversely, $\bg_i(X)$ is referred to as the \hl{$t$-fold composition} of the polynomials $\bg_i^{(0)}(X), \bg_i^{(1)}(X), \ldots, \bg_i^{(t-1)}(X)$.

For $0 \leq j < t$, the following polynomial
\begin{equation}
\bv^{(j)}(X) = \bu^{(0)}(X)\, \bg_0^{(j)}(X) + \bu^{(1)}(X)\, \bg_1^{(j)}(X) + \cdots + \bu^{(r-1)}(X)\, \bg_{r-1}^{(j)}(X),
\label{eq:codepoly}
\end{equation}
is the code polynomial at the $j$-th output terminal of the encoder for the rate-$r/t$ convolutional code $\mathcal{C}_{\mathrm{convol}}$. The overall code polynomial $\bv(X)$ corresponding to the information polynomial $\bu(X)$ is obtained by multiplexing the $t$ code polynomials $\bv^{(0)}(X), \bv^{(1)}(X), \ldots, \bv^{(t-1)}(X)$. After multiplexing,
\begin{equation}
\bv(X) = \bv^{(0)}(X^t) + X\, \bv^{(1)}(X^t) + \cdots + X^{t-1} \bv^{(t-1)}(X^t).
\label{eq:codepolyfinal}
\end{equation}

The entire transform domain encoding operation can be put into matrix form~\cite{lin2004}. Using the $rt$ $t$-fold decompositions of the generator polynomials $\bg_0(X), \bg_1(X), \ldots, \bg_{r-1}(X)$ of the $r$ cyclic codes in the code-chain $\mathcal{C}_0 \supset \mathcal{C}_1 \supset \cdots \supset \mathcal{C}_{r-1}$, we form the following $r \times t$ transform domain generator matrix for the rate-$r/t$ convolutional code $\mathcal{C}_{\mathrm{convol}}$:
\begin{equation}
\bG_{\mathrm{convol}}(X) =
\begin{bmatrix}
\bg_0^{(0)}(X) & \bg_0^{(1)}(X) & \cdots & \bg_0^{(t-1)}(X) \\
\bg_1^{(0)}(X) & \bg_1^{(1)}(X) & \cdots & \bg_1^{(t-1)}(X) \\
\vdots       & \vdots       & \ddots & \vdots         \\
\bg_{r-1}^{(0)}(X) & \bg_{r-1}^{(1)}(X) & \cdots & \bg_{r-1}^{(t-1)}(X)
\end{bmatrix}.
\label{eq:genmatrix}
\end{equation}
Let $\bU(X) = [\bu^{(0)}(X), \bu^{(1)}(X), \ldots, \bu^{(r-1)}(X)]$ be the $r$-tuple of input information sequences (polynomials). Then,
\begin{equation}
\bV(X) = \bU(X)\, \bG_{\mathrm{convol}}(X)
= [\bv^{(0)}(X), \bv^{(1)}(X), \ldots, \bv^{(t-1)}(X)]
\label{eq:outputtuple}
\end{equation}
is the $t$-tuple of output code sequences (polynomials).

For $0 \leq i < r$, let $\nu_i$ be the largest degree among the $t$ decomposed polynomials $\bg_i^{(0)}(X), \bg_i^{(1)}(X), \ldots, \bg_i^{(t-1)}(X)$ of $\bg_i(X)$, i.e.,
\begin{equation}
\nu_i = \max_{0 \leq j < t} \left\{ \deg\, \bg_i^{(j)}(X) \right\} \leq \left\lceil \frac{n - k_i}{t} \right\rceil.
\label{eq:memoryorder}
\end{equation}
Then, the overall constraint length of the convolutional code $\mathcal{C}_{\mathrm{convol}}$ is
\begin{equation}
\nu = \nu_0 + \nu_1 + \cdots + \nu_{r-1}.
\label{eq:constraintlength}
\end{equation}
Hence, $\mathcal{C}_{\mathrm{convol}}$ is a rate-$r/t$ $(t, r, \nu)$ convolutional code with overall constraint length $\nu$ and local minimum distance $d_{\mathrm{local}}$ lower bounded by the minimum distance $d_0$ of the mother code $\mathcal{C}_0$ in the code chain $\mathcal{C}_0 \supset \mathcal{C}_1 \supset \cdots \supset \mathcal{C}_{r-1}$. The parameter $\nu_i$ is referred to as the memory order of the $i$-th input information sequence $\bu^{(i)}(X)$. For $r = 1$ and $t > 1$, $\nu$ is the memory order of a rate-$1/t$ $(t, 1, \nu)$ convolutional code.

If the shifting factor $t$ is a power of 2, say $t = 2^\ell$, then the generator polynomial $\bg_i(X)$ of the $i$-th cyclic code in the code-chain $\mathcal{C}_0 \supset \mathcal{C}_1 \supset \cdots \supset \mathcal{C}_{r-1}$ given by~\eqref{eq:genpoly} can be put in the following $2^\ell$-fold composition form:
\begin{align}
\bg_i(X) &= g_i^{(0)}\left(X^{2^\ell}\right) + X\, \bg_i^{(1)}\left(X^{2^\ell}\right) + \cdots
+ X^{2^\ell - 1} \bg_i^{(2^\ell - 1)}\left(X^{2^\ell}\right) \nonumber \\
&= \left(\bg_i^{(0)}(X)\right)^{2^\ell} + X\, \left(\bg_i^{(1)}(X)\right)^{2^\ell} + \cdots
+ X^{2^\ell - 1} \left(\bg_i^{(2^\ell - 1)}(X)\right)^{2^\ell}.
\label{eq:power2decomp}
\end{align}

For $r = 1$ and $t = 2^\ell$, the $2^\ell$ generator polynomials $\bg_0^{(0)}(X), \bg_0^{(1)}(X), \ldots, \bg_0^{(2^\ell - 1)}(X)$ of the rate-$1/2^\ell$ convolutional code $\mathcal{C}_{\mathrm{convol}}(\bg_0)$ are relatively prime, and their greatest common divisor (GCD) is $1$, i.e.,
\[
\GCD\left\{\, \bg_0^{(j)}(X) : 0 \leq j \leq 2^\ell - 1 \,\right\} = 1.
\]
To show this, suppose they are not relatively prime and their GCD is a nonzero degree polynomial. Let $\bb(X)$ be the GCD of the $2^\ell$ generator polynomials $\bg_0^{(0)}(X), \bg_0^{(1)}(X), \ldots, \bg_0^{(2^\ell - 1)}(X)$ of $\mathcal{C}_{\mathrm{convol}}(\bg_0)$. Since $\bb(X)$ is binary,
\(
\bb\left( X^{2^\ell} \right) = \left( \bb(X) \right)^{2^\ell}.
\)
It follows from~\eqref{eq:power2decomp} that $\bb\left( X^{2^\ell} \right)$ divides the generator polynomial $\bg_0(X)$ of the cyclic code $\mathcal{C}_0$. Therefore, $\bb\left( X^{2^\ell} \right)$ divides $X^n + 1$. However, this contradicts the fact that $X^n + 1$ is square-free (power-2-free) for odd $n$. Hence, all the $2^\ell$ generator polynomials of the rate-$1/2^\ell$ convolutional code $\mathcal{C}_{\mathrm{convol}}(\bg_0)$ must be relatively prime and their GCD must be $1$. This relative-primality property ensures the encoder of the rate-$1/2^\ell$ convolutional code $\mathcal{C}_{\mathrm{convol}}(\bg_0)$ is noncatastrophic~\cite{lin2004}.

\subsection*{Example 3}
\begin{example} \label{eg3}
In Example \ref{eg1}, we constructed a 4-fold full rank code chain $\mathcal{C}_0 \supset \mathcal{C}_1 \supset \mathcal{C}_2$ of length $r = 3$ in which the mother code $\mathcal{C}_0$ is the $(63, 36)$ BCH code with dimension $k_0 = 36$ and minimum distance $11$. The two descendant codes of the mother code $\mathcal{C}_0$ are $(63, 35)$ and $(63, 33)$ codes with dimensions $k_1 = 35$ and $k_2 = 33$, respectively. With shifting factor $t = 4$ and shifting span $b = 8$, we showed that the $(3, 4)$-PCS-local subcode of the mother code $\mathcal{C}_0$ in the code chain is a $(63, 27)$ code with local minimum distance at least $11$.

In the following, we use the code chain $\mathcal{C}_0 \supset \mathcal{C}_1 \supset \mathcal{C}_2$ and its two subchains to construct 3 convolutional codes with rates $3/4$, $1/2$, and $1/4$. We start with the construction of the rate-$3/4$ convolutional code $\mathcal{C}_{\mathrm{convol}}$ using the code chain $\mathcal{C}_0 \supset \mathcal{C}_1 \supset \mathcal{C}_2$. To construct the generator matrix $G_{\mathrm{convol}}(g_0, g_1, g_2)$ of $\mathcal{C}_{\mathrm{convol}}(\bg_0, \bg_1, \bg_2)$, we first put the generator polynomials of three codes in the code chain in the 4-fold composition form:
{\small
\begin{align*}
\bg_0(X) &= 1 + X + X^4 + X^8 + X^{15} + X^{17} + X^{18} + X^{19}   + X^{21} + X^{22} + X^{27} \\
       &= (1 + X^4 + X^8) + X(1 + X^{16} + X^{20}) + X^2(X^{16} + X^{20})   + X^3(X^{12} + X^{16} + X^{24}), \\
\bg_1(X) &= 1 + X^2 + X^4 + X^5 + X^8 + X^9 + X^{15} + X^{16}  + X^{17} + X^{20} + X^{21} + X^{23} + X^{27} + X^{28} \\
       &= (1 + X^4 + X^8 + X^{16} + X^{20} + X^{28}) + X(X^4 + X^8   + X^{16} + X^{20}) + X^2 + X^3(X^{12} + X^{20} + X^{24}), \\
\bg_2(X) &= 1 + X + X^3 + X^7 + X^8 + X^{11} + X^{15} + X^{17}   + X^{19} + X^{20} + X^{24} + X^{25} + X^{27} + X^{30} \\
       &= (1 + X^8 + X^{20} + X^{24}) + X(1 + X^{16} + X^{24})   + X^2(X^{28}) + X^3(1 + X^4 + X^8 + X^{12} + X^{16} + X^{24}).
\end{align*}
}
From the 4-fold composition forms above, we find their 4-fold decompositions as follows:
{\small
\begin{align*}
\bg_0^{(0)}(X) &= 1 + X + X^2, \bg_0^{(1)}(X) = 1 + X^4 + X^5, \bg_0^{(2)}(X) = X^4 + X^5, \bg_0^{(3)}(X) = X^3 + X^4 + X^6, \\
\bg_1^{(0)}(X) &= 1 + X + X^2 + X^4 + X^5 + X^7, \bg_1^{(1)}(X) = X + X^2 + X^4 + X^5, \bg_1^{(2)}(X) = X^2, \bg_1^{(3)}(X) = X^3 + X^5 + X^6, \\
\bg_2^{(0)}(X) &= 1 + X^2 + X^5 + X^6, \bg_2^{(1)}(X) = 1 + X^4 + X^6, \bg_2^{(2)}(X) = X^7, \bg_2^{(3)}(X) = 1 + X + X^2 + X^3 + X^4 + X^6.
\end{align*}
}
For $0 \leq i \leq 2$, we find that $\GCD\{\bg_i^{(j)}(X): 0 \leq j \leq 3\} = 1$.

The 4-fold decompositions above give the generator polynomials of the rate-$3/4$ convolutional code $\mathcal{C}_{\mathrm{convol}}$. The generator matrix of $\mathcal{C}_{\mathrm{convol}}$ is a $3 \times 4$ matrix:
\[
\bG_{\mathrm{convol}}(\bg_0, \bg_1, \bg_2)(X) =
\begin{bmatrix}
\bg_0^{(0)}(X) & \bg_0^{(1)}(X) & \bg_0^{(2)}(X) & \bg_0^{(3)}(X) \\
\bg_1^{(0)}(X) & \bg_1^{(1)}(X) & \bg_1^{(2)}(X) & \bg_1^{(3)}(X) \\
\bg_2^{(0)}(X) & \bg_2^{(1)}(X) & \bg_2^{(2)}(X) & \bg_2^{(3)}(X)
\end{bmatrix}.
\]

From the degree distributions of the generator polynomials, the overall constraint length of $\mathcal{C}_{\mathrm{convol}}(\bg_0, \bg_1, \bg_2)$ is $\nu = 6 + 7 + 7 = 20$. Hence, $\mathcal{C}_{\mathrm{convol}}$ is a rate-$3/4$ $(4, 3, 20)$ convolutional code of overall constraint length $20$ with minimum local distance at least $11$. Since the code chain $\mathcal{C}_0 \supset \mathcal{C}_1 \supset \mathcal{C}_2$ is a full rank code chain, $\mathcal{C}_{\mathrm{convol}}$ is the direct-sum of its three rate-$1/4$ constituent convolutional codes constructed based on the three codes $\mathcal{C}_0$, $\mathcal{C}_1$, and $\mathcal{C}_2$ in the code chain, separately. The encoder of each rate-$1/4$ constituent convolutional code of $\mathcal{C}_{\mathrm{convol}}$ is noncatastrophic.

Suppose we set the shifting factor to $t = 2$ and the shifting span $b = (k_1 - 1)/2 = (35 - 1)/2 = 17$. Using the $(63, 35)$ code $\mathcal{C}_1$ in the code chain, we can construct a rate-$1/2$ $(2, 1, 14)$ convolutional code $\mathcal{C}_{\mathrm{convol}}(\bg_1)$ of memory order $14$ with local minimum distance $12$. The 2-fold decomposition of the generator polynomial $\bg_1(X)$ of $\mathcal{C}_1$ gives the following two generator polynomials:
\begin{align*}
\bg_1^{(0)}(X) &= 1 + X + X^2 + X^4 + X^8 + X^{10} + X^{14}, \\
\bg_1^{(1)}(X) &= X^2 + X^4 + X^7 + X^8 + X^{10} + X^{11} + X^{13}.
\end{align*}
The two polynomials $\bg_1^{(0)}(X)$ and $\bg_1^{(1)}(X)$ are relatively prime. Hence, the encoder of $\mathcal{C}_{\mathrm{convol}}(\bg_1)$ is noncatastrophic.

For $t = 4$, the rate-$1/4$ convolutional code $\mathcal{C}_{\mathrm{convol}}(\bg_2)$ constructed based on the $(63, 33)$ cyclic code $\mathcal{C}_2$ in the code chain is a $(4, 1, 7)$ convolutional code of memory order $7$ with minimum local distance at least $12$. The generator polynomials of $\mathcal{C}_{\mathrm{convol}}(\bg_2)$ are formed by taking the 4-fold decomposition of the generator polynomial $\bg_2(X)$ of $\mathcal{C}_2$ as given above.
\end{example}

\subsection*{Example 4}
\begin{example} \label{eg4}
In Example \ref{eg2}, we constructed a 5-fold full rank code chain $\mathcal{C}_0 \supset \mathcal{C}_1 \supset \mathcal{C}_2 \supset \mathcal{C}_3$ of length $r = 4$ in which the mother code $\mathcal{C}_0$ is the $(255, 115)$ BCH code with designed minimum distance $d_0 = 43$. The $(4, 5)$-PCS-local subcode $\mathcal{C}_{\mathrm{local}}$ of the mother code is a $(255, 84)$ code with local minimum distance lower bounded by $43$. Using this code chain, we can construct a rate-$4/5$ convolutional code with overall constraint length upper bounded by $119$ and local minimum distance at least $43$. It is the direct-sum of its 4 rate-$1/5$ constituent convolutional codes with memory orders $28$, $30$, $30$, and $31$, respectively, and local minimum distances all lower bounded by $43$. The encoder of each rate-$1/5$ constituent convolutional code is noncatastrophic.

If we use the subchain $\mathcal{C}_1 \supset \mathcal{C}_2 \supset \mathcal{C}_3$ of the code chain with $\mathcal{C}_1$ as the mother code, we can construct a rate-$3/5$ convolutional code of overall constraint length upper bounded by $91$ and local minimum distance at least the minimum distance $45$ of $\mathcal{C}_1$. If we use the subchain $\mathcal{C}_2 \supset \mathcal{C}_3$ of the code chain, we can construct a rate-$2/5$ convolutional code with overall constraint length upper bounded by $61$ and local minimum distance at least $45$.

Suppose we set the shifting factor $t$ to $4$ and use the subchain $\mathcal{C}_0 \supset \mathcal{C}_1 \supset \mathcal{C}_2$ for construction of a rate-$3/4$ convolutional code. Set the shifting span $b$ to $\lceil \frac{k_2 - 1}{4} \rceil = \lceil \frac{105 - 1}{4} \rceil = 26$. The 4-fold GMD matrix formed by the 4-fold decompositions of the generator multipliers $\mathbf{f}_0(X) = 1$, $\mathbf{f}_1(X) = 1 + X^6 + X^7 + X^8$, and $\mathbf{f}_2(X) = 1 + X + X^2 + X^6 + X^8 + X^{10}$ (given in Example \ref{eg2}) is a $3 \times 4$ matrix:
\[
\bF(X) =
\begin{bmatrix}
1 & 0 & 0 & 0 \\
1 + X^2 & 0 & X & X \\
1 + X^2 & 1 + X^2 + X^3 & 1 + X^3 & 1 + X^2
\end{bmatrix}.
\]
The first 3 columns of $\bF(X)$ form a $3 \times 3$ submatrix with nonzero determinant. Hence, $\bF(X)$ is a full rank matrix and the subchain $\mathcal{C}_0 \supset \mathcal{C}_1 \supset \mathcal{C}_2$ is a 4-fold full rank code chain. The $(3, 4)$-PCS-local subcode of the mother code $\mathcal{C}_0$ in the subchain is a $(255, 81)$ code with minimum distance at least $43$. Using the subchain, we can construct a rate-$3/4$ convolutional code with overall constraint length upper bounded by $110$ and local minimum distance at least $43$. It is the direct-sum of its 3 rate-$1/4$ constituent convolutional codes with memory orders $35$, $37$, and $38$, respectively, and local minimum distances all lower bounded by $43$.

Suppose we set the shifting factor $t$ to $2$ and the shifting span $b$ to $(115 - 1)/2 = 57$. Using the mother code $\mathcal{C}_0$ alone, we can construct a rate-$1/2$ $(2, 1, 70)$ convolutional code of memory order $70$ with local minimum distance at least $43$. If we set the shifting factor $t$ to $3$ and the shifting span $b$ to $(115 - 1)/3 = 38$, we can construct a rate-$1/3$ $(3, 1, 47)$ convolutional code of memory order $47$ with local minimum distance at least $43$ using the mother code $\mathcal{C}_0$ alone.
\end{example}

The above example shows that by using a full rank code chain, we can construct convolutional codes of various rates and overall constraint lengths.

\subsection{Free Distance}

A common measure of the high SNR performance of a convolutional code is its free distance~\cite{lin2004}. The free distance $d_{\mathrm{free}}$ of a convolutional code $\mathcal{C}_{\mathrm{convol}}$ with generator matrix $\mathbf{G}_{\mathrm{convol}}$ is defined as
\begin{equation} \label{eqn:d_free_define}
d_{\free}  \triangleq \min_{\bu, \bu^{\prime}} \{ d(\bv, \bv^{\prime}): \bu \ne \bu^{\prime} \},
\end{equation}
where $\mathbf{v}$ and $\mathbf{v}'$ are code sequences corresponding to the information sequences $\mathbf{u}$ and $\mathbf{u}'$, respectively, and $d(\mathbf{v}, \mathbf{v}')$ is the Hamming distance between $\mathbf{v}$ and $\mathbf{v}'$. Because a convolutional code is a linear code, we have
\begin{equation} \label{eqn:d_free_define2}
\begin{array}{lll}
d_{\free} &=& \min_{\bu, \bu^{\prime}} \{ d(\bv, \bv^{\prime}): \bu \ne \bu^{\prime} \}\\
&=& \min_{\bu} \{ \omega(\bv): \bu \neq {\bf 0} \} \\
&=& \min_{\bu} \{ \omega(\bu \bG_{\convol}): \bu \neq {\bf 0} \}, \\
\end{array}
\end{equation}
where $w(\mathbf{v})$ denotes the Hamming weight of the code sequence $\mathbf{v}$. This shows that the free distance $d_{\mathrm{free}}$ of a convolutional code $\mathcal{C}_{\mathrm{convol}}$ is equal to the minimum weight of its nonzero code sequences.

In~\cite{massey1973}, Massey, Costello, and Justesen (MCJ) proved a lower bound on the free distance of a rate-$1/t$ convolutional code constructed based on a single cyclic code for even values of $t$. Let $\mathcal{C}$ be a cyclic code of odd length $n$ with minimum distance $d_{\min}$ and let $\mathcal{C}_h$ be the dual code of $\mathcal{C}$ with minimum distance $d_{h,\min}$. The authors in~\cite{massey1973} showed that the free distance $d_{\mathrm{free}}$ of the rate-$1/t$ convolutional code constructed based on $\mathcal{C}$ is lower bounded as follows:
\begin{equation}
d_{\mathrm{free}} \geq \min\{ d_{\min},\ 2 d_{h,\min} \}.
\end{equation}
We call the lower bound on $d_{\mathrm{free}}$ given above the MCJ-bound. From the MCJ-bound, we see that, for a rate-$1/t$ convolutional code with even value of $t$ to have a good free distance, $\mathcal{C}$ should be chosen such that $d_{\min}$ and $2 d_{h,\min}$ are as close as possible to each other. A good choice of $\mathcal{C}$ is a code for which $d_{\min}$ is large for a given $n$ and $d_{\min} \approx 2 d_{h,\min}$~\cite{lin2004,massey1973}.

\subsection*{Example 5}
\begin{example} \label{eg5}
Consider the $(255, 139)$ BCH code $\mathcal{C}_0$ constructed based on the field $\mathrm{GF}(2^8)$. Let $\alpha$ be a primitive element in $\mathrm{GF}(2^8)$. The generator polynomial $\bg_0(X)$ of $\mathcal{C}_0$ has $\alpha, \alpha^2, \ldots, \alpha^{30}$ and their conjugates as roots. The minimum distance of this BCH code $\mathcal{C}_0$ is lower bounded by $31$. The dual code $\mathcal{C}_h$ of $\mathcal{C}_0$ is a $(255, 116)$ code whose generator polynomial $\bg_h(X)$ has $1, \alpha, \alpha^2, \ldots, \alpha^{14}$ as roots and it does not have $\alpha^{15}$ as a root. Since $\bg_h(X)$ has $1$ as a root, all the codewords in $\mathcal{C}_h$ have even weights. The BCH bound on the minimum distance $d_h$ of $\mathcal{C}_h$ is at least $16$. (Minimum distance of the dual code of a cyclic code will be discussed in Section V.)

Set shifting factor $t = 2$ and shifting span $b = (139 - 1)/2 = 69$. The $(1,2)$-PCS-local code $\mathcal{C}_{\mathrm{local}}$ of $\mathcal{C}_0$ is a $(255, 70)$ code with local minimum distance $d_{\mathrm{local}} \geq 31$. Using $\mathcal{C}_0$, we can construct a rate-$1/2$ $(2,1,58)$ convolutional code $\mathcal{C}_{\mathrm{convol}}(\bg_0)$ of memory order $58$ with local minimum distance $d_{\mathrm{local}} \geq 31$. The MCJ-bound gives the free distance of $\mathcal{C}_{\mathrm{convol}}(\bg_0)$:
\[
d_{\mathrm{free}} \geq \min\{31, 2 \times 16\} = 31,
\]
i.e., the free distance $d_{\mathrm{free}}$ of $\mathcal{C}_{\mathrm{convol}}(\bg_0)$ is at least $31$, same as its local minimum distance $d_{\mathrm{local}}$.

Suppose we set the shifting factor $t$ to $4$ and $6$. Using $\mathcal{C}_0$, we can construct a rate-$1/4$ $(4,1,29)$ and a rate-$1/6$ $(6,1,20)$ convolutional codes of memory orders $29$ and $20$, respectively. The local minimum and free distances of both codes are lower bounded by $31$.
\end{example}

\subsection{A Sliding-Window Successive-Cancellation Decoding Scheme}\label{sect4e:sliding_window_decoding}

As shown earlier in this section, the rate-$r/t$ convolutional code $\mathcal{C}_{\mathrm{convol}}$ constructed based on a full rank code-chain $\mathcal{C}_0 \supset \mathcal{C}_1 \supset \cdots \supset \mathcal{C}_{r-1}$ is composed of an infinite number of copies of identical local codes $\mathcal{C}_{\mathrm{local}}$ with local minimum distance $d_{\mathrm{local}}$, each is a $(r, t)$-PCS-local subcode of the mother code $\mathcal{C}_0$ of the code chain. Based on the local and diagonal band structure of its generator matrix $\mathbf{G}_{\mathrm{convol}}(\bg_0, \bg_1, \ldots, \bg_{r-1})$ given by (\ref{mat:G_convol}), the convolutional code $\mathcal{C}_{\mathrm{convol}}$ can be decoded using a localized sliding window approach~\cite{zhu2017}.

In decoding a received code sequence, each local decoding is performed on an \emph{adjusted received local codeword} within a local window of size $n$. The decoding is carried out based on a designed parity-check matrix $\mathbf{H}_0$ of the mother code $\mathcal{C}_0$ of the code chain. After completion of a local decoding, the first estimated message block retrieved from the decoded local codeword is delivered to the user(s). Prior to decoding the next received local codeword, the \emph{contribution} of the retrieved message block to the transmitted code sequence is removed from the received code sequence. (We assume that the $r$ information bits are correctly retrieved.) This results in an adjusted received code sequence. Then, we slide the local window ($t$-bit positions) and proceed to decode the next received local codeword.

At the time unit $\ell$, $\ell \geq 0$, the $\ell$-th adjusted received local codeword $\mathbf{R}_\ell$ consists of a sequence of $\omega = \lceil n/t \rceil$ consecutive code blocks. When the decoding of the $\ell$-th adjusted received local codeword $\mathbf{R}_\ell$ is completed, the estimated message block $\hat{\mathbf{u}}_\ell$ is retrieved from the decoded local codeword. Before proceeding to the next local decoding, the contribution of $\hat{\mathbf{u}}_\ell$ to the transmitted code sequence is removed from the received code sequence. This operation is called \emph{cancellation}. After the cancellation operation, the window slides to the right $t$-bit positions. Then, the decoder is ready to decode the next adjusted received local codeword $\mathbf{R}_{\ell+1}$, which is formed by removing the received code block $\mathbf{r}_\ell$ from the adjusted received local codeword $\mathbf{R}_\ell$ and adding the newly received code block $\mathbf{r}_{\ell+\omega+1}$.

The process of cancellation, window sliding, local decoding, and information retrieval continues until the end of the last local received codeword is reached. At the completion of the last local decoding, the \emph{entire} decoded local codeword is accepted and the estimated information bits are retrieved and delivered to the user(s). In this case, termination of the transmitted information sequence with zeros~\cite{lin2004} is not required. Note that in the decoding process, cancellation must be performed before the next local decoding.

The above decoding process consists of 4 phases: successive cancellation, window sliding, local decoding, and information retrieval. We refer to this decoding process as \emph{sliding-window successive-cancellation (SWSC) decoding}. If the mother code $\mathcal{C}_0$ of the code chain is a cyclic LDPC code, such as a cyclic finite geometry LDPC code, local decoding can be conducted with a soft-decision iterative decoding algorithm.

A critical problem in decoding a convolutional code is \emph{error propagation}~\cite{lin2004}. For the above proposed sliding-window successive-cancellation decoding, if at any local decoding step, the estimated message block is incorrect, the information cancellation using this erroneously decoded message block may introduce extra errors in the next adjusted received local codeword, which may cause incorrect decoding of the next local decoding and an incorrect estimate of the next message block. As a result, incorrect local decoding and erroneous estimates of information bits may propagate.

An approach to reduce error propagation effect is to transmit a known sequence such as a sequence of zeros from time to time to reset the decoder. Another approach is simply not to adjust successively received local codewords based on preceding estimates. After decoding for a prefixed length of code sequence (or a sequence of local codewords) is terminated, the entire information sequence is retrieved from the decoded code sequence.

If the mother code $\mathcal{C}_0$ of the code chain is a cyclic LDPC code, in decoding a received code sequence, it is important to have a \emph{reliable start} to reduce the chances of decoder error propagation. To achieve this, we should set the number of decoding iterations for each of some initial preset number of local decoding steps to a larger value. After the decoding process reaches a steady state, we can reduce the number of iterations and start to output the hard decisions of the estimated information blocks, thereby limiting the decoding latency to the window size. During this steady state phase, the decoder error propagation can be reduced by periodically increasing the number of decoding iterations. Finally, in the last local decoding, we can use a large number of iterations to improve the reliability of final decoded blocks. After the completion of the last local decoding, we then make hard decisions on all the code bits and output the entire sequence of estimated information bits.

As shown in Section \ref{sect4a:code_construct_in_time_domain}, a rate-$r/t$ convolutional code $\mathcal{C}_{\mathrm{convol}}$ constructed based on a full rank code-chain $\mathcal{C}_0 \supset \mathcal{C}_1 \supset \cdots \supset \mathcal{C}_{r-1}$ has multi-layer structure, each layer is a rate-$1/t$ convolutional code constructed based on a constituent code in the code chain. Using this multi-layer structure, multi-layer SWSC decoding might be possible. The decoding can be carried out in successive manner layer after layer, starting with the $0$-th layer based on the constituent convolutional code $\mathcal{C}_{\mathrm{convol}}(\bg_0)$ constructed based on $\mathcal{C}_0$ and ending at the $(r-1)$-th layer based on the constituent convolutional code $\mathcal{C}_{\mathrm{convol}}(\bg_{r-1})$ constructed based on $\mathcal{C}_{r-1}$.

With the SWSC decoding based on a properly designed parity-check matrix $\mathbf{H}_0$ for $\mathcal{C}_0$ (or $\mathbf{H}_{\mathrm{local}}$ for $\mathcal{C}_{\mathrm{local}}$), the local minimum distance $d_{\mathrm{local}}$ of $\mathcal{C}_{\mathrm{convol}}$ plays the most important role in determining the performance of the code. However, if $\mathcal{C}_{\mathrm{convol}}$ is decoded globally or over a span longer than the length of the constituent codes in the code chain, a large free distance plays the most important role in determining the performance.

\section{Construction of Convolutional Codes Based on BCH Codes}\label{sect5:convol_from_BCH}

There are various categories of cyclic codes~\cite{lin2004,lin2022}. Based on each category, we can construct a family of convolutional codes with various rates, constraint lengths (memory orders), and local minimum distances. The most well-known category of cyclic codes is BCH codes~\cite{hocquenhem1959,bose1960}, which can be decoded effectively with the well-known Berlekamp-Massey algorithm~\cite{berlekamp1965,berlekamp1984}, a hard-decision decoding algorithm. Convolutional codes with guaranteed local minimum distances can be constructed from BCH codes. The commonly used BCH codes in applications are binary primitive BCH codes. In this section, we focus on the construction of convolutional codes based on binary primitive BCH codes over finite fields of characteristic 2.

\subsection{Binary Primitive BCH Codes and Their Dual Codes}

Let $\mathrm{GF}(2^m)$ be the Galois field of characteristic 2 with $2^m$ elements constructed based on a binary primitive polynomial $\bp(X)$ over $\mathrm{GF}(2)$ of degree $m$ and let $\alpha$ be a primitive element in $\mathrm{GF}(2^m)$. The powers of $\alpha$, $\alpha^0 = 1, \alpha, \alpha^2, \ldots, \alpha^{2^m-2}$, form all the $2^m-1$ nonzero elements of $\mathrm{GF}(2^m)$. For any $m \geq 3$, and $1 \leq \lambda < 2^{m-1}$, there is a binary $(n, k)$ primitive BCH code $\mathcal{C}$ with the following parameters~\cite{lin2004,lin2022}:
\[
n = 2^m - 1, \quad n - k \leq m\lambda, \quad d \geq 2\lambda + 1.
\]
The generator polynomial $\bg(X)$ of $\mathcal{C}$ is the polynomial over $\mathrm{GF}(2)$ of \emph{least degree} which has $\alpha, \alpha^2, \ldots, \alpha^{2\lambda}$ and their conjugates as roots. To construct $\bg(X)$, we first find the minimal polynomials $\mathbf{m}_1(X), \mathbf{m}_2(X), \ldots, \mathbf{m}_{2\lambda}(X)$ of $\alpha, \alpha^2, \ldots, \alpha^{2\lambda}$ and their conjugates. Then, the \emph{least common multiple (LCM)} of $\mathbf{m}_1(X), \mathbf{m}_2(X), \ldots, \mathbf{m}_{2\lambda}(X)$ gives the generator polynomial $\bg(X)$ of $\mathcal{C}$, i.e.,
\begin{equation}
\bg(X) = \mathrm{LCM}\{ \mathbf{m}_1(X), \mathbf{m}_2(X), \ldots, \mathbf{m}_{2\lambda}(X) \}.
\end{equation}
For the binary case,
\begin{equation}
\bg(X) = \mathrm{LCM}\{ \mathbf{m}_1(X), \mathbf{m}_3(X), \ldots, \mathbf{m}_{2\lambda-1}(X) \}.
\end{equation}

The minimum distance $d$ of $\mathcal{C}$ is at least $2\lambda + 1$, which is called the \emph{BCH-bound} (a lower bound) on $d$. The powers $1, 2, \ldots, 2\lambda$ are said to form a \emph{consecutive power span (CPS)} of the roots of $\bg(X)$. The BCH code $\mathcal{C}$ is capable of correcting $\lambda$ or fewer random errors over the binary symmetric channel (BSC). The parameters $\lambda$ and $2\lambda + 1$ are called the \emph{designed error-correcting capability} and \emph{designed minimum distance} of $\mathcal{C}$, respectively. The true minimum distance of $\mathcal{C}$ may be greater than $2\lambda + 1$. If $2\lambda + 1$ is a factor of $2^m - 1$, the true minimum distance of the BCH code $\mathcal{C}$ is equal to the BCH bound~\cite{lin2004}.

Let $\delta$ be the largest power of the roots in $\bg(X)$. If $1$ is not a root of $\bg(X)$, then $\alpha^{\delta+1}, \alpha^{\delta+2}, \ldots, \alpha^{2^m-2}, \alpha^{2^m-1} = 1$ are roots of the parity polynomial $\bh(X) = (X^n + 1)/\bg(X)$ of the BCH code $\mathcal{C}$, and they form a CPS of $\bh(X)$ of length $\rho = 2^m - \delta - 1$. In this case, the reciprocal polynomial $\bg_h(X) = X^k \bh(X^{-1})$ of $\bh(X)$ has $1, \alpha, \alpha^2, \ldots, \alpha^{\rho-1}$ and their conjugates as roots. Then, the cyclic code generated by $\bg_h(X)$ is an $(n, n-k)$ code which is the dual code $\mathcal{C}_h$ of the $(n, k)$ BCH code $\mathcal{C}$. It follows from the BCH-bound that the minimum distance $d_h$ of $\mathcal{C}_h$ is at least $\rho + 1$. Since $1$ is a root of $\bg_h(X)$, $d_h$ is even.

If $1$ is a root of $\bg(X)$, then the CPS of $\bh(X)$ consists of $\alpha^{\delta+1}, \alpha^{\delta+2}, \ldots, \alpha^{2^m-2}$ as roots and $\bg_h(X)$ has $\alpha, \alpha^2, \ldots, \alpha^{\rho-1}$ as roots. In this case, the minimum distance of $\mathcal{C}_h$ is lower bounded by $\rho$. The BCH bound on the minimum distance $d_h$ of the dual code $\mathcal{C}_h$ of the BCH code $\mathcal{C}$ may not be tight. We can use either the BCH code $\mathcal{C}$ or its dual code $\mathcal{C}_h$ for constructing a convolutional code. Note that the dual code of the BCH code $\mathcal{C}$ is not necessarily a BCH code.

\subsection{Construction of Full Rank BCH Code Chains and Convolutional Codes}

For $m > 3$, let $\alpha$ be a primitive element in the Galois field $\mathrm{GF}(2^m)$ and $\bg_0(X)$ be the generator polynomial of a BCH code $\mathcal{C}_0$ of length $n = 2^m - 1$ with minimum distance $d_0$ whose roots are elements in $\mathrm{GF}(2^m)$. Let $r$ and $t$ be two positive integers such that $1 \leq r < t$. For $0 < i < r$, let $\mathbf{f}_i(X)$ be the product of the minimal polynomials of a chosen set of elements in $\mathrm{GF}(2^m)$ such that the $r$ polynomials $\mathbf{f}_1(X), \mathbf{f}_2(X), \ldots, \mathbf{f}_{r-1}(X)$ satisfy the following conditions:
\begin{enumerate}
    \item For $1 \leq i < r$, $\GCD(\mathbf{f}_i(X), \bg_0(X)) = 1$;
    \item For $1 \leq i < r-1$, $\mathbf{f}_{i+1}(X) = \mathbf{q}_{i+1}(X) \mathbf{f}_i(X)$, where $\mathbf{q}_{i+1}(X)$ is the product of the minimal polynomials of a chosen set of elements in $\mathrm{GF}(2^m)$ and $\mathbf{q}_{i+1}(X)$ and $\mathbf{f}_i(X)$ are relatively prime, i.e., $\GCD(\mathbf{q}_{i+1}(X), \mathbf{f}_i(X)) = 1$;
    \item For $1 \leq i, j < r-1$ and $i \neq j$, the ratio $\mathbf{f}_j(X)/\mathbf{f}_i(X)$ is not a polynomial in $X^t$;
    \item The $t$-fold GMD-matrix $\bF(X)$ formed by the $t$-fold decompositions of the polynomials $\mathbf{f}_1(X), \mathbf{f}_2(X), \ldots, \mathbf{f}_{r-1}(X)$ is a full rank matrix.
\end{enumerate}
Clearly, the $r$ polynomials $\mathbf{f}_1(X), \mathbf{f}_2(X), \ldots, \mathbf{f}_{r-1}(X)$ are factors of $X^n + 1$.

Using the polynomials $\mathbf{f}_0(X) = 1, \mathbf{f}_1(X), \mathbf{f}_2(X), \ldots, \mathbf{f}_{r-1}(X)$ as the generator multipliers, we form the following $r$ generator polynomials:
\begin{equation} \label{eqn:bch_generator_polynomials}
\begin{array}{lll}
\bg_0(X) &=& \mathbf{f}_0(X) \bg_0(X),\\
\bg_1(X) &=& \mathbf{f}_1(X) \bg_0(X), \\
\bg_2(X) &=& \mathbf{f}_2(X) \bg_0(X), \\
&\vdots& \\
\bg_{r-1}(X) &=& \mathbf{f}_{r-1}(X) \bg_0(X).
\end{array}
\end{equation}
Clearly, the $r$ polynomials $\bg_0(X), \bg_1(X), \bg_2(X), \ldots, \bg_{r-1}(X)$ are factors of $X^n + 1$. It follows from the properties of $\mathbf{f}_0(X) = 1, \mathbf{f}_1(X), \mathbf{f}_2(X), \ldots, \mathbf{f}_{r-1}(X)$, and Theorems \ref{Thm1} and \ref{Thm2} that the cyclic codes $\mathcal{C}_0, \mathcal{C}_1, \mathcal{C}_2, \ldots, \mathcal{C}_{r-1}$ of length $n = 2^m - 1$ generated by $\bg_0(X), \bg_1(X), \bg_2(X), \ldots, \bg_{r-1}(X)$ form a $t$-fold full rank code chain, $\mathcal{C}_0 \supset \mathcal{C}_1 \supset \cdots \supset \mathcal{C}_{r-1}$ of length $r$, called a \emph{BCH code chain}. For $0 \leq i < r$, let $k_i$ and $d_i$ be the dimension and minimum distance of $\mathcal{C}_i$, respectively. Then, for $0 \leq i < r-1$, $k_i > k_{i+1}$ and $d_i \leq d_{i+1}$.

Using the BCH code chain $\mathcal{C}_0 \supset \mathcal{C}_1 \supset \cdots \supset \mathcal{C}_{r-1}$, we can construct a rate-$r/t$ convolutional code $\mathcal{C}_{\mathrm{convol}}$ with overall constraint length $\nu$ upper bounded by $\lceil \frac{n - k_0}{t} \rceil + \lceil \frac{n - k_1}{t} \rceil + \cdots + \lceil \frac{n - k_{r-1}}{t} \rceil$ and local minimum distance $d_{\mathrm{local}}$ lower bounded by the minimum distance $d_0$ of the mother code $\mathcal{C}_0$ in the code chain $\mathcal{C}_0 \supset \mathcal{C}_1 \supset \cdots \supset \mathcal{C}_{r-1}$. The rate-$r/t$ convolutional code $\mathcal{C}_{\mathrm{convol}}$ is called a \hl{BCH convolutional code}. It is the direct-sum of its $r$ constituent rate-$1/t$ convolutional codes with memory orders upper bounded by $\lceil \frac{n - k_0}{t} \rceil, \lceil \frac{n - k_1}{t} \rceil, \ldots, \lceil \frac{n - k_{r-1}}{t} \rceil$, respectively. For $1 \leq s \leq r$, using a subchain of length $s$ of the BCH code chain $\mathcal{C}_0 \supset \mathcal{C}_1 \supset \cdots \supset \mathcal{C}_{r-1}$, we can construct a rate-$s/t$ BCH convolutional code with local minimum distance lower bounded by $d_0$. Using binary primitive BCH codes as mother codes, we can construct a large class of BCH convolutional codes of various rates, local minimum distances and overall constraint lengths.

For even $t$ and $r = 1$, the convolutional code $\mathcal{C}_{\mathrm{convol}}(\bg_0)$ constructed based on the $(n, k_0)$ cyclic code $\mathcal{C}_0$ alone is a rate-$1/t$ convolutional code of memory order bounded by $\lceil \frac{n - k_0}{t} \rceil$ with local distance at least $d_0$ and free distance lower bounded by $\min(d_0, 2 d_h)$.

The full rank code chains constructed in Examples \ref{eg1} and \ref{eg2} are BCH code chains. Convolutional codes constructed based on these code chains in Examples \ref{eg3} and \ref{eg4} are BCH convolutional codes.

\subsection*{Example 6}
\begin{example} \label{eg6}
In this example, we design a full rank BCH code chain of length 3 to construct convolutional codes of rates $3/4$, $2/3$, $1/2$, and $1/4$. In the design of the code chain, the $(1023, 737)$ even weight BCH code $\mathcal{C}_0$ with dimension $k_0 = 737$ and designed minimum distance $d_0 = 60$ is chosen as the mother code. The even weight $(1023, 737)$ BCH code $\mathcal{C}_0$ is constructed based on the Galois field $\mathrm{GF}(2^{10})$. Let $\alpha$ be a primitive element in $\mathrm{GF}(2^{10})$. The generator polynomial $\bg_0(X)$ of $\mathcal{C}_0$ has $1, \alpha, \alpha^2, \ldots, \alpha^{58}$ and their conjugates as roots. The largest power of the roots in $\bg_0(X)$ is $\delta = 992$. The dual code $\mathcal{C}_h$ of $\mathcal{C}_0$ is a $(1023, 286)$ code with minimum distance $d_h$ lower bounded by $2^{10} - 1 - \delta = 1023 - 992 = 31$.

Next, we set the shifting factor $t$ to $4$ and design a 4-fold full rank code chain $\mathcal{C}_0 \supset \mathcal{C}_1 \supset \mathcal{C}_2$ of length 3 using the $(1023, 737)$ BCH code $\mathcal{C}_0$ as the mother code. To construct the two descendant codes $\mathcal{C}_1$ and $\mathcal{C}_2$ of $\mathcal{C}_0$, we need to design 2 generator multipliers for $\mathcal{C}_1$ and $\mathcal{C}_2$. Examining the roots in $\bg_0(X)$ of $\mathcal{C}_0$, we find that the two elements $\alpha^{59}$ and $\alpha^{61}$ in $\mathrm{GF}(2^{10})$ are not roots of $\bg_0(X)$. The minimal polynomials of these two elements are
\begin{align*}
\mathbf{m}_{59}(X) &= 1 + X^3 + X^4 + X^5 + X^8 + X^9 + X^{10}, \\
\mathbf{m}_{61}(X) &= 1 + X + X^4 + X^5 + X^6 + X^7 + X^8 + X^9 + X^{10}.
\end{align*}
Using the two minimal polynomials $\mathbf{m}_{59}(X)$ and $\mathbf{m}_{61}(X)$, we form two generator multipliers:
\begin{align*}
\mathbf{f}_1(X) &= \mathbf{f}_0(X) \mathbf{m}_{59}(X) = \mathbf{m}_{59}(X), \\
\mathbf{f}_2(X) &= \mathbf{f}_1(X) \mathbf{m}_{61}(X) = \mathbf{m}_{59}(X) \mathbf{m}_{61}(X) \\
       &= 1 + X + X^3 + X^4 + X^5 + X^{13} + X^{14} + X^{16} + X^{17} + X^{18} + X^{20}.
\end{align*}
With $\mathbf{f}_0(X) = 1$, none of the ratios $\mathbf{f}_1(X)/\mathbf{f}_0(X)$, $\mathbf{f}_2(X)/\mathbf{f}_0(X)$, $\mathbf{f}_2(X)/\mathbf{f}_1(X)$ is a polynomial in $X^4$.

The 4-fold GMD-matrix formed by the 4-fold decompositions of $\mathbf{f}_0(X) = 1$, $\mathbf{f}_1(X)$ and $\mathbf{f}_2(X)$ is the following $3 \times 4$ matrix:
\[
\bF(X) =
\begin{bmatrix}
1 & 0 & 0 & 1 \\
1 + X + X^2 & X + X^2 & X^2 & 1 \\
1 + X + X^4 + X^5 & 1 + X + X^3 + X^4 & X^3 + X^4 & 1
\end{bmatrix}.
\]
The first 3 columns of $\bF(X)$ form a $3 \times 3$ submatrix with nonzero determinant. Hence, $\bF(X)$ is a full rank matrix and the polynomials $\mathbf{f}_0(X)$, $\mathbf{f}_1(X)$ and $\mathbf{f}_2(X)$ can be used to construct a 4-fold full rank code chain of length 3 with $\mathcal{C}_0$ as the mother code.

Using $\bg_0(X)$, $\mathbf{f}_1(X)$, and $\mathbf{f}_2(X)$, we form two generator polynomials $\bg_1(X) = \mathbf{f}_1(X)\bg_0(X)$ and $\bg_2(X) = \mathbf{f}_2(X)\bg_0(X)$. The two cyclic codes $\mathcal{C}_1$ and $\mathcal{C}_2$ generated by $\bg_1(X)$ and $\bg_2(X)$ are descendant codes of $\mathcal{C}_0$. $\mathcal{C}_1$ is a $(1023, 727)$ code with dimension $k_1 = 727$ and $\mathcal{C}_2$ is a $(1023, 717)$ code with dimension $k_2 = 717$. Both $\mathcal{C}_1$ and $\mathcal{C}_2$ are even weight BCH codes with designed minimum distances $62$ and $64$, respectively. The largest powers of roots in both $\bg_1(X)$ and $\bg_2(X)$ are $992$, same as the largest power of roots in $\bg_0(X)$. Hence, the minimum distances of the dual codes of the three codes $\mathcal{C}_0$, $\mathcal{C}_1$, and $\mathcal{C}_2$ are lower bounded by $31$. Since $\bF(X)$ is a full rank matrix, the code chain $\mathcal{C}_0 \supset \mathcal{C}_1 \supset \mathcal{C}_2$ formed by the 3 codes is a 4-fold full rank BCH code chain.

Set the shifting span $b = (k_2 - 1)/4 = (717 - 1)/4 = 179$. The $(3, 4)$-PCS-local subcode of $\mathcal{C}_0$ is a $(1023, 540)$ code with local minimum distance lower bounded by $60$.

Using the code chain $\mathcal{C}_0 \supset \mathcal{C}_1 \supset \mathcal{C}_2$ with shifting factor $t = 4$, we can construct a rate-$3/4$ $(4, 3, 223)$ BCH convolutional code $\mathcal{C}_{\mathrm{convol}}$ of overall constraint length $223$ with local minimum distance lower bounded by $60$. It is the direct-sum of its 3 rate-$1/4$ constituent convolutional codes $\mathcal{C}_{0,\mathrm{convol}}(\bg_0)$, $\mathcal{C}_{1,\mathrm{convol}}(\bg_1)$, and $\mathcal{C}_{2,\mathrm{convol}}(\bg_2)$. The constituent code $\mathcal{C}_{0,\mathrm{convol}}(\bg_0)$ is a $(4, 1, 72)$ convolutional code of memory order $72$ with both the local minimum distance and the free distance lower bounded by $60$. The constituent code $\mathcal{C}_{1,\mathrm{convol}}(\bg_1)$ is a $(4, 1, 74)$ convolutional code of memory order $74$ with both the local minimum distance and the free distance lower bounded by $62$. The constituent code $\mathcal{C}_{2,\mathrm{convol}}(\bg_2)$ is a $(4, 1, 77)$ convolutional code of memory order $77$ with the local minimum distance and the free distance lower bounded by $64$ and $62$, respectively.

Suppose we take the subchain $\mathcal{C}_0 \supset \mathcal{C}_1$ of the BCH code chain $\mathcal{C}_0 \supset \mathcal{C}_1 \supset \mathcal{C}_2$ to construct a rate-$2/3$ BCH convolutional code. In this case, the shifting factor $t$ is set to $3$. The 3-fold GMD-matrix formed by the 3-fold decompositions of $\mathbf{f}_0(X) = 1$ and $\mathbf{f}_1(X)$ is a $2 \times 3$ matrix:
\[
\bF(X) =
\begin{bmatrix}
1 & 0 & 0 \\
1 + X + X^3 & X + X^3 & X + X^2
\end{bmatrix}.
\]
Clearly, $\bF(X)$ is a full rank matrix. Hence, $\mathcal{C}_0 \supset \mathcal{C}_1$ is a 3-fold full rank matrix. Using the subchain $\mathcal{C}_0 \supset \mathcal{C}_1$ of the BCH code chain, we can construct a rate-$2/3$ $(3, 2, 195)$ BCH convolutional code of overall constraint length $195$ with local minimum distance $60$.

If we use the mother code $\mathcal{C}_0$ of the code chain alone and set the shifting factor $t$ to $2$, we can construct a rate-$1/2$ $(2, 1, 143)$ convolutional code $\mathcal{C}_{\mathrm{convol}}(\bg_0)$. The two generator polynomials $\bg_0^{(0)}(X)$ and $\bg_0^{(1)}(X)$ of $\mathcal{C}_{\mathrm{convol}}(\bg_0)$ are obtained by taking the 2-fold decomposition of the generator polynomial $\bg_0(X)$ of $\mathcal{C}_0$. The local minimum distance of $\mathcal{C}_{\mathrm{convol}}(\bg_0)$ is at least $60$, and the free distance $d_{\mathrm{free}}$ of $\mathcal{C}_{\mathrm{convol}}(\bg_0)$ is lower bounded by $\min\{d_0, 2 d_h\} \geq \min\{60, 2 \times 31\} = 60$.

If we use $\mathcal{C}_1$ alone and set $t = 3$, we can construct a rate-$1/3$ $(3, 1, 99)$ BCH convolutional code $\mathcal{C}_{\mathrm{convol}}(\bg_1)$ of constraint length $99$ with local minimum distance lower bounded by $62$. If we use $\mathcal{C}_2$ alone and set $t = 4$, we can construct a rate-$1/4$ $(3, 1, 77)$ BCH convolutional code $\mathcal{C}_{\mathrm{convol}}(\bg_2)$ of constraint length $77$ with local minimum distance lower bounded by $64$ and free distance lower bounded by $\min\{64, 2 \times 31\} = 62$.
\end{example}

Convolutional codes can also be constructed by using nonprimitive BCH code chains.

\section{Construction of Convolutional Codes Based on Cyclic Reed-Muller Codes}\label{sect6:convol_from_RM}

Reed-Muller (RM) codes form another well-known class of error-correcting codes discovered in 1954~\cite{muller1954,reed1954}. There are many approaches to the construction of these codes. Each construction approach reveals certain unique structures---algebraic, geometric, combinatorial, graphical, and others---which allow them to be used for certain specific applications with certain decoding methods or algorithms. Lately, these codes were found to be closely related to capacity-approaching polar codes~\cite{arikan2009}. There are also several generalizations of RM codes. In this section, binary cyclic RM codes are used to construct convolutional codes.

\subsection{Reed-Muller Codes}

For two positive integers $m$ and $\mu$ with $m \geq 3$ and $0 \leq \mu \leq m$, there is an RM code $\mathcal{C}(\mu, m)$ of order $\mu$ with length $n = 2^m$, dimension
\begin{equation}
k(\mu, m) = 1 + \binom{m}{1} + \binom{m}{2} + \cdots + \binom{m}{\mu},
\end{equation}
and minimum distance $2^{m-\mu}$~\cite{lin2004}. The dual code of the $\mu$-th order RM code $\mathcal{C}(\mu, m)$ is the $(m-\mu-1)$-th order RM code $\mathcal{C}(m-\mu-1, m)$ with dimension
\begin{equation}
k(m-\mu-1, m) = 1 + \binom{m}{1} + \binom{m}{2} + \cdots + \binom{m}{m-\mu-1},
\end{equation}
and minimum distance $2^{\mu+1}$.

For any nonzero integer $1 \leq \kappa < \mu$, the $\kappa$-th order RM code $\mathcal{C}(\kappa, m)$ is a subcode of the $(\kappa+1)$-th order RM code $\mathcal{C}(\kappa+1, m)$. Hence, the RM codes $\mathcal{C}(1, m), \ldots, \mathcal{C}(\mu, m)$ form an inclusion chain:
\begin{equation}
\mathcal{C}(\mu, m) \supset \mathcal{C}(\mu-1, m) \supset \cdots \supset \mathcal{C}(1, m).
\end{equation}

RM codes are invariant under the group of \hl{affine permutations}~\cite{lin1967, lin2004}. With a specific affine permutation, an RM code of length $2^m$ with the leftmost component of each permuted codeword removed becomes a cyclic RM code of length $2^m-1$~\cite{lin1967, lin2004}. The $\mu$-th order cyclic RM code $\mathcal{C}_{\mathrm{cyc}}(\mu, m)$ has the same dimension $k(\mu, m)$ as that of the $\mu$-th order RM code $\mathcal{C}(\mu, m)$ of length $2^m$, but its minimum distance is $2^{m-\mu}-1$. The $\mu$-th order cyclic RM code $\mathcal{C}_{\mathrm{cyc}}(\mu, m)$ consists of codewords with both even and odd weights.

Let $h$ be a nonnegative integer less than $2^m$. Express $h$ in radix-2 form as follows~\cite{lin2004,kou2001}:
\begin{equation} \label{eqn:rm_code_condition1}
h = h_0 + h_1 2 + h_2 2^2 + \cdots + h_{m-1} 2^{m-1},
\end{equation}
where $h_i = 0$ or $1$ for $0 \leq i < m$. The real-sum of the coefficients $h_0, h_1, \ldots, h_{m-1}$,
\begin{equation} \label{eqn:rm_code_condition2}
w(h) = h_0 + h_1 + h_2 + \cdots + h_{m-1}
\end{equation}
is called the radix-2 weight of $h$.

Let $\alpha$ be a primitive element of the Galois field $\mathrm{GF}(2^m)$ and let $\bg(X)$ be the generator polynomial of the $\mu$-th order cyclic RM code $\mathcal{C}_{\mathrm{cyc}}(\mu, m)$. Then, $\alpha^h$ is a root of $\bg(X)$ if and only if the following condition holds~\cite{lin2004,kou2001}:
\begin{equation} \label{eqn:rm_code_condition3}
0 < w(h) \leq m - \mu - 1.
\end{equation}
The number of nonzero integers that satisfy the condition above is
\begin{equation} \label{eqn:rm_code_condition4}
k(m-\mu-1, m) - 1 = \binom{m}{1} + \cdots + \binom{m}{m-\mu-1}.
\end{equation}
From (\ref{eqn:rm_code_condition3}), we see that $\alpha^0 = 1$ (i.e., $h=0$) is not a root of the generator $\bg(X)$ of the $\mu$-th order cyclic RM code $\mathcal{C}_{\mathrm{cyc}}(\mu, m)$. For the even-weight $\mu$-th order cyclic RM code $\mathcal{C}_{\mathrm{cyc},e}(\mu, m)$, $\alpha^h$ is a root of its generator polynomial $\bg(X)$ if and only if
\begin{equation}  \label{eqn:rm_code_condition5}
0 \leq w(h) \leq m - \mu - 1.
\end{equation}
Note that $h = 0$ satisfies the condition given by (\ref{eqn:rm_code_condition5}) and hence $\alpha^0 = 1$ is a root of $\bg(X)$. The minimum distance of $\mathcal{C}_{\mathrm{cyc},e}(\mu, m)$ is $2^{m-\mu}$.

The dual code $\mathcal{C}_{\mathrm{cyc},e}(m-u-1, m)$ of the $\mu$-th order cyclic RM code $\mathcal{C}_{\mathrm{cyc}}(\mu, m)$ has dimension $k(m-\mu-1, m) - 1$, minimum distance $2^{\mu+1}$, and its generator polynomial $\bg_e(X)$ has $1$ as a root. Hence, $\mathcal{C}_{\mathrm{cyc},e}(m-u-1, m)$ is the even-weight subcode of the $(m-\mu-1)$-th order cyclic RM code $\mathcal{C}_{\mathrm{cyc}}(m-\mu-1, m)$. We call $\mathcal{C}_{\mathrm{cyc},e}(m-u-1, m)$ an even-weight cyclic RM code. The subscript “$e$” stands for even-weight.

If we include $1$ as a root in the generator polynomial $\bg(X)$ of the $\mu$-th order cyclic RM code $\mathcal{C}_{\mathrm{cyc}}(\mu, m)$, the resultant code is the even-weight subcode $\mathcal{C}_{\mathrm{cyc},e}(\mu, m)$ of $\mathcal{C}_{\mathrm{cyc}}(\mu, m)$ with dimension $k(\mu, m) - 1$ and minimum distance $2^{m-\mu}$, and its dual code is the $(m-\mu-1)$-th order cyclic RM code $\mathcal{C}_{\mathrm{cyc}}(m-\mu-1, m)$ with dimension $k(m-\mu-1, m)$ and minimum distance $2^{\mu+1} - 1$. Hence, the dual code of a cyclic RM code with both odd and even weight codewords is an even-weight cyclic RM code, and conversely, the dual code of an even-weight cyclic RM code is a cyclic RM code with both odd and even weight codewords. Cyclic RM codes also have the inclusion structure as shown above. An RM code, cyclic or non-cyclic, has multi-layer structure and can be decoded with a successive cancellation decoding algorithm~\cite{lin2004,kou2001,reed1954}.

\subsection{Construction of RM Convolutional Codes}

Using a cyclic RM code of any order as the mother code, we can construct an RM code chain in the same way as constructing a BCH code chain, as shown in (\ref{eqn:bch_generator_polynomials}). For $1 \leq r < t$, a $t$-fold full rank cyclic RM code chain $\mathcal{C}_0 \supset \mathcal{C}_1 \supset \cdots \supset \mathcal{C}_{r-1}$ of length $r$ with shifting factor $t$ can be used to construct a rate-$r/t$ RM convolutional code with local minimum distance lower bounded by the minimum distance of the cyclic RM mother code $\mathcal{C}_0$.

Consider the case for which $m$ is even. Suppose we choose $\mu = (m-2)/2$ and set the shifting factor $t = 2^\ell$, with $\ell \geq 1$. Using the $((m-2)/2)$-th order cyclic RM code $\mathcal{C}_{\mathrm{cyc}}((m-2)/2, m)$ (or $\mathcal{C}_{\mathrm{cyc},e}((m-2)/2, m)$), we can construct a rate-$1/2^\ell$ RM convolutional code of memory order $\lceil \frac{k(m/2, m) - 1}{2^\ell} \rceil$ (or $\lceil \frac{k(m/2, m)}{2^\ell} \rceil$) with local minimum distance $2^{(m+2)/2} - 1$ (or $2^{(m+2)/2}$). Since the minimum distance of the dual code of $\mathcal{C}_{\mathrm{cyc}}((m-2)/2, m)$ (or $\mathcal{C}_{\mathrm{cyc},e}((m-2)/2, m)$) is $2^{m/2}$ (or $2^{m/2} - 1$), the MCJ lower bound on the free distance of the rate-$1/2^\ell$ RM convolutional code is
\[
\min\{2^{(m+2)/2} - 1, 2 \times 2^{m/2}\} = 2^{(m+2)/2} - 1,
\]
(or $\min\{2^{(m+2)/2}, 2 \times (2^{m/2} - 1)\} = 2^{(m+2)/2} - 2$).

\subsection*{Example 7}
\begin{example} \label{eg7}
Let $m = 8$. Suppose we choose $\mu = (m-2)/2 = 3$. The 3rd order cyclic RM code $\mathcal{C}_0$ is a $(255, 93)$ code of length $255$ with dimension $k_0 = 93$ and minimum distance $d_0 = 31$. The generator polynomial $\bg_0(X)$ of $\mathcal{C}_0$ has $162$ roots which are elements of $\mathrm{GF}(2^8)$ and can be determined by using the weight condition in (\ref{eqn:rm_code_condition3}). The minimal polynomials of these roots are 
$\mathbf{m}_{1}(X)$, $\mathbf{m}_{3}(X)$, $\mathbf{m}_{5}(X)$, $\mathbf{m}_{7}(X)$, $\mathbf{m}_{9}(X)$, $\mathbf{m}_{11}(X)$, $\mathbf{m}_{13}(X)$, $\mathbf{m}_{15}(X)$, $\mathbf{m}_{17}(X)$, $\mathbf{m}_{19}(X)$, $\mathbf{m}_{21}(X)$, $\mathbf{m}_{23}(X)$, $\mathbf{m}_{25}(X)$, $\mathbf{m}_{27}(X)$, $\mathbf{m}_{29}(X)$, $\mathbf{m}_{37}(X)$, $\mathbf{m}_{39}(X)$, $\mathbf{m}_{43}(X)$, $\mathbf{m}_{45}(X)$, $\mathbf{m}_{51}(X)$, $\mathbf{m}_{53}(X)$, and $\mathbf{m}_{85}(X)$. The dual code of $\mathcal{C}_0$ is the 4th order even-weight $(255, 162)$ RM code $\mathcal{C}_{0,h}$ with minimum distance $16$.

Note that $k_0 - 1 = 92 = 2 \times 2 \times 23$. Using $\mathcal{C}_0$, we can construct rate-$1/2$ and rate-$1/4$ RM convolutional codes with memory orders $81$ and $40$, respectively. Their local minimum distances are $31$ and their free distances are lower bounded by $\min\{31, 2 \times 16\} = 31$.

Next, we use the $(255, 93)$ cyclic RM code $\mathcal{C}_0$ as the mother code to form a code chain $\mathcal{C}_0 \supset \mathcal{C}_1 \supset \mathcal{C}_2$ of length $3$ with shifting factor $t$ set to $4$. We find that $\alpha^{55}$ and $\alpha^{59}$ are not roots in $\bg_0(X)$. The minimal polynomials of $\alpha^{55}$ and $\alpha^{59}$ are, respectively,
\[
\mathbf{m}_{55}(X) = 1 + X^4 + X^5 + X^7 + X^8 ~ \text{and} ~ \mathbf{m}_{59}(X) = 1 + X^2 + X^3 + X^6 + X^8.
\]
Using the minimal polynomials $\mathbf{m}_{55}(X)$ and $\mathbf{m}_{59}(X)$, we form the following generator multipliers for the two descendant codes $\mathcal{C}_1$ and $\mathcal{C}_2$ of $\mathcal{C}_0$:
\begin{align*}
\mathbf{f}_1(X) &= \mathbf{m}_{55}(X), \\
\mathbf{f}_2(X) &= \mathbf{m}_{55}(X) \mathbf{m}_{59}(X) \\
       &= 1 + X^2 + X^3 + X^4 + X^5 + X^7 + X^8 + X^9 + X^{10}  + X^{12} + X^{14} + X^{15} + X^{16}.
\end{align*}
With $\mathbf{f}_0(X) = 1$, we can easily check that none of the 3 ratios $\mathbf{f}_1(X)/\mathbf{f}_0(X)$, $\mathbf{f}_2(X)/\mathbf{f}_0(X)$, and $\mathbf{f}_2(X)/\mathbf{f}_1(X)$ is a polynomial in $X^4$.

The $3 \times 4$ GMD-matrix formed by the 4-fold decompositions of the generator multipliers $\mathbf{f}_0(X)$, $\mathbf{f}_1(X)$, and $\mathbf{f}_2(X)$ is
\[
\bF(X) =
\begin{bmatrix}
1 & 0 & 0 & 0 \\
1 + X + X^2 & X & 0 & X \\
1 + X + X^2 + X^3 + X^4 & X + X^2 & 1 + X^2 + X^3 & 1 + X + X^3
\end{bmatrix}.
\]
The first 3 columns of $\bF(X)$ form a $3 \times 3$ submatrix with nonzero determinant. Hence, $\bF(X)$ is a full rank matrix and $\mathbf{f}_0(X)$, $\mathbf{f}_1(X)$, and $\mathbf{f}_2(X)$ satisfy the conditions given in Theorems \ref{Thm1} and \ref{Thm2}.

Using $\bg_0(X)$, $\mathbf{f}_1(X)$, and $\mathbf{f}_2(X)$, we form the generator polynomials of the 2 descendant codes $\mathcal{C}_1$ and $\mathcal{C}_2$ of $\mathcal{C}_0$ as follows:
\begin{align*}
\bg_1(X) &= \bg_0(X) \mathbf{f}_1(X) = \bg_0(X) \mathbf{m}_{55}(X), \\
\bg_2(X) &= \bg_0(X) \mathbf{f}_2(X) = \bg_0(X) \mathbf{m}_{55}(X) \mathbf{m}_{59}(X).
\end{align*}
The two descendant codes $\mathcal{C}_1$ and $\mathcal{C}_2$ of $\mathcal{C}_0$ generated by $\bg_1(X)$ and $\bg_2(X)$ are $(255, 85)$ and $(255, 77)$ codes with dimensions $k_1 = 85$ and $k_2 = 77$, respectively, both with minimum distances at least $31$. The largest CPS of the roots in both $\bg_1(X)$ and $\bg_2(X)$ is $240$. Hence, the minimum distances of the dual codes $\mathcal{C}_{1,h}$ and $\mathcal{C}_{2,h}$ of $\mathcal{C}_1$ and $\mathcal{C}_2$ are lower bounded by $16$.

Since $\bF(X)$ is a 4-fold full rank matrix, the 3 codes $\mathcal{C}_0, \mathcal{C}_1, \mathcal{C}_2$ form a 4-fold full rank RM code chain $\mathcal{C}_0 \supset \mathcal{C}_1 \supset \mathcal{C}_2$ of length $3$. Using the code chain $\mathcal{C}_0 \supset \mathcal{C}_1 \supset \mathcal{C}_2$ with shifting factor $t = 4$, we can construct a rate-$3/4$ $(4, 3, 129)$ RM convolutional code $\mathcal{C}_{\mathrm{convol}}$ of constraint length $129$ with local minimum distance $31$. The rate-$3/4$ convolutional code $\mathcal{C}_{\mathrm{convol}}$ is the direct-sum of its 3 rate-$1/4$ constituent convolutional codes of memory orders $41$, $43$, and $45$, respectively, with local minimum distances at least $31$ and free distances lower bounded by $31$.

With shifting factor $t = 3$, we can show that the subchain $\mathcal{C}_0 \supset \mathcal{C}_1$ of the code chain $\mathcal{C}_0 \supset \mathcal{C}_1 \supset \mathcal{C}_2$ with generator multipliers $\mathbf{f}_0(X) = 1$ and $\mathbf{f}_1(X) = \mathbf{m}_{55}(X)$ given above is a 3-fold full rank code chain. Using this subchain with $t = 3$, we can construct a rate-$2/3$ $(3, 2, 111)$ RM convolutional code $\mathcal{C}_{\mathrm{convol}}(\bg_0, \bg_1)$ of constraint length $111$ with local minimum distance $31$. The rate-$2/3$ RM convolutional code $\mathcal{C}_{\mathrm{convol}}(\bg_0, \bg_1)$ is the direct-sum of its 2 rate-$1/3$ constituent convolutional codes, a $(3, 1, 54)$ convolutional code and a $(3, 1, 57)$ convolutional code, with local minimum distances $31$ and at least $31$, respectively.

Suppose we want to construct a 5-fold full rank code chain with the 3rd order $(255, 93)$ RM code as the mother code to construct a rate-$4/5$ RM convolutional code. This can be done by adding another descendant code $\mathcal{C}_3$ to the code chain $\mathcal{C}_0 \supset \mathcal{C}_1 \supset \mathcal{C}_2$ constructed above as the end code. To achieve this, we choose $\mathbf{f}_3(X) = \mathbf{f}_2(X)(1 + X)$ as the generator multiplier for $\mathcal{C}_3$. Then, the generator polynomial of $\mathcal{C}_3$ is given by
\[
\bg_3(X) = \mathbf{f}_3(X)\bg_0(X) = \bg_0(X) \mathbf{m}_{55}(X) \mathbf{m}_{59}(X) (1 + X).
\]
The code $\mathcal{C}_3$ generated by $\bg_3(X)$ is a $(255, 76)$ code with minimum distance at least $32$ and its dual code is a $(255, 179)$ code with minimum distance at least $15$. It can be shown that the $4 \times 5$ GMD-matrix $\bF(X)$ formed by the 5-fold decompositions of $\mathbf{f}_0(X)$, $\mathbf{f}_1(X)$, $\mathbf{f}_2(X)$, and $\mathbf{f}_3(X)$ is a full rank matrix. Hence, the code chain $\mathcal{C}_0 \supset \mathcal{C}_1 \supset \mathcal{C}_2 \supset \mathcal{C}_3$ of length $4$ is a 5-fold full rank RM code chain. Using this code chain, we can construct a rate-$4/5$ $(5, 4, 139)$ RM convolutional code $\mathcal{C}_{\mathrm{convol}}$ of constraint length $139$ with local minimum distance $31$. It is the direct-sum of its 4 rate-$1/5$ constituent convolutional codes with memory orders $33$, $34$, $36$, and $36$, respectively, each with local minimum distance at least $31$.
\end{example}

\subsection*{Example 8}
\begin{example} \label{eg8} 
Let $m = 10$ and $\mu = (m-2)/2 = 4$. The 4th order even-weight cyclic RM code $\mathcal{C}_{\mathrm{cyc},e}(4, 10)$ is a $(1023, 385)$ code with minimum distance $64$. The dual code of $\mathcal{C}_{\mathrm{cyc},e}(4, 10)$ is the 5th order cyclic RM code $\mathcal{C}_{\mathrm{cyc}}(5, 10)$ which is a $(1023, 638)$ code with minimum distance $31$. Let $\alpha$ be a primitive element in $\mathrm{GF}(2^{10})$. It follows from the weight condition that for $0 \leq h < 2^{10} - 1$, $\alpha^h$ is a root of the generator polynomial $\bg_0(X)$ of $\mathcal{C}_{\mathrm{cyc},e}(4, 10)$ if and only if $0 \leq w(h) \leq 5$. Using this condition, we can determine all the roots of $\bg_0(X)$. We find that the elements $\alpha^{63}$ and $\alpha^{127}$ in $\mathrm{GF}(2^{10})$ are not roots of $\bg_0(X)$. The minimal polynomials of $\alpha^{63}$ and $\alpha^{127}$ are, respectively,
\begin{align*}
\mathbf{m}_{63}(X) &= 1 + X^2 + X^3 + X^5 + X^7 + X^9 + X^{10}, \\
\mathbf{m}_{127}(X) &= 1 + X + X^2 + X^3 + X^4 + X^5 + X^6 + X^7 + X^{10}.
\end{align*}

Suppose we use the 4th order even-weight cyclic RM code $\mathcal{C}_{\mathrm{cyc},e}(4, 10)$ as the mother code $\mathcal{C}_0$ to form an RM code chain $\mathcal{C}_0 \supset \mathcal{C}_1 \supset \mathcal{C}_2$ with shifting factor $t$ set to $4$. Let $\mathbf{f}_1(X) = \mathbf{m}_{63}(X)$ and $\mathbf{f}_2(X) = \mathbf{m}_{63}(X) \mathbf{m}_{127}(X)$ be the generator multipliers for the two descendant codes $\mathcal{C}_1$ and $\mathcal{C}_2$ of $\mathcal{C}_0$. The generator polynomials for $\mathcal{C}_1$ and $\mathcal{C}_2$ are $\bg_1(X) = \bg_0(X)\mathbf{f}_1(X)$ and $\bg_2(X) = \bg_0(X)\mathbf{f}_2(X)$, respectively. The two descendant codes $\mathcal{C}_1$ and $\mathcal{C}_2$ of $\mathcal{C}_0$ generated by $\bg_1(X)$ and $\bg_2(X)$ are $(1023, 375)$ and $(1023, 365)$ codes, respectively, with minimum distances at least $64$.

We find that the 4-fold GMD-matrix $\bF(X)$ formed by the 4-fold decompositions of the generator multipliers $\mathbf{f}_0(X) = 1$, $\mathbf{f}_1(X)$, and $\mathbf{f}_2(X)$ for the code chain $\mathcal{C}_0 \supset \mathcal{C}_1 \supset \mathcal{C}_2$ is a full rank matrix. Hence, the code chain $\mathcal{C}_0 \supset \mathcal{C}_1 \supset \mathcal{C}_2$ is a 4-fold full rank code chain. Using the code chain with shifting factor $t = 4$, we can construct a rate-$3/4$ $(4, 3, 487)$ RM convolutional code $\mathcal{C}_{\mathrm{convol}}$ with constraint length $487$ and local minimum distance $64$. The code $\mathcal{C}_{\mathrm{convol}}$ is the direct-sum of 3 rate-$1/4$ convolutional codes $\mathcal{C}_{\mathrm{convol}}(\bg_0)$, $\mathcal{C}_{\mathrm{convol}}(\bg_1)$, and $\mathcal{C}_{\mathrm{convol}}(\bg_2)$ which are $(4, 1, 160)$, $(4, 1, 162)$, and $(4, 1, 165)$ convolutional codes with local minimum distances of at least $64$ and free distances lower bounded by $\min\{64, 2 \times 31\} = 62$.

Using the mother RM code $\mathcal{C}_0$ in the code chain $\mathcal{C}_0 \supset \mathcal{C}_1 \supset \mathcal{C}_2$ alone with shifting factors $2, 4, 6, 8, 12, 16$, we can construct rate-$1/2$ $(2, 1, 319)$, rate-$1/4$ $(4, 1, 160)$, rate-$1/6$ $(6, 1, 107)$, rate-$1/8$ $(8, 1, 80)$, rate-$1/12$ $(12, 1, 54)$, and rate-$1/16$ $(16, 1, 40)$ RM convolutional codes with local minimum distances lower bounded by $64$ and free distances lower bounded by $62$. If we set shifting factor to $3$, we can construct a rate-$1/3$ $(3, 1, 213)$ convolutional code with local minimum distance at least $64$.

If we remove the root $1$ from the generator polynomials $\bg_0(X)$, $\bg_1(X)$, and $\bg_2(X)$ of the three codes $\mathcal{C}_0$, $\mathcal{C}_1$, and $\mathcal{C}_2$ in the code chain $\mathcal{C}_0 \supset \mathcal{C}_1 \supset \mathcal{C}_2$, we obtain three codes $\mathcal{C}_0^*, \mathcal{C}_1^*, \mathcal{C}_2^*$ which are $(1023, 386)$, $(1023, 376)$, and $(1023, 366)$ codes. The generator polynomials $\bg_0^*(X), \bg_1^*(X), \bg_2^*(X)$ are $\bg_0(X)/(X+1)$, $\bg_1(X)/(X+1)$, and $\bg_2(X)/(X+1)$, respectively. The code $\mathcal{C}_0^*$ is the 4th order RM code $\mathcal{C}_{\mathrm{cyc}}(4, 10)$ with minimum distance $63$. The three codes $\mathcal{C}_0^*, \mathcal{C}_1^*, \mathcal{C}_2^*$ form an RM code chain $\mathcal{C}_0^* \supset \mathcal{C}_1^* \supset \mathcal{C}_2^*$. The generator multipliers for the three codes $\mathcal{C}_0^*, \mathcal{C}_1^*, \mathcal{C}_2^*$ are the same as for $\mathcal{C}_0, \mathcal{C}_1, \mathcal{C}_2$.

Set shifting factor $t = 5$. We find that the RM code chain $\mathcal{C}_0^* \supset \mathcal{C}_1^* \supset \mathcal{C}_2^*$ is a 5-fold full rank code chain. Using the code chain $\mathcal{C}_0^* \supset \mathcal{C}_1^* \supset \mathcal{C}_2^*$, we can construct a rate-$3/5$ $(5, 3, 390)$ RM convolutional code $\mathcal{C}^*_{\mathrm{convol}}$ with constraint length $390$ and local minimum distance $63$.

If we use the subchain $\mathcal{C}_0^* \supset \mathcal{C}_1^*$ of the RM code chain, we can form a rate-$2/5$ $(5, 2, 258)$ RM convolutional code $\mathcal{C}^*_{\mathrm{convol}}(\bg_0^*, \bg_1^*)$ of constraint length $258$ with local minimum distance $63$. If we use the mother code $\mathcal{C}_0^*$ alone, we can construct a rate-$1/5$ $(5, 1, 128)$ RM convolutional code $\mathcal{C}^*_{\mathrm{convol}}(\bg_0^*)$ of memory order $128$ with local minimum distance $63$.
\end{example}

\subsection*{Example 9}
\begin{example} \label{eg9} 
Choose $m = 16$ and $\mu = (m-2)/2 = 7$. The 7th order cyclic RM code is a $(65535, 26333)$ code $\mathcal{C}_{\mathrm{cyc},e}(7, 16)$ with minimum distance $2^9 - 1 = 511$. The dual code of $\mathcal{C}_{\mathrm{cyc}}(7, 16)$ is the 8th order even-weight RM code $\mathcal{C}_{\mathrm{cyc},e}(8, 16)$ which is a $(65535, 39202)$ code with minimum distance $2^8 = 256$. Note that the minimum distance of $\mathcal{C}_{\mathrm{cyc}}(7, 16)$ is about twice that of $\mathcal{C}_{\mathrm{cyc},e}(8, 16)$. The roots of the generator polynomials of $\mathcal{C}_{\mathrm{cyc}}(7, 16)$ and $\mathcal{C}_{\mathrm{cyc},e}(8, 16)$ are elements in the field $\mathrm{GF}(2^{16})$.

Set $t = 2$ and $t = 4$. Using $\mathcal{C}_{\mathrm{cyc}}(7, 16)$, we can construct rate-$1/2$ $(2, 1, 19601)$ and rate-$1/4$ $(4, 1, 9800)$ RM convolutional codes of memory orders $19601$ and $9800$, respectively, with local minimum distance $511$ and free distance lower bounded by $511$.

If we use $\mathcal{C}_{\mathrm{cyc}}(7, 16)$ (or $\mathcal{C}_{\mathrm{cyc},e}(7, 16)$) as the mother code and the minimal polynomials of elements in $\mathrm{GF}(2^{16})$, we can form full rank code chains of various lengths. Using these cyclic RM code chains, we can construct RM convolutional codes of various rates and constraint lengths with local minimum distance $511$ (or $512$).
\end{example}

As pointed out in the beginning of this section, there are also generalized cyclic RM codes~\cite{kasami1968a,weldon1968,delsarte1970}. Using these generalized cyclic RM codes, we can construct even larger classes of RM convolutional codes. RM codes of moderate lengths can be effectively decoded with ordered statistic decoding and iterative multistage MLD algorithms~\cite{fossorier1995,lin2004}.

\section{Construction of Convolutional Codes Based on Finite-Geometry-Based LDPC Codes}\label{sect7:convol_from_PaG}

In this section, we consider construction of convolutional codes using cyclic finite-geometry (FG) codes~\cite{lin2004,kou2001,rudolph1964,weldon1967,smith1967,lin2022}. In general, the minimum distance of an FG code is smaller than the minimum distance of a BCH code of the same length and rate. However, several classes of FG codes have decoding advantages over the BCH codes. FG codes in these classes are cyclic LDPC codes~\cite{lin2004,kou2001} which have distinctive algebraic and geometric structures, and they can be decoded with soft-decision iterative decoding algorithms based on belief propagation such as the sum-product (SP)~\cite{gallager1962,mackay1999} and min-sum (MS)~\cite{chen2002} algorithms and their variations to achieve very good error performance with practical decoder implementation~\cite{lin2004,lin2022}. They can also be decoded with the simple OSML, bit-flipping, and weighted bit-flipping~\cite{kou2001} algorithms.

There are two special categories of finite geometries, namely Euclidean and projective geometries. Based on the finite geometries in these two categories, cyclic codes can be constructed. In this section, we mainly focus on constructions of cyclic LDPC codes based on two-dimensional FGs~\cite{lin2004,lin2022}. Using these two-dimensional FG-LDPC codes, we construct FG-LDPC convolutional codes. Exclusive coverage of construction of general cyclic FG-codes, LDPC or not, can be found in~\cite{lin2004,lin2022}.

\subsection{Two-Dimensional Euclidean Geometries}

The two-dimensional Euclidean geometry (EG) over the field $\mathrm{GF}(2^s)$ with $s \geq 1$, denoted $\mathrm{EG}(2, 2^s)$, consists of $2^{2s}$ points and $2^s(2^s+1)$ lines. Let $\alpha$ be a primitive element in the extension field $\mathrm{GF}(2^{2s})$ of $\mathrm{GF}(2^s)$. The $2^{2s}$ points of $\mathrm{EG}(2, 2^s)$ are represented by the $2^{2s}$ elements,
\[
\alpha^{-\infty} = 0, \alpha^0 = 1, \alpha, \alpha^2, \ldots, \alpha^{2^{2s}-2},
\]
of $\mathrm{GF}(2^{2s})$, and the $2^s(2^s+1)$ lines in $\mathrm{EG}(2, 2^s)$ are represented by $2^s+1$ additive subgroups of order $2^s$ of $\mathrm{GF}(2^{2s})$ and their associated cosets. Each point in $\mathrm{EG}(2, 2^s)$ is intersected by (or lies on) $2^s+1$ lines, and each line in $\mathrm{EG}(2, 2^s)$ has $2^s-1$ lines parallel to it. The $2^s+1$ lines that intersect at a point in $\mathrm{EG}(2, 2^s)$ are said to form an \emph{intersecting bundle}, and each group of $2^s$ parallel lines are said to form a \emph{parallel bundle}. $\mathrm{EG}(2, 2^s)$ consists of $2^s+1$ parallel bundles of lines. The point represented by the zero element $\alpha^{-\infty} = 0$ of $\mathrm{GF}(2^{2s})$ is called the \emph{origin} of $\mathrm{EG}(2, 2^s)$.

Let $\mathrm{EG}^*(2, 2^s)$ be the subgeometry of $\mathrm{EG}(2, 2^s)$ obtained by removing the origin and the $2^s+1$ lines intersecting at it. Hence, $\mathrm{EG}^*(2, 2^s)$ contains the $2^{2s}-1$ non-origin points $\alpha^0 = 1, \alpha, \ldots, \alpha^{2^{2s}-2}$ and $2^{2s}-1$ lines not passing through the origin. Each line in $\mathrm{EG}^*(2, 2^s)$ consists of $2^s$ points, and each point is intersected by $2^s$ lines not passing through the origin. The $2^{2s}-1$ lines in $\mathrm{EG}^*(2, 2^s)$ can be grouped into $2^s+1$ parallel bundles, each consisting of $2^s-1$ parallel lines.

Let $L = \{\alpha^{j_0}, \alpha^{j_1}, \ldots, \alpha^{j_{2^s-1}}\}$ be a line in $\mathrm{EG}^*(2, 2^s)$ which consists of the $2^s$ points $\alpha^{j_0}, \alpha^{j_1}, \ldots, \alpha^{j_{2^s-1}}$. Form a $(2^{2s}-1)$-tuple over $\mathrm{GF}(2)$, $\mathbf{v} = (v_0, v_1, \ldots, v_{2^{2s}-2})$, in which the $j$-th component $v_j$ is set to $1$ if $\alpha^j$ is a point on $L$, otherwise $v_j$ is set to $0$. This $(2^{2s}-1)$-tuple $\mathbf{v}$ is called the \emph{incidence vector} of the line $L$ whose weight is $2^s$. The incidence vectors of lines in $\mathrm{EG}^*(2, 2^s)$ have cyclic structure~\cite{lin2004,kou2001, lin2022}. The cyclic-shift of the incidence vector of a line one position to the right gives the incidence vector of another line. The $2^{2s}-1$ cyclic-shifts of the incidence vector of any line in $\mathrm{EG}^*(2, 2^s)$ give the incidence vectors of all the $2^{2s}-1$ lines in $\mathrm{EG}^*(2, 2^s)$. From this point of view, we may regard $\mathrm{EG}^*(2, 2^s)$ as a two-dimensional cyclic partial Euclidean geometry over $\mathrm{GF}(2^s)$~\cite{lin2022}, a subgeometry of $\mathrm{EG}(2, 2^s)$.

Next, we form a $(2^{2s}-1) \times (2^{2s}-1)$ circulant $\mathbf{H}$ over $\mathrm{GF}(2)$ with the incidence vectors of the lines in $\mathrm{EG}^*(2, 2^s)$ as rows arranged in cyclic manner. This circulant $\mathbf{H}$ can be formed by taking the incidence vector of a line as the first (top) row and then cyclically shifting it $2^{2s}-2$ times to the right to form the rest of the rows. The top row of $\mathbf{H}$ is called the \emph{generator} of $\mathbf{H}$. The columns of $\mathbf{H}$ are labeled with the points $\alpha^0 = 1, \alpha, \ldots, \alpha^{2^{2s}-2}$. Hence, the rows and columns of $\mathbf{H}$ correspond to the lines and points in $\mathrm{EG}^*(2, 2^s)$, respectively. The circulant $\mathbf{H}$ is called the \emph{incidence matrix} of $\mathrm{EG}^*(2, 2^s)$. Both the column and row weights of $\mathbf{H}$ are $2^s$. The rank of the circulant $\mathbf{H}$ is $3^s-1$~\cite{lin2004,kou2001}.

The incidence matrix $\mathbf{H}$ of $\mathrm{EG}^*(2, 2^s)$ has the following property: any two rows (columns) in $\mathbf{H}$ have at most one $1$-component in common. Such property is called the \emph{row and column (RS) constraint (RC-constraint)} structure~\cite{lin2022, ryan2025}. It follows from the interesting and parallel structure of lines in $\mathrm{EG}^*(2, 2^s)$ that for each column position $\alpha^j$, there are $2^s$ rows which have $1$-components at the position $j$. Hence, $\mathbf{H}$ has self-orthogonal structure~\cite{lin2004}. Since $\mathbf{H}$ satisfies the RC-constraint and is self-orthogonal, the Tanner graph~\cite{lin2004,lin2022, ryan2025} associated with $\mathbf{H}$ has girth at least $6$.

\subsection{Cyclic EG-LDPC Codes and Their Associated Convolutional Codes}

The density of the incidence matrix $\mathbf{H}$ of the cyclic subgeometry $\mathrm{EG}^*(2, 2^s)$ of the 2-dimensional Euclidean geometry $\mathrm{EG}(2, 2^s)$ is $2^s/(2^{2s}-1)$. For $s \geq 3$, $\mathbf{H}$ is a low-density (or sparse) matrix. Using $\mathbf{H}$ as a parity-check matrix, the null space of $\mathbf{H}$ over $\mathrm{GF}(2)$ gives a $(2^{2s}-1, 2^{2s}-3^s)$ cyclic low-density parity-check (LDPC) code $\mathcal{C}$ which is called \hl{an EG-LDPC code}. It is a $(2^s, 2^s)$-regular LDPC code.

Let $c$ be a nonnegative integer less than $2^{2s}$. Then, $c$ can be expressed in the following radix-$2^s$ form:
\begin{equation}
c = c_0 + c_1 2^s,
\end{equation}
where $c_0$ and $c_1$ are two nonnegative integers with $0 \leq c_0, c_1 < 2^s$. We define the real sum $c_0 + c_1$ as the $2^s$-weight of $c$, denoted by $W_{2^s}(c)$, i.e.,
\begin{equation}
W_{2^s}(c) = c_0 + c_1.
\end{equation}
For a nonnegative integer $\ell$ with $0 \leq \ell < s$, let $c^{(\ell)}$ be the remainder resulting from dividing $c 2^\ell$ by $2^{2s}-1$. Express $c^{(l)}$ in radix-$2^s$ form:
\begin{equation}
c^{(\ell)} = c_0^{(\ell)} + c_1^{(\ell)} 2^s,
\end{equation}
where $c_0^{(\ell)}$ and $c_1^{(\ell)}$ are two nonnegative integers with $0 \leq c_0^{(\ell)}, c_1^{(\ell)} < 2^s$. The $2^s$-weight of $c^{(\ell)}$ is
\begin{equation}
W_{2^s}(c^{(\ell)}) = c_0^{(\ell)} + c_1^{(\ell)}.
\end{equation}

The generator polynomial $\bg(X)$ of the 2-dimensional cyclic EG-LDPC code $\mathcal{C}$ constructed based on the cyclic Euclidean geometry $\mathrm{EG}^*(2, 2^s)$ has $\alpha^c$ as a root if and only if the following condition holds~\cite{lin2022,lin2004,kou2001}:
\begin{equation} \label{eqn:EG_root_condition}
0 < \max_{0 \leq \ell < s} W_{2^s}(c^{(\ell)}) \leq 2^s - 1.
\end{equation}
Using the condition above, we determine all the roots of $\bg(X)$ and their minimal polynomials. Then, $\bg(X)$ is the least common multiple of the minimal polynomials of the roots of $\bg(X)$. It can be easily checked that the $2^s$-weights of the $2^s$ consecutive integers $1, 2, \ldots, 2^s$ satisfy the condition given by (\ref{eqn:EG_root_condition}). Hence, $\bg(X)$ has $\alpha, \alpha^2, \ldots, \alpha^{2^s}$ as roots. Then, it follows from the BCH-bound that the minimum distance of the 2-dimensional cyclic EG-LDPC code $\mathcal{C}$ is at least $2^s+1$. In fact, the minimum distance of $\mathcal{C}$ is exactly $2^s+1$, which is equal to the number of lines in $\mathrm{EG}^*(2, 2^s)$ that intersect at a point plus one.

The row space of $\mathbf{H}$ gives a $(2^{2s}-1, 3^s-1)$ cyclic code $\mathcal{C}_h$ which is the dual code of the cyclic 2-dimensional EG-LDPC code $\mathcal{C}$. The generator polynomial of $\mathcal{C}_h$ contains $(X+1)$ as a factor. The minimum distance of $\mathcal{C}_h$ is $2^s$, which is the weight of the incidence vector of a line in $\mathrm{EG}^*(2, 2^s)$. Note that the dimension $3^s-1$ of $\mathcal{C}_h$ is even.

Two-dimensional cyclic EG-LDPC codes can be decoded with soft-decision iterative decoding based on belief propagation with either the sum-product algorithm (SPA)~\cite{gallager1962,mackay1999} or the min-sum algorithm (MSA)~\cite{chen2002}, or their variations to achieve error performance close to Shannon capacity. Two-dimensional cyclic EG-LDPC codes have much larger minimum distances compared with the other types of LDPC codes, and their Tanner graphs do not have harmful trapping sets with sizes smaller than their minimum distances, which is a unique feature of 2-dimensional cyclic EG-LDPC codes~\cite{huang2012,diao2013}. Hence, when they are decoded with the SPA or MSA, they can achieve very low (bit and block) error rates without visible error floors. Furthermore, the Tanner graphs of these codes have high degrees of variable-node (VN) connectivity~\cite{kou2001}. Each VN is connected to a large number of other VNs through paths of length 2. In each decoding iteration, each VN receives a large amount of reliability information from other VNs. This enhances the reliability of each VN in each decoding iteration. As a result, decoding converges quickly.

Due to the self-orthogonal structure of $\mathbf{H}$, the 2-dimensional cyclic EG-LDPC code $\mathcal{C}$ constructed based on $\mathrm{EG}(2, 2^s)$ over $\mathrm{GF}(2^s)$ is one-step majority-logic (OSML) decodable~\cite{lin2004}. For each code bit position $\alpha^j$ with $0 \leq j < 2^{2s}-1$, $2^s$ parity-check sums orthogonal on the bit position $\alpha^j$ can be formed by taking the inner products of a received codeword with $2^s$ rows of $\mathbf{H}$. These parity-check sums are called orthogonal check-sums~\cite{lin2004}. This self-orthogonal structure of $\mathbf{H}$ allows the code $\mathcal{C}$ to correct $2^{s-1}$ or fewer random errors over the BSC with OSML-decoding~\cite{lin2004}, which meets the bound on the number of correctable errors guaranteed by the minimum distance $2^s+1$.

Using cyclic 2-dimensional EG-LDPC codes, we can construct EG-LDPC-convolutional codes of various rates and constraint lengths. Consider the special case in which the $(2^{2s} - 1,\, 2^{2s} - 3^s)$ cyclic EG-LDPC code $\mathcal{C}$ is used to construct a rate-$1/2$ convolutional code. Note that the dimension of $\mathcal{C}$ is $k = 2^{2s} - 3^s$, which is an odd integer. Hence, $k - 1$ is divisible by $2$. Set the shifting factor $t = 2$ and the shifting span $b = (2^{2s} - 3^s - 1)/2$. The $(1, 2)$-PCS-local subcode $\mathcal{C}_{\mathrm{local}}$ of the $(2^{2s} - 1,\, 2^{2s} - 3^s)$ cyclic EG-LDPC code $\mathcal{C}$ is a $(2^{2s} - 1,\, (2^{2s} - 3^s + 1)/2)$ code with minimum distance $2^s + 1$. The convolutional code constructed based on $\mathcal{C}$ is a rate-$1/2$ $(2, 1, (3^s - 1)/2)$ convolutional code $\mathcal{C}_{\mathrm{convol}}(\bg)$ with memory order $(3^s - 1)/2$, which consists of an infinite number of copies of the $(1, 2)$-PCS-local subcode $\mathcal{C}_{\mathrm{local}}$ of $\mathcal{C}$ as local codes. The local minimum distance of $\mathcal{C}_{\mathrm{convol}}(\bg)$ is $2^s + 1$. Since the minimum distance of the dual code $\mathcal{C}_h$ of $\mathcal{C}$ is $2^s$, the free distance $d_{\mathrm{free}}$ of $\mathcal{C}_{\mathrm{convol}}(\bg)$ is lower bounded by
\[
d_{\mathrm{free}} \geq \min\{2^s + 1,\, 2 \times 2^s\} = 2^s + 1.
\]

Using the $(2^{2s} - 1,\, 2^{2s} - 3^s)$ two-dimensional cyclic EG-LDPC code $\mathcal{C}$ as the mother code $\mathcal{C}_0$ and minimal polynomials of selected elements in $\mathrm{GF}(2^{2s})$, we can form full rank code chains of various lengths with various shifting factors. Using these code chains, we can construct EG-LDPC convolutional codes of various rates, constraint lengths, and local minimum distances.

\subsection*{Example 10}
\begin{example} \label{eg10}
Let $s = 4$. The Euclidean geometry $\mathrm{EG}(2, 2^4)$ over $\mathrm{GF}(2^8)$ consists of $272$ lines and $256$ points, each line consisting of $16$ points and each point lying on (or intersected by) $17$ lines. Removing the origin from $\mathrm{EG}(2, 2^4)$ and the $17$ lines passing through it, we obtain the cyclic subgeometry $\mathrm{EG}^*(2, 2^4)$ which consists of $255$ points and $255$ lines.

Using the incidence vectors of the lines in $\mathrm{EG}^*(2, 2^4)$, we form the incidence matrix $\mathbf{H}$ of $\mathrm{EG}^*(2, 2^4)$ which is a $255 \times 255$ circulant with both column and row weights $16$. The null space of $\mathbf{H}$ gives a $(255, 175)$ cyclic 2-dimensional EG-LDPC code $\mathcal{C}_0$ with minimum distance $d_0 = 17$ whose Tanner graph has girth at least $6$ and contains no harmful trapping set of size less than $17$. Each variable node (VN) in its Tanner graph is connected to $240$ other VNs by paths of length $2$.

The generator polynomial $\bg_0(X)$ of $\mathcal{C}_0$ is a polynomial over $\mathrm{GF}(2)$ of degree $80$. Let $\alpha$ be a primitive element in the field $\mathrm{GF}(2^8)$. For $0 \leq c < 255$, $\alpha^c$ is a root of $\bg_0(X)$ if and only if $c$ satisfies the following condition:
\begin{equation*}
0 < \max_{0 \leq \ell < 4} W_{16}(c^{(\ell)}) \leq 15.
\end{equation*}
Using this condition, we find that the roots of $\bg_0(X)$ are the $80$ elements in the following $10$ conjugate sets in $\mathrm{GF}(2^8)$:
\begin{align*}
\Omega_1 &= \{\alpha, \alpha^2, \alpha^4, \alpha^8, \alpha^{16}, \alpha^{32}, \alpha^{64}, \alpha^{128}\}, \\
\Omega_3 &= \{\alpha^3, \alpha^6, \alpha^{12}, \alpha^{24}, \alpha^{48}, \alpha^{96}, \alpha^{192}, \alpha^{129}\}, \\
\Omega_5 &= \{\alpha^5, \alpha^{10}, \alpha^{20}, \alpha^{40}, \alpha^{80}, \alpha^{160}, \alpha^{65}, \alpha^{130}\}, \\
\Omega_7 &= \{\alpha^7, \alpha^{14}, \alpha^{28}, \alpha^{56}, \alpha^{112}, \alpha^{224}, \alpha^{193}, \alpha^{131}\}, \\
\Omega_9 &= \{\alpha^9, \alpha^{18}, \alpha^{36}, \alpha^{72}, \alpha^{144}, \alpha^{33}, \alpha^{66}, \alpha^{132}\}, \\
\Omega_{11} &= \{\alpha^{11}, \alpha^{22}, \alpha^{44}, \alpha^{88}, \alpha^{176}, \alpha^{97}, \alpha^{194}, \alpha^{133}\}, \\
\Omega_{13} &= \{\alpha^{13}, \alpha^{26}, \alpha^{52}, \alpha^{104}, \alpha^{208}, \alpha^{161}, \alpha^{67}, \alpha^{134}\}, \\
\Omega_{15} &= \{\alpha^{15}, \alpha^{30}, \alpha^{60}, \alpha^{120}, \alpha^{240}, \alpha^{225}, \alpha^{195}, \alpha^{135}\}, \\
\Omega_{37} &= \{\alpha^{37}, \alpha^{74}, \alpha^{148}, \alpha^{41}, \alpha^{82}, \alpha^{164}, \alpha^{73}, \alpha^{146}\}, \\
\Omega_{45} &= \{\alpha^{45}, \alpha^{90}, \alpha^{180}, \alpha^{105}, \alpha^{210}, \alpha^{165}, \alpha^{75}, \alpha^{150}\}.
\end{align*}

The element with the largest power in the above conjugate sets is $\alpha^{240}$. The dual code $\mathcal{C}_h$ of $\mathcal{C}_0$ is the row space of $\mathbf{H}$. It is a $(255, 80)$ cyclic code whose generator polynomial has $1, \alpha, \alpha^2, \ldots, \alpha^{14}$ consecutive powers of $\alpha$ as roots. The BCH-bound on the minimum distance $d_h$ of $\mathcal{C}_h$ is $16$. Since the row weight of $\mathbf{H}$ is $16$, the true minimum distance $d_h$ of $\mathcal{C}_h$ is $16$.

Set the shifting factor $t = 2$. Using the $(255, 175)$ cyclic 2-dimensional EG-LDPC code, we can construct a rate-$1/2$ $(2, 1, 40)$ EG-LDPC-convolutional code of memory order $40$ with local minimum distance $17$ and free distance lower bounded by $\min\{17, 2 \times 16\} = 17$. The code is a self-orthogonal code and is capable of correcting up to $8$ random errors over the binary symmetric channel (BSC). The code is much better than the code $(2, 1, 179)$ self-orthogonal code with random error correcting capability $8$ listed in \cite[Table 13.2(a)]{lin2004}.

If we set the shifting factor $t = 4$, we can construct a rate-$1/4$ $(4, 1, 20)$ EG-LDPC-convolutional code of memory order $20$ with local minimum distance $17$ and free distance lower bounded by $17$. If we set $t = 6$, we can construct a rate-$1/6$ $(6, 1, 14)$ EG-LDPC-convolutional code of memory order $14$ with local minimum distance $17$ and free distance lower bounded by $17$. Both rate-$1/4$ and $1/6$ codes are self-orthogonal codes and can correct up to $8$ random errors over the BSC. No self-orthogonal convolutional codes with such correcting capability are found in literature for comparison.

If we set $t = 3$, we can construct a rate-$1/3$ $(3, 1, 27)$ with local minimum distance $17$. With OSML decoding, the code is also capable of correcting up to $8$ random errors over the BSC and is better than the rate-$1/3$ $(3, 1, 35)$ orthogonalizable code listed in~\cite[Table 13.3(b)]{lin2004}.
\end{example}

The developments above show that rate-$1/t$ self-orthogonal convolutional codes can be effectively constructed based on a cyclic two-dimensional EG-LDPC code.

\subsection*{Example 11}
\begin{example} \label{eg11}
In the following, we continue Example \ref{eg10} to construct EG-LDPC convolutional codes with rates $3/4$, $4/5$, $3/5$, $2/5$, and $1/5$ using the $(255, 175)$ EG-LDPC code $\mathcal{C}_0$ as the mother code.

To construct a rate-$3/4$ convolutional code, we need to construct a 4-fold full rank code chain $\mathcal{C}_0 \supset \mathcal{C}_1 \supset \mathcal{C}_2$ of length $r = 3$. Based on the condition on roots of the generator polynomial $\bg_0(X)$ of $\mathcal{C}_0$, we find that $\alpha^{85}$ and $\alpha^{119}$ are not roots of $\bg_0(X)$. The minimal polynomials of $\alpha^{85}$ and $\alpha^{119}$ are $\mathbf{m}_{85}(X) = 1 + X + X^2$ and $\mathbf{m}_{119}(X) = 1 + X^3 + X^4$, respectively. Next, we form the following generator multipliers for the two descendant codes $\mathcal{C}_1$ and $\mathcal{C}_2$ of $\mathcal{C}_0$:
\begin{align*}
\mathbf{f}_1(X) &= \mathbf{m}_{85}(X) = 1 + X + X^2, \\
\mathbf{f}_2(X) &= \mathbf{m}_{85}(X) \mathbf{m}_{119}(X) = 1 + X + X^2 + X^3 + X^6.
\end{align*}
The three generator multipliers $\mathbf{f}_0(X) = 1$, $\mathbf{f}_1(X)$, and $\mathbf{f}_2(X)$ satisfy the ratio condition given by Theorem \ref{Thm1}. The 4-fold GMD matrix formed by the 4-fold decompositions of $\mathbf{f}_0(X) = 1$, $\mathbf{f}_1(X)$, and $\mathbf{f}_2(X)$ is
\[
\bF(X) =
\begin{bmatrix}
1 & 0 & 0 & 0 \\
1 & 1 & 1 & 0 \\
1 & 1 & 1 + X & 1
\end{bmatrix}
\]
which is a full rank matrix.

Using $\bg_0(X)$ and $\mathbf{f}_1(X)$ and $\mathbf{f}_2(X)$, we form the generator polynomials of the two descendant codes $\mathcal{C}_1$ and $\mathcal{C}_2$ of the mother code $\mathcal{C}_0$ as follows: $\bg_1(X) = \mathbf{f}_1(X)\bg_0(X)$ and $\bg_2(X) = \mathbf{f}_2(X)\bg_0(X)$. The two cyclic codes $\mathcal{C}_1$ and $\mathcal{C}_2$ generated by $\bg_1(X)$ and $\bg_2(X)$ are $(255, 173)$ and $(255, 169)$ subcodes of $\mathcal{C}_0$ with minimum distances at least $17$. The largest powers among the roots of $\bg_1(X)$ and $\bg_2(X)$ are still $240$. Hence, the minimum distances of the dual codes of $\mathcal{C}_1$ and $\mathcal{C}_2$ are still lower bounded by $16$.

The three codes $\mathcal{C}_0, \mathcal{C}_1, \mathcal{C}_2$ form a full rank code chain $\mathcal{C}_0 \supset \mathcal{C}_1 \supset \mathcal{C}_2$ of length $3$. Since the 4-fold GMD-matrix for the code chain is a full rank matrix, the code chain $\mathcal{C}_0 \supset \mathcal{C}_1 \supset \mathcal{C}_2$ is a 4-fold full rank code chain. Set the shifting span $b = (169-1)/4 = 42$. The $(3, 4)$-PCS-local subcode of the mother code $\mathcal{C}_0$ formed by the code chain is a $(255, 129)$ code with minimum distance $17$.

Using the code chain $\mathcal{C}_0 \supset \mathcal{C}_1 \supset \mathcal{C}_2$, we can construct a rate-$3/4$ $(4, 3, 63)$ EG-LDPC convolutional code $\mathcal{C}_{\mathrm{convol}}$ of constraint length $63$ with local minimum distance $17$. $\mathcal{C}_{\mathrm{convol}}$ is the direct-sum of its 3 rate-$1/4$ constituent convolutional codes of memory orders $20$, $21$, and $22$, respectively, all with local minimum distances of $17$ and free distances lower bounded by $17$.

To construct a rate-$4/5$ convolutional code, we add a code $\mathcal{C}_3$ to the code chain $\mathcal{C}_0 \supset \mathcal{C}_1 \supset \mathcal{C}_2$ as the end code. The generator multiplier for $\mathcal{C}_3$ is designed as
\[
\mathbf{f}_3(X) = \mathbf{m}_{85}(X) \mathbf{m}_{119}(X) \mathbf{m}_{87}(X) = 1 + X^4 + X^5 + X^7 + X^8 + X^{13} + X^{14}
\]
where $\mathbf{m}_{87}(X) = 1 + X + X^5 + X^7 + X^8$ is the minimal polynomial of the element $\alpha^{87}$ in $\mathrm{GF}(2^8)$ which is not a root of $\bg_0(X)$. The generator polynomial of $\mathcal{C}_3$ is $\bg_3(X) = \mathbf{f}_3(X)\bg_0(X)$ and the code $\mathcal{C}_3$ generated by $\bg_3(X)$ is a $(255, 161)$ code with minimum distance at least $17$. The largest power of $\bg_3(X)$ is still $240$. Hence, the minimum distance of the dual code of $\mathcal{C}_3$ is lower bounded by $16$.

The 4 codes $\mathcal{C}_0, \mathcal{C}_1, \mathcal{C}_2, \mathcal{C}_3$ form a code chain $\mathcal{C}_0 \supset \mathcal{C}_1 \supset \mathcal{C}_2 \supset \mathcal{C}_3$ of length $4$. The 4 generator multipliers $\mathbf{f}_0(X) = 1, \mathbf{f}_1(X), \mathbf{f}_2(X), \mathbf{f}_3(X)$ satisfy the ratio condition given by Theorem \ref{Thm1}. The 5-fold GMD matrix formed by the 5-fold decompositions of $\mathbf{f}_0(X), \mathbf{f}_1(X), \mathbf{f}_2(X), \mathbf{f}_3(X)$ is
\[
\bF(X) =
\begin{bmatrix}
1 & 0 & 0 & 0 & 0 \\
1 & 1 & 1 & 0 & 0 \\
1 & 1 + X & 1 & 1 + X & 0 \\
1 + X & 0 & X & X + X^2 & 1 + X
\end{bmatrix},
\]
which is a full rank matrix. Hence, the 5-fold code chain $\mathcal{C}_0 \supset \mathcal{C}_1 \supset \mathcal{C}_2 \supset \mathcal{C}_3$ is a 5-fold full rank code chain.

Note that the dimension $k_3$ is $161$ and $k_3 - 1 = 160$ which is divisible by $5$. Set shifting span $b = (161-1)/5 = 32$. The $(4, 5)$-PCS-local subcode of $\mathcal{C}_0$ is a $(255, 132)$ code with minimum distance $17$. Using the code chain $\mathcal{C}_0 \supset \mathcal{C}_1 \supset \mathcal{C}_2 \supset \mathcal{C}_3$, we can construct a rate-$4/5$ $(5, 4, 70)$ EG-LDPC convolutional code $\mathcal{C}_{\mathrm{convol}}(\bg_0, \bg_1, \bg_2, \bg_3)$ of constraint length $70$ with local minimum distance $17$. The convolutional code $\mathcal{C}_{\mathrm{convol}}(\bg_0, \bg_1, \bg_2, \bg_3)$ is the direct-sum of its 4 rate-$1/5$ constituent convolutional codes of memory orders $16$, $17$, $18$, and $19$, respectively, all with local minimum distance at least $17$.

Suppose we use the subchain $\mathcal{C}_1 \supset \mathcal{C}_2 \supset \mathcal{C}_3$ of the code chain $\mathcal{C}_0 \supset \mathcal{C}_1 \supset \mathcal{C}_2 \supset \mathcal{C}_3$ for convolutional code construction. The generator polynomial of the mother code $\mathcal{C}_1$ of the subchain is $\bg_1(X) = \bg_0(X)\mathbf{f}_1(X)$. Using the generator multipliers $\mathbf{f}_2(X)/\mathbf{f}_1(X) = \mathbf{m}_{119}(X)$ and $\mathbf{f}_3(X)/\mathbf{f}_1(X) = \mathbf{m}_{119}(X) \mathbf{m}_{87}(X)$ for $\mathcal{C}_2$ and $\mathcal{C}_3$, respectively, we can construct a rate-$3/5$ $(5, 3, 54)$ convolutional code $\mathcal{C}_{\mathrm{convol}}(\bg_1, \bg_2, \bg_3)$ of constraint length $54$ with local minimum distance at least $17$. If we use the subchain $\mathcal{C}_2 \supset \mathcal{C}_3$ of the code chain, the generator polynomial $\bg_2(X) = \bg_0(X) \mathbf{m}_{85}(X) \mathbf{m}_{119}(X)$ of $\mathcal{C}_2$, and the generator multiplier $\mathbf{f}_3(X)/\mathbf{f}_2(X) = \mathbf{m}_{87}(X)$ for $\mathcal{C}_3$, we can construct a rate-$2/5$ $(5, 2, 37)$ convolutional code $\mathcal{C}_{\mathrm{convol}}(\bg_2, \bg_3)$ of constraint length $37$ with local minimum distance at least $17$. Using $\mathcal{C}_3$ alone, we can construct a rate-$1/5$ $(5, 1, 19)$ convolutional code $\mathcal{C}_{\mathrm{convol}}(\bg_3)$ of memory order $19$ with local minimum distance at least $17$.
\end{example}

\subsection*{Example 12}
\begin{example} \label{eg12}
Let $s = 6$. The 2-dimensional Euclidean geometry $\mathrm{EG}(2, 2^6)$ over $\mathrm{GF}(2^6)$ consists of $4160$ lines and $4096$ points, each line consisting of $64$ points and each point lying on $65$ lines. If we remove the origin from $\mathrm{EG}(2, 2^6)$ and the $65$ lines passing through it, we obtain the cyclic subgeometry $\mathrm{EG}^*(2, 2^6)$ which consists of $4095$ points and $4095$ lines.

Using the incidence vectors of the lines in $\mathrm{EG}^*(2, 2^6)$, we form the incidence matrix $\mathbf{H}$ of $\mathrm{EG}^*(2, 2^6)$ which is a $4095 \times 4095$ circulant with both column and row weights $64$. The null space of $\mathbf{H}$ gives a $(4095, 3367)$ cyclic 2-dimensional EG-LDPC code $\mathcal{C}_0$ with minimum distance $65$ whose Tanner graph contains no harmful trapping sets of size less than $65$~\cite{huang2012,diao2013}. Each variable node (VN) in its Tanner graph is connected to $4032$ other VNs by paths of length $2$. This high degree of VN connectivity allows all the VNs to share information in each decoding iteration and enhance their reliabilities using either the SPA or the MSA iterative decoding algorithm. The decoding can converge very quickly.

The generator polynomial $\bg_0(X)$ of $\mathcal{C}_0$ is a polynomial over $\mathrm{GF}(2)$ of degree $728$. Let $\alpha$ be a primitive element in the extension field $\mathrm{GF}(2^{12})$ of $\mathrm{GF}(2^6)$. For $0 \leq c < 4096$, $\alpha^c$ is a root of $\bg_0(X)$ if and only if $c$ satisfies the following condition:
\begin{equation*}
0 < \max_{0 \leq \ell < 6} W_{64}(c^{(\ell)}) \leq 63.
\end{equation*}
The dual code $\mathcal{C}_h$ of $\mathcal{C}_0$ is a $(4095, 728)$ cyclic EG-LDPC code which is the row space of $\mathbf{H}$. The minimum distance of $\mathcal{C}_h$ is $64$.

The bit error rate (BER) and block error rate (BLER) performances of the code over the AWGN channel using BPSK modulation decoded based on $\mathbf{H}$ with $5$, $10$, and $50$ iterations of the MSA scaled by a factor of $0.625$ are shown in Fig. \ref{fig:ch9_3_awgn_bec_a} ~\cite[Figure 12.3]{lin2022}. At the BER $10^{-8}$ with $50$ iterations of the MSA, the code performs $1.82$ dB from the Shannon limit without visible error-floor. Due to the large VN-connectivity of the Tanner graph of the code, the MSA decoding of the code converges quickly as shown in Fig. \ref{fig:ch9_3_awgn_bec_a}. Since the parity-check matrix $\mathbf{H}$ of $\mathcal{C}_0$ has self-orthogonal structure, it can be decoded with OSML-decoding to correct up to $32$ random errors over the BSC and up to $64$ random erasures over the binary erasure channel (BEC)~\cite{lin2022}. The unrecovered erasure error rate (UEBR) and unrecovered erasure block error rate (UEBLR) performances over BEC decoded with the SPA are shown in Fig. \ref{fig:ch9_3_awgn_bec_b}.

\begin{figure}
\centering
\subfigure[]
{
\includegraphics[height=3.0in,width=0.8\textwidth]{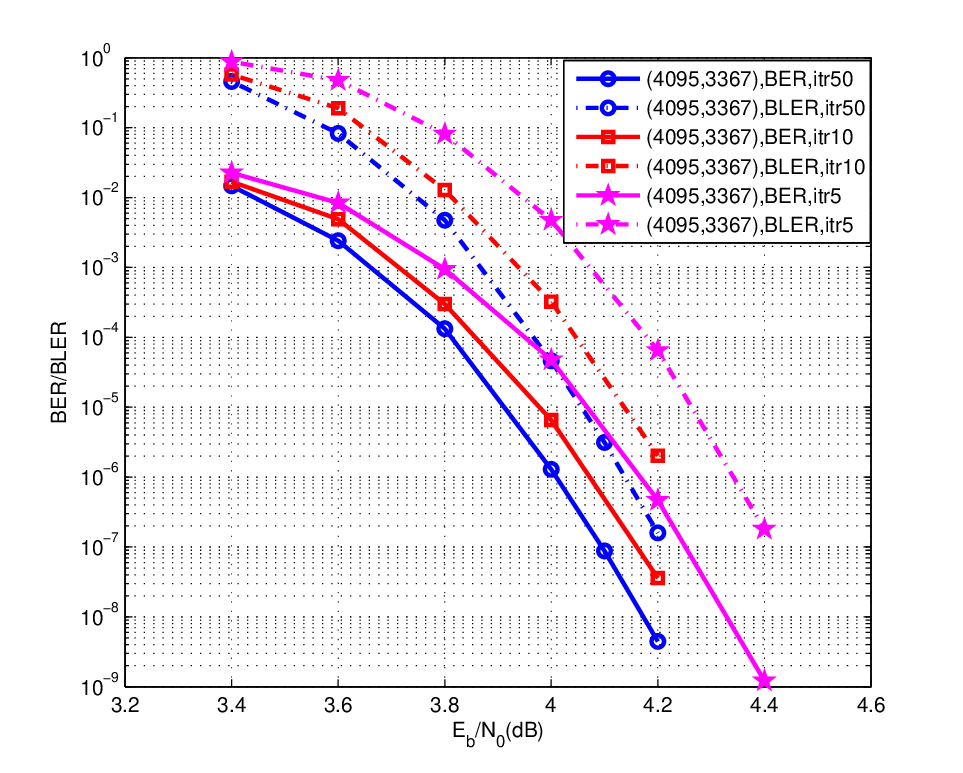}
\label{fig:ch9_3_awgn_bec_a}
}
\subfigure[]
{
\includegraphics[height=3.0in,width=0.8\textwidth]{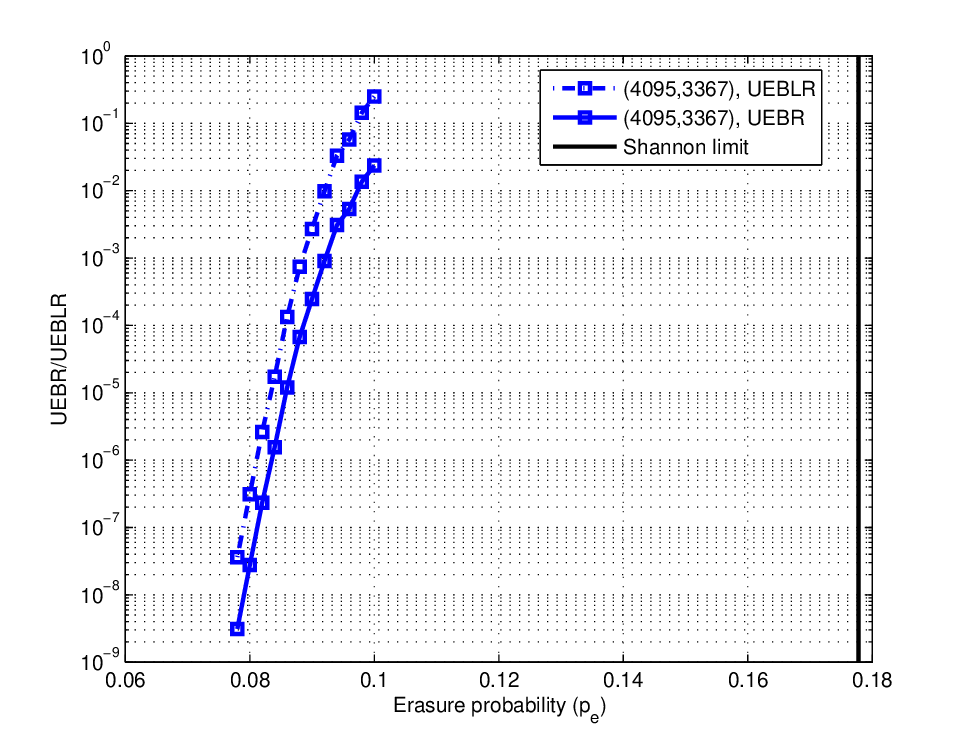}
\label{fig:ch9_3_awgn_bec_b}
}
\caption{The error performances of the $(4095, 3367)$ cyclic 2-dimensional EG-LDPC code given in Example \ref{eg12} over: (a) AWGN channel; and (b) BEC.}
 \label{fig:ch9_3_awgn_bec}
\end{figure} 

Set shifting factor $t = 2$ and shifting span $b = (3367-1)/2 = 1683$. The $(1, 2)$-PCS-local subcode $\mathcal{C}_{\mathrm{local}}$ of $\mathcal{C}_0$ is a $(4095, 1684)$ EG-LDPC code with minimum distance $65$. Using $\mathcal{C}_0$, we can construct a rate-$1/2$ $(2, 1, 364)$ EG-LDPC convolutional code $\mathcal{C}_{\mathrm{convol}}(\bg_0)$ of memory order $364$ with local minimum distance $65$ and free distance lower bounded by $\min\{65, 2 \times 64\} = 65$. It is a self-orthogonal convolutional code which can correct up to $32$ random errors over the BSC with OSML-decoding. It is far better than the last self-orthogonal convolutional code listed in \cite[Table 13.2]{lin2004} which is a $(2, 1, 425)$ convolutional code with error correcting capability only $12$.

Suppose we set the shifting factor $t = 3$ and shifting span $b = (3367-1)/3 = 1122$. Using the $(4095, 3367)$ EG-LDPC code, we can construct a rate-$1/3$ $(3, 1, 243)$ EG-LDPC convolutional code $\mathcal{C}_{\mathrm{convol}}(\bg_0)$ of memory order $243$ with local minimum distance $65$. It is also a self-orthogonal convolutional code.
\end{example}

\subsection*{Example 13} 
\begin{example} \label{eg13}
Continue Example \ref{eg12}. In the following, we form an EG-LDPC code chain $\mathcal{C}_0 \supset \mathcal{C}_1 \supset \mathcal{C}_2 \supset \mathcal{C}_3$ of length $4$ using the $(4095, 3367)$ 2-dimensional EG-LDPC code $\mathcal{C}_0$ constructed in Example \ref{eg12} as the mother code. Using the code chain and the shifting factor $t = 5$, we can construct convolutional codes of rates $4/5$, $3/5$, $2/5$, and $1/5$.

First, we find that the elements $\alpha^{127}$, $\alpha^{255}$, and $\alpha^{511}$ in $\mathrm{GF}(2^{12})$ are not roots in the generator polynomial $\bg_0(X)$ of $\mathcal{C}_0$. Next, we choose the following generator multipliers for the descendant codes $\mathcal{C}_1$, $\mathcal{C}_2$, and $\mathcal{C}_3$ of $\mathcal{C}_0$:
\begin{align*}
\mathbf{f}_1(X) &= \mathbf{m}_{127}(X), \\
\mathbf{f}_2(X) &= \mathbf{m}_{127}(X) \mathbf{m}_{255}(X), \\
\mathbf{f}_3(X) &= \mathbf{m}_{127}(X) \mathbf{m}_{255}(X) \mathbf{m}_{511}(X),
\end{align*}
where $\mathbf{m}_{127}(X)$, $\mathbf{m}_{255}(X)$, and $\mathbf{m}_{511}(X)$ are minimal polynomials of the elements $\alpha^{127}$, $\alpha^{255}$, and $\alpha^{511}$ in $\mathrm{GF}(2^{12})$. The generator multipliers $\mathbf{f}_0(X) = 1$, $\mathbf{f}_1(X)$, $\mathbf{f}_2(X)$, and $\mathbf{f}_3(X)$ for the codes $\mathcal{C}_0$, $\mathcal{C}_1$, $\mathcal{C}_2$, and $\mathcal{C}_3$ satisfy the ratio condition given by Theorem \ref{Thm1}.

With the chosen generator multipliers, the generator polynomials of the descendant codes $\mathcal{C}_1$, $\mathcal{C}_2$, and $\mathcal{C}_3$ of $\mathcal{C}_0$ are $\bg_1(X) = \mathbf{f}_1(X)\bg_0(X)$, $\bg_2(X) = \mathbf{f}_2(X)\bg_0(X)$, and $\bg_3(X) = \mathbf{f}_3(X)\bg_0(X)$, respectively. The 5-fold GMD-matrix formed by the 5-fold decompositions of the multipliers $\mathbf{f}_0(X)$, $\mathbf{f}_1(X)$, $\mathbf{f}_2(X)$, and $\mathbf{f}_3(X)$ is a full rank matrix. Hence, the code chain $\mathcal{C}_0 \supset \mathcal{C}_1 \supset \mathcal{C}_2 \supset \mathcal{C}_3$ is a 5-fold full rank EG-LDPC code chain. The codes $\mathcal{C}_1$, $\mathcal{C}_2$, and $\mathcal{C}_3$ are $(4095, 3355)$, $(4095, 3343)$, and $(4095, 3331)$ codes, respectively, all with minimum distances at least $65$. Set the shifting span $b = (3331-1)/5 = 666$. The $(4, 5)$-PCS-local subcode $\mathcal{C}_{\mathrm{local}}$ of $\mathcal{C}_0$ is a $(4095, 2668)$ code with local minimum distance at least $65$.

Using the code chain $\mathcal{C}_0 \supset \mathcal{C}_1 \supset \mathcal{C}_2 \supset \mathcal{C}_3$ and the 5-fold decompositions of the generator polynomials $\bg_0(X)$, $\bg_1(X)$, $\bg_2(X)$, and $g_3(X)$, we can construct a rate-$4/5$ $(5, 4, 598)$ EG-LDPC convolutional code $\mathcal{C}_{\mathrm{convol}}(\bg_0, \bg_1, \bg_2, \bg_3)$ of constraint length $598$ with local minimum distance at least $65$. It consists of an infinite number of copies of the $(4095, 2668)$ $(4, 5)$-PCS-local subcode $\mathcal{C}_{\mathrm{local}}$ of $\mathcal{C}_0$ as local codes. The rate-$4/5$ convolutional code $\mathcal{C}_{\mathrm{convol}}(\bg_0, \bg_1, \bg_2, \bg_3)$ is the direct-sum of its 4 rate-$1/5$ constituent codes $\mathcal{C}_{\mathrm{convol}}(\bg_0)$, $\mathcal{C}_{\mathrm{convol}}(\bg_1)$, $\mathcal{C}_{\mathrm{convol}}(\bg_2)$, and $\mathcal{C}_{\mathrm{convol}}(\bg_3)$ which are $(5, 1, 146)$, $(5, 1, 148)$, $(5, 1, 151)$, and $(5, 1, 153)$ convolutional codes all with local minimum distances at least $65$.

Using the subchains $\mathcal{C}_1 \supset \mathcal{C}_2 \supset \mathcal{C}_3$, $\mathcal{C}_2 \supset \mathcal{C}_3$, and $\mathcal{C}_3$ of the code chain $\mathcal{C}_0 \supset \mathcal{C}_1 \supset \mathcal{C}_2 \supset \mathcal{C}_3$, we can construct rates $3/5$, $2/5$, and $1/5$ convolutional codes of constraint lengths $452$, $304$, and $153$, respectively, all with local minimum distance at least $65$.
\end{example}

Note that a cyclic EG-LDPC code can be constructed by using the lines of multi-dimensional Euclidean geometries~\cite{lin2004,kou2001}.

\subsection{Cyclic LDPC and Convolutional Codes Constructed Based on Two-dimensional Projective Geometries}

In the following, we present a class of cyclic LDPC codes constructed based on 2-dimensional projective geometries (PGs) over fields of characteristic $2$~\cite{lin2004,kou2001,rudolph1964,weldon1967,smith1967,lin2022}. Using these PG-LDPC codes, convolutional codes can be constructed. PG-LDPC-convolutional codes are self-orthogonal.

Let $s$ be a positive integer, $\mathrm{GF}(2^{3s})$ be the extension field of the field $\mathrm{GF}(2^s)$, and $\alpha$ be a primitive element in $\mathrm{GF}(2^{3s})$. The order of $\alpha$ is $2^{3s} - 1$ which can be factored as the product of $2^s - 1$ and $2^{2s} + 2^s + 1$. Let
\begin{equation}
n = 2^{2s} + 2^s + 1.
\end{equation}
The 2-dimensional projective geometry over the field $\mathrm{GF}(2^s)$, denoted by $\mathrm{PG}(2, 2^s)$, consists of $n$ points and $n$ lines. Each line in $\mathrm{PG}(2, 2^s)$ consists of $2^s + 1$ points, and each point in $\mathrm{PG}(2, 2^s)$ is intersected by (or on) $2^s + 1$ lines. The $n$ points in $\mathrm{PG}(2, 2^s)$ are represented by the $n$ elements $\alpha^0 = 1, \alpha^1, \alpha^2, \ldots, \alpha^{n-1}$ in $\mathrm{GF}(2^{3s})$. The incidence vector of a line $L$ in $\mathrm{PG}(2, 2^s)$ is an $n$-tuple $\mathbf{v} = (v_0, v_1, \ldots, v_{n-1})$ over $\mathrm{GF}(2)$ with coordinates corresponding to the points $\alpha^0 = 1, \alpha^1, \ldots, \alpha^{n-1}$ of $\mathrm{PG}(2, 2^s)$, where $v_j$ is $1$ if $\alpha^j$ is a point on $L$, otherwise $v_j$ is $0$ for $0 \leq j < n$. The incidence vectors of lines in $\mathrm{PG}(2, 2^s)$ have cyclic structure. The $n$ cyclic-shifts of the incidence vector of a line (any line) in $\mathrm{PG}(2, 2^s)$ give the incidence vectors of all the $n$ lines in $\mathrm{PG}(2, 2^s)$. The weight of the incidence vector of a line is $2^s + 1$.

Let $\mathbf{H}$ be the $n \times n$ circulant over $\mathrm{GF}(2)$ formed by the incidence vectors of the $n$ lines in $\mathrm{PG}(2, 2^s)$ with cyclic arrangement. The matrix $\mathbf{H}$ is the incidence matrix of $\mathrm{PG}(2, 2^s)$ with both column and row weights $2^s + 1$ which satisfies the RC-constraint and has orthogonal structure. The rank of $\mathbf{H}$ is $3^s + 1$. The null space of $\mathbf{H}$ over $\mathrm{GF}(2)$ gives an $(n, n - 3^s - 1)$ two-dimensional cyclic PG-LDPC code $\mathcal{C}$ with minimum distance at least $2^s + 2$. The Tanner graph of $\mathcal{C}$ has girth $6$ and contains no harmful trapping sets with size smaller than $2^s + 1$~\cite{huang2012,diao2013}. For large $s$, decoded iteratively with the SPA or MSA, $\mathcal{C}$ can achieve a very low error rate without error-floor and perform close to the Shannon limit~\cite{lin2004,lin2022}. Furthermore, $\mathcal{C}$ is a self-orthogonal code and can be decoded with OSML-decoding based on the incidence matrix $\mathbf{H}$ of $\mathrm{PG}(2, 2^s)$ to correct $2^{s-1}$ or fewer random errors over the BSC.

Let $h$ be a non-negative integer less than $2^{3s} - 1$. For a nonnegative integer $\ell$ with $0 \leq \ell < s$, let $h^{(\ell)}$ be the remainder resulting from dividing $h 2^\ell$ by $2^{3s} - 1$. Express $h^{(\ell)}$ in radix-$2^s$ form,
\begin{equation}
h^{(\ell)} = h_0^{(\ell)} + h_1^{(\ell)} 2^s + h_2^{(\ell)} 2^{2s},
\end{equation}
where $0 \leq h_i^{(\ell)} < 2^s$ and $0 \leq i \leq 2$. The $2^s$-weight of $h^{(\ell)}$ is
\begin{equation}
W_{2^s}(h^{(\ell)}) = h_0^{(\ell)} + h_1^{(\ell)} + h_2^{(\ell)}.
\end{equation}
The generator polynomial $\bg(X)$ of the cyclic 2-dimensional PG-LDPC code $\mathcal{C}$ has $\alpha^h$ as a root if and only if $h$ is divisible by $2^s - 1$ and the following condition holds~\cite{lin2004,kou2001, lin2022}:
\begin{equation} \label{eq:PG_root_condition}
0 < \max_{0 \leq \ell < s} W_{2^s}(h^{(\ell)}) = t(2^s - 1)
\end{equation}
with $0 \leq t \leq 2$. Since $h$ is less than $2^{3s} - 1$ and divisible by $2^s - 1$, $h$ must be a product of $j$ and $2^s - 1$ with $0 \leq j < n$, i.e., $h = j(2^s - 1)$.

Let $\beta = \alpha^{2^s - 1}$. Then, the order of $\beta$ is $n$, i.e., $\beta^n = 1$. The set
\[
\Lambda = \{\beta^0 = 1, \beta, \beta^2, \ldots, \beta^{n-1}\}
\]
of $n$ elements in $\mathrm{GF}(2^{3s})$ forms a cyclic subgroup of the multiplicative group of $\mathrm{GF}(2^{3s})$. It follows from the necessary and sufficient condition on a root of $\bg(X)$ given by (\ref{eq:PG_root_condition}) that a root of $\bg(X)$ must be an element in $\Lambda$. Hence, $\bg(X)$ is a factor of $X^n + 1$. Based on the condition given by (\ref{eq:PG_root_condition}), we find that $\bg(X)$ has the following $2^s + 1$ consecutive powers of $\beta$, $\beta^0 = 1, \beta, \beta^2, \ldots, \beta^{2^s}$ as roots and $\beta^{2^s + 1}$ is not a root. Then, it follows from the BCH-bound that the minimum distance of the $(n, n - 3^s - 1)$ cyclic two-dimensional PG-LDPC code $\mathcal{C}$ is at least $2^s + 2$. The true minimum distance is exactly $2^s + 2$~\cite{lin2004}. Using the condition given by (\ref{eq:PG_root_condition}), we find all the roots of $\bg(X)$. Then, $\bg(X)$ is the LCM of the minimal polynomials of the roots of $\bg(X)$ which has degree $3^s + 1$. The row space of the parity-check matrix $\mathbf{H}$ gives an $(n, 3^s + 1)$ cyclic code which is the dual code $\mathcal{C}_h$ of $\mathcal{C}$. The minimum distance $d_h$ of $\mathcal{C}_h$ is $2^s + 1$, which is the weight of the incidence vector of a line in $\mathrm{PG}(2, 2^s)$, i.e., the number of points on a line in $\mathrm{PG}(2, 2^s)$.

Using a two-dimensional PG-LDPC code $\mathcal{C}$ as the mother code in a code chain, we can construct a PG-LDPC convolutional code. For an even shifting factor $t$, the rate-$1/t$ convolutional code constructed based on the 2-dimensional PG-LDPC code is a $(t, 1, \lceil (3^s + 1)/t \rceil)$ convolutional code of memory order upper bounded by $\lceil (3^s + 1)/t \rceil$ with local minimum distance $2^s + 2$ and free distance lower bounded by $\min\{2^s + 2, 2(2^s + 1)\} = 2^s + 2$.

\subsection*{Example 14}
\begin{example} \label{eg14}
Set $s = 5$ and $n = (2^{15} - 1) / (2^5 - 1) = 1057$. The two-dimensional projective geometry $\mathrm{PG}(2, 2^5)$ over $\mathrm{GF}(2^5)$ consists of $n = 1057$ points and $1057$ lines. Each line in $\mathrm{PG}(2, 2^5)$ consists of $33$ points, and each point is intersected by $33$ lines. The incidence matrix $\mathbf{H}$ of $\mathrm{PG}(2, 2^5)$ is a $1057 \times 1057$ circulant with both column and row weights $33$ which is a low-density matrix and satisfies the RC-constraint. The null space of $\mathbf{H}$ over $\mathrm{GF}(2)$ gives a $(1057, 813)$ cyclic PG-LDPC code $\mathcal{C}$ with minimum distance $34$. The Tanner graph of $\mathcal{C}$ has girth $6$ and contains no harmful trapping set with size smaller than $34$. Each VN in the Tanner graph is connected to other $1056$ VNs in the graph, a very large degree of VN-connectivity.

The bit-error performance of $\mathcal{C}$ decoded iteratively with $5$, $10$, and $50$ iterations of the MSA (with scaling factor $0.625$) over the AWGN channel using BPSK signaling is shown in Fig. \ref{fig:eg_ch9_2} \cite[Figure 12.4]{lin2022}. From the figure, we see that with only $5$ iterations of the MSA, the code achieves a BER $10^{-10}$ without error-floor. Due to the high degree of VN-connectivity, decoding converges very fast. From the figure, we see that the performance gaps between $5$, $10$, and $50$ iterations are very small.

\begin{figure}
\centering
\includegraphics[height=3.0in,width=0.8\textwidth]{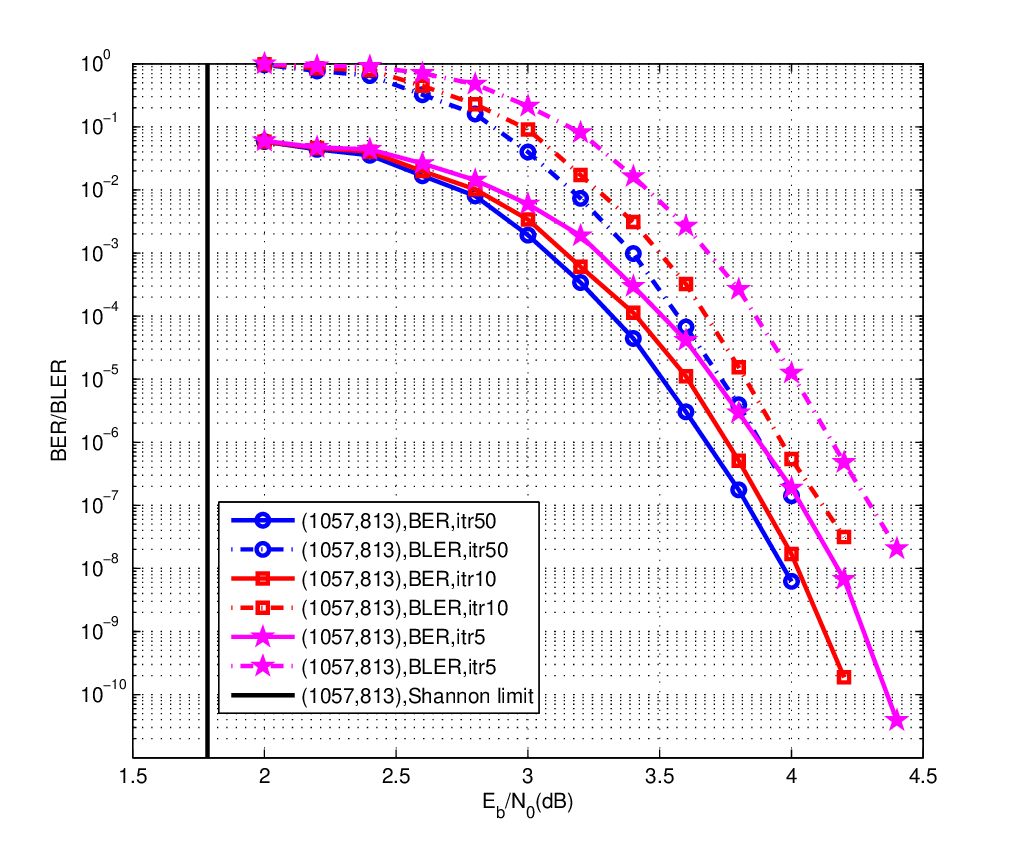}
\caption{The BER and BLER performances of the $(1057, 813)$ cyclic PG-LDPC code given in Example \ref{eg14}.}
\label{fig:eg_ch9_2}
\end{figure}

The dual code of $\mathcal{C}$ is a $(1057, 244)$ code with minimum distance $33$. Note that the dimension $k$ of the $(1057, 813)$ PG-LDPC code is $813$ and $k-1 = 812$ is divisible by $2$ and $4$. Set $t$ to $2$ and $4$. Using $\mathcal{C}$, we can construct convolutional codes with rates $1/2$ and $1/4$, which are $(2, 1, 122)$ and $(4, 1, 61)$ convolutional codes with memory orders $122$ and $61$, respectively. They all have local minimum distances $34$ and free distances lower bounded by $\min\{34, 2 \times 33\} = 34$. Also, they are self-orthogonal. The rate-$1/2$ convolutional code is capable of correcting $16$ or fewer random errors over the BSC with the OSML-decoding. The code is far better than the self-orthogonal code with memory order $425$ listed at the bottom of \cite[Table 13.2]{lin2004}.

Using the $(1057, 813)$ PG-LDPC code $\mathcal{C}$ as the mother code, full-rank code chains of various lengths can be designed. Based on these coding chains, PG-convolutional codes of various rates with local minimum distances at least $34$ can be designed.

Suppose we set $s = 7$. Using the lines of the 2-dimensional projective geometry $\mathrm{PG}(2, 2^7)$, we can construct a $(16513, 14325)$ 2-dimensional PG-LDPC code $\mathcal{C}$ with minimum distance $130$. The dual code of $\mathcal{C}$ is a $(16513, 2188)$ code with minimum distance $129$. Using the code $\mathcal{C}$ as the mother code, we construct code chains of various lengths for constructing PG-LDPC convolutional codes of various rates and constraint lengths with local minimum distance at least $130$. Suppose we set $t = 2$. Using $\mathcal{C}$, we construct a rate $(2, 1, 1094)$ PG-LDPC convolutional code of memory order $1094$ with both local minimum and free distances lower bounded by $130$.
\end{example}

\section{Construction of Convolutional Codes Based on a Class of Doubly Transitive Invariant Cyclic Codes}\label{sect8:convol_from_DTI}

In this section, we present a class of doubly transitive invariant (DTI) cyclic codes. The dual codes of DTI cyclic codes have self-orthogonal structure which allows them to be decoded with OSML decoding and an iterative decoding algorithm based on the belief-propagation principle. Using the dual codes of the DTI cyclic codes, convolutional codes of various rates, constraint lengths and local minimum distances can be constructed.

\subsection{Doubly Transitive Invariant Cyclic Codes}

Let $\mathcal{C}$ be a $(2^m-1, k)$ cyclic code over $\mathrm{GF}(2)$ of length $2^m-1$, $m > 3$, generated by a polynomial $\bg(X)$ over $\mathrm{GF}(2)$ with roots in $\mathrm{GF}(2^m)$ in which $1$ is not a root. Extend each codeword $\mathbf{v} = (v_0, v_1, \ldots, v_{2^m-2})$ in $\mathcal{C}$ by adding an overall parity bit, denoted by $v_\infty = v_0 + v_1 + \cdots + v_{2^m-2}$ to the left of $\mathbf{v}$. Adding the overall parity bit $v_\infty$ to $\mathbf{v}$, we obtain an extended codeword $\mathbf{v}_{\mathrm{ext}} = (v_\infty, v_0, v_1, \ldots, v_{2^m-2})$. The $2^k$ extended codewords of $\mathcal{C}$ form a $(2^m, k)$ linear code $\mathcal{C}_{\mathrm{ext}}$, called an extension of $\mathcal{C}$.

Let $\alpha$ be a primitive element in $\mathrm{GF}(2^m)$. Label the components $v_\infty, v_0, v_1, \ldots, v_{2^m-2}$ of a codeword $\mathbf{v}_{\mathrm{ext}}$ in $\mathcal{C}_{\mathrm{ext}}$ with $\alpha^\infty = 0, \alpha^0, \alpha^1, \ldots, \alpha^{2^m-2}$, called the \hl{location numbers}. Let $Y$ denote the location of a component in $\mathbf{v}_{\mathrm{ext}}$. Consider a permutation that carries the component at the location $Y$ to the location $Z = aY + b$, where $a$ and $b$ are elements in the field $\mathrm{GF}(2^m)$ and $a \neq 0$. This permutation is called an \emph{affine permutation}~\cite{carmichael1956}. The extended cyclic code $\mathcal{C}_{\mathrm{ext}}$ of $\mathcal{C}$ is said to be \emph{invariant under the group of affine permutations} if every affine permutation carries every codeword in $\mathcal{C}_{\mathrm{ext}}$ into another unique codeword in $\mathcal{C}_{\mathrm{ext}}$.

If we set $a = \alpha^\ell$ with $0 \leq \ell < 2^m-1$ and $b = 0$, then the affine permutation $Z = \alpha^\ell Y$ cyclically shifts the $2^m-1$ components $v_0, v_1, \ldots, v_{2^m-2}$ of $\mathbf{v}_{\mathrm{ext}} = (v_\infty, v_0, v_1, \ldots, v_{2^m-2})$ at the locations $\alpha^0, \alpha^1, \ldots, \alpha^{2^m-2}$ $\ell$ places to the right. The component $v_\infty$ of $\mathbf{v}_{\mathrm{ext}}$ remains at the same location $\alpha^\infty$. Since the cyclic code $\mathcal{C}$ is invariant under the affine permutation $Z = \alpha^\ell Y$ and its extension $\mathcal{C}_{\mathrm{ext}}$ is invariant under the affine permutation $Z = \alpha^\ell Y + b$, the cyclic code $\mathcal{C}$ is called a \emph{doubly transitive invariant} (DTI) cyclic code. Several categories of cyclic codes have been proved to be DTI cyclic codes, including primitive BCH codes and cyclic RM codes~\cite{lin1967, kasami1968b}.

\subsection{A Class of DTI Cyclic Codes}

Let $J$ and $L$ be two proper factors of $2^m-1$ such that $J \cdot L = 2^m-1$. Since $2^m-1$ is odd, both $J$ and $L$ are odd. The polynomial $X^{2^m-1} + 1$ over $\mathrm{GF}(2)$ can be factored as a product of two polynomials as follows:
\begin{equation}
X^{2^m-1} + 1 = (1 + X^J)\left(1 + X^J + X^{2J} + \cdots + X^{(L-1)J}\right).
\end{equation}
Let
\begin{equation} \label{eqn:DTI_generator_polynomial}
\bm{\pi}(X) = 1 + X^J + X^{2J} + \cdots + X^{(L-1)J}.
\end{equation}
Since $\alpha^L$ is a root of $X^J + 1$, the polynomial $X^J + 1$ has $\alpha^0 = 1, \alpha^L, \alpha^{2L}, \ldots, \alpha^{(J-1)L}$ as all its roots. Therefore, the polynomial $\bm{\pi}(X)$ has $\alpha^f$ as a root if and only if $f$ is not a multiple of $L$, for $0 < f < 2^m-1$. We can readily see that $\bm{\pi}(X)$ has the $L-1$ consecutive powers $\alpha, \alpha^2, \ldots, \alpha^{L-1}$ of $\alpha$ as roots.

Express the integer $f$ in radix-2 form as follows:
\begin{equation}
f = f_0 + f_1 2 + f_2 2^2 + \cdots + f_{m-1} 2^{m-1},
\end{equation}
where $f_i = 0$ or $1$ for $0 \leq i < m$. Let $f'$ be another nonnegative integer less than $2^m$ whose radix-2 expansion is
\begin{equation}
f' = f'_0 + f'_1 2 + f'_2 2^2 + \cdots + f'_{m-1} 2^{m-1},
\end{equation}
where $f'_i = 0$ or $1$ for $0 \leq i < m$. The integer $f'$ is called a \emph{radix-2 descendant} of $f$ if $f'_i \leq f_i$ for $0 \leq i < m$. Let $\Delta(f)$ denote the set of all nonzero proper radix-2 descendants of $f$.

In~\cite{kasami1968b}, the authors proved that a cyclic code $\mathcal{C}$ of length $2^m-1$ is a DTI cyclic code if and only if for every $\alpha^f$ that is a root of its generator polynomial, $\alpha^{f'}$ is also a root of its generator polynomial for every radix-2 descendant $f' \in \Delta(f)$. Using this necessary and sufficient condition, a DTI cyclic code based on the polynomial $\bm{\pi}(X)$ given by (\ref{eqn:DTI_generator_polynomial}) can be constructed~\cite{lin1980}.

Let $\bh(X)$ be the factor of $\bm{\pi}(X)$ in which if $\alpha^f$ is a root, $\alpha^{f'}$ with $f' \in \Delta(f)$ is also a root. Since $\bm{\pi}(X)$ is a factor of $X^{2^m-1} + 1$, $\bh(X)$ is a factor of $X^{2^m-1} + 1$. Let $k$ be the degree of $\bh(X)$. The code $\mathcal{C}_h$ generated by $\bh(X)$ is a $(2^m-1, 2^m-k-1)$ cyclic code. Since the roots of the generator polynomial $\bh(X)$ of $\mathcal{C}_h$ satisfy the necessary and sufficient condition for the generator polynomial of a DTI cyclic code, $\mathcal{C}_h$ is a DTI cyclic code, and the extended code $\mathcal{C}_{h,\mathrm{ext}}$ of $\mathcal{C}_h$ is invariant under the group of affine permutations.

Since $\bm{\pi}(X)$ is a multiple of $\bh(X)$ and has degree less than $2^m-1$, it is a code polynomial in $\mathcal{C}_h$. If $f < L$, its descendant $f'$ is also less than $L$. Hence, both $\alpha^f$ and $\alpha^{f'}$ are roots of $\bm{\pi}(X)$. Therefore, $\bh(X)$ has $\alpha, \alpha^2, \ldots, \alpha^{L-1}$ as roots. Then, it follows from the BCH-bound that the minimum distance of $\mathcal{C}_h$ is at least $L$. Since the weight of $\bm{\pi}(X)$ is $L$, the minimum distance of $\mathcal{C}_h$ is exactly $L$. Since $L$ is odd and $\mathcal{C}_h$ contains the all-one vector $(1, 1, \ldots, 1)$ as a codeword, $\mathcal{C}_h$ contains half even-weight codewords and half odd-weight codewords. Since $\bm{\pi}(X)$ is used for code construction, $\bm{\pi}(X)$ is referred to as the \emph{code construction polynomial}.

\subsection{Dual Codes of DTI Cyclic Codes}

Let
\begin{equation}
\bz(X) = \frac{X^{2^m-1} + 1}{\bh(X)}
\end{equation}
and
\begin{equation}
\bg_{\mathrm{type-0}}(X) = X^{2^m-k-1} \bz(X^{-1}).
\end{equation}
The polynomial $\bg_{\mathrm{type-0}}(X)$ given above is the reciprocal of $\bz(X)$ which is also a factor of $X^{2^m-1} + 1$. It has degree $2^m-k-1$ and has $(X+1)$ as a factor. The code generated by $\bg_{\mathrm{type-0}}(X)$ is a $(2^m-1, k)$ cyclic code $\mathcal{C}_{\mathrm{type-0}}$ of length $2^m-1$, which is the dual code of the $(2^m-1, 2^m-k-1)$ DTI cyclic code $\mathcal{C}_h$ generated by $\bh(X)$. The generator polynomial $\bg_{\mathrm{type-0}}(X)$ has $J$ consecutive powers of $\alpha$, $\alpha^0 = 1, \alpha, \alpha^2, \ldots, \alpha^{J-1}$, as roots. It follows from the BCH-bound that the minimum distance of $\mathcal{C}_{\mathrm{type-0}}$ is at least $J+1$. All the codewords in $\mathcal{C}_{\mathrm{type-0}}$ have even weights. We call $\mathcal{C}_{\mathrm{type-0}}$ the \emph{type-0 dual of the DTI code} $\mathcal{C}_h$. For simplicity, we call $\mathcal{C}_{\mathrm{type-0}}$ the \emph{type-0 D-DTI code}.

For each code symbol $u_j$, $0 \leq j < 2^m-1$, of a codeword $(u_0, u_1, \ldots, u_{2^m-2})$ in $\mathcal{C}_{\mathrm{type-0}}$, $J$ parity-check sums orthogonal on $u_j$, called \emph{orthogonal check-sums} on $u_j$, can be formed based on the $J$ code polynomials $\bm{\pi}(X), X\bm{\pi}(X), \ldots, X^{J-1}\bm{\pi}(X)$ in $\mathcal{C}_h$ and an affine permutation, say $Z = \alpha Y + \alpha^{2^m-2}$~\cite{lin2004,lin1980}.

To form the orthogonal parity-check sums of the code symbols of a codeword $\mathbf{u} = (u_0, u_1, \ldots, u_{2^m-2})$ in $\mathcal{C}_{\mathrm{type-0}}$ at the locations $\alpha^0, \alpha^1, \ldots, \alpha^{2^m-2}$, we first form $J$ codewords $\mathbf{v}_0, \mathbf{v}_1, \ldots, \mathbf{v}_{J-1}$ in the extension code $\mathcal{C}_{h,\mathrm{ext}}$ of $\mathcal{C}_h$ by adding an overall parity-check symbol to each of the codewords formed by the $J$ code polynomials $\bm{\pi}(X), X\bm{\pi}(X), \ldots, X^{J-1}\bm{\pi}(X)$ in $\mathcal{C}_h$. The overall parity symbol at the location $\alpha^\infty$ of each of the extended codewords $\mathbf{v}_0, \mathbf{v}_1, \ldots, \mathbf{v}_{J-1}$ is $1$, and the weight of each of the extended codewords is $L+1$.

Applying the affine permutation $Z = \alpha Y + \alpha^{2^m-2}$ to each of the extended codewords, we obtain $J$ permuted codewords $\mathbf{w}_0, \mathbf{w}_1, \ldots, \mathbf{w}_{J-1}$ in $\mathcal{C}_{h,\mathrm{ext}}$. The code symbol at the location $\alpha^{2^m-2}$ of each of the permuted codewords is $1$. At a symbol location other than $\alpha^{2^m-2}$, one and only one of the permuted codewords has a $1$-component. Deleting the code symbols of the permuted codewords at the location $\alpha^\infty$, we obtain $J$ codewords $\mathbf{z}_0, \mathbf{z}_1, \ldots, \mathbf{z}_{J-1}$ in the DTI cyclic code $\mathcal{C}_h$.

Using the $J$ codewords $\mathbf{z}_0, \mathbf{z}_1, \ldots, \mathbf{z}_{J-1}$, we form a $J \times (2^m-1)$ matrix $\mathbf{H}_{2^m-2}$ in which the $(2^m-2)$-th column contains $J$ ones and each of the other $2^m-2$ columns contains a single one. Taking the inner products of a codeword $\mathbf{u} = (u_0, u_1, \ldots, u_{2^m-2})$ in $\mathcal{C}_{\mathrm{type-0}}$ with the rows of $\mathbf{H}_{2^m-2}$, we obtain $J$ parity check-sums orthogonal on the code symbol $u_{2^m-2}$ at the location $\alpha^{2^m-2}$. One of the $J$ orthogonal check-sums contains $L$ code symbols and each of the other $J-1$ orthogonal check-sums contains $L+1$ code symbols. The $\mathbf{H}_{2^m-2}$ is called the \emph{orthogonal matrix} on the code symbol at the location $\alpha^{2^m-2}$.

For $1 \leq \ell \leq 2^m-2$, cyclically shifting the columns of $\mathbf{H}_{2^m-2}$ together $\ell$ places to the left gives a $J \times (2^m-1)$ matrix $\mathbf{H}_{2^m-2-\ell}$ orthogonal on the code symbol $u_{2^m-2-\ell}$ of a codeword $\mathbf{u} = (u_0, u_1, \ldots, u_{2^m-2})$ at the location $\alpha^{2^m-2-\ell}$. The $(2^m-2-\ell)$-th column of $\mathbf{H}_{2^m-2-\ell}$ contains $J$ ones and each of the other $2^m-2$ columns contains a single one. The inner products of a codeword $\mathbf{u}$ in $\mathcal{C}_{\mathrm{type-0}}$ with the rows of $\mathbf{H}_{2^m-2-\ell}$ give $J$ check-sums orthogonal on the code symbol $u_{2^m-2-\ell}$ at the location $\alpha^{2^m-2-\ell}$.

Using the matrices $\mathbf{H}_{2^m-2-\ell}$ with $0 \leq \ell \leq 2^m-2$, we can form $J$ orthogonal parity-check-sums on each of the code symbols of a codeword in the type-0 D-DTI code $\mathcal{C}_{\mathrm{type-0}}$. Hence, $\mathcal{C}_{\mathrm{type-0}}$ is an OSML decodable code and is capable of correcting $(J-1)/2$ or fewer random errors over the binary symmetric channel (BSC) with the OSML decoding~\cite{lin2004,lin1980}. The Tanner graph associated with the matrix $\mathbf{H}_{2^m-2-\ell}$ is acyclic and free of cycles.

The overall orthogonal parity-check matrix of $\mathcal{C}_{\mathrm{type-0}}$ is
\begin{equation}
\mathbf{H}_{\mathrm{type-0}} = [\mathbf{H}_{2^m-2} ~ \mathbf{H}_{2^m-3} ~ \cdots  ~ \mathbf{H}_1  ~ \mathbf{H}_0 ]^T,
\end{equation}
which is a $J(2^m-1) \times (2^m-1)$ matrix with column weight $J + 2^m - 2$ and two different row weights $L$ and $L+1$. It has self-orthogonal structure~\cite{lin2004,lin1980}. The density of $\mathbf{H}_{\mathrm{type-0}}$ is $1/J + 1/(2^m-1)$. For large $2^m-1$ and relatively large factor $J$ of $2^m-1$, $\mathbf{H}_{\mathrm{type-0}}$ is a low-density matrix. However, it does not satisfy the RC-constraint, even though it is self-orthogonal. Its associated Tanner graph $\mathcal{G}_{\mathrm{type-0}}$ has cycles of length $4$ but for $0 \leq \ell \leq 2^m-2$, its subgraph associated with the matrix $\mathbf{H}_{2^m-2-\ell}$ orthogonal on code symbol $u_{2^m-2-\ell}$ at the location $2^m-2-\ell$ is free of cycles. Hence, the Tanner graph $\mathcal{G}_{\mathrm{type-0}}$ is cycle-free locally. So, the Tanner graph associated with $\mathbf{H}_{\mathrm{type-0}}$ is composed of $2^m-1$ interconnected subgraphs which are cycle-free.

If the root $\alpha^0 = 1$ is removed from the generator polynomial $\bg_{\mathrm{type-0}}(X)$ of the type-0 D-DTI code $\mathcal{C}_{\mathrm{type-0}}$, we obtain a polynomial
\begin{equation}
\bg_{\mathrm{type-1}}(X) = \frac{\bg_{\mathrm{type-0}}(X)}{X+1}.
\end{equation}
The polynomial $\bg_{\mathrm{type-1}}(X)$ generates a $(2^m-1, k+1)$ cyclic code $\mathcal{C}_{\mathrm{type-1}}$ with minimum distance at least $J$. The code $\mathcal{C}_{\mathrm{type-1}}$ has both even and odd weight codewords and is called the type-1 D-DTI code of $\mathcal{C}_h$. The dual code of $\mathcal{C}_{\mathrm{type-1}}$ is the even-weight subcode of the DTI cyclic code $\mathcal{C}_h$.

The type-1 D-DTI code $\mathcal{C}_{\mathrm{type-1}}$ is also OSML decodable. For each code symbol $u_j$, $0 \leq j < 2^m-1$, of a codeword $\mathbf{u} = (u_0, u_1, \ldots, u_{2^m-2})$ in $\mathcal{C}_{\mathrm{type-1}}$, $J-1$ parity-check sums orthogonal on $u_j$ can be formed based on the code polynomials $\bm{\pi}(X), X\bm{\pi}(X), \ldots, X^{J-1}\bm{\pi}(X)$ in $\mathcal{C}_h$ and the affine permutation $Z = \alpha Y + \alpha^{2^m-2}$.

For $0 \leq \ell \leq 2^m-2$, the matrix $\mathbf{H}^*_{2^m-2-\ell}$ orthogonal on the code symbol $u_{2^m-2-\ell}$ of a codeword $\mathbf{u} = (u_0, u_1, \ldots, u_{2^m-2})$ at the location $\alpha^{2^m-2-\ell}$ is a $(J-1) \times (2^m-1)$ matrix which is obtained by removing the row of weight $L$ from $\mathbf{H}_{2^m-2-l}$. Taking the inner products of a codeword in $\mathcal{C}_{\mathrm{type-1}}$ and $J-1$ rows of $\mathbf{H}^*_{2^m-2-\ell}$, we obtain $J-1$ parity-check sums orthogonal on the code symbol at location $\alpha^{2^m-2-\ell}$. Hence, $\mathcal{C}_{\mathrm{type-1}}$ is capable of correcting $(J-1)/2$ or fewer random errors over the binary symmetric channel.

The overall orthogonal parity-check matrix of $\mathcal{C}_{\mathrm{type-1}}$ is
\begin{equation}
\mathbf{H}_{\mathrm{type-1}} = [ \mathbf{H}^*_{2^m-2} ~ \mathbf{H}^*_{2^m-3} ~ \cdots ~ \mathbf{H}^*_1 ~ \mathbf{H}^*_0 ]^T,
\end{equation}
which is a $(J-1)(2^m-1) \times (2^m-1)$ matrix with column and row weights $J + 2^m - 4$ and $L+1$, respectively. It has a self-orthogonal structure. The Tanner graph $\mathcal{G}_{\mathrm{type-1}}$ associated with $\mathbf{H}_{\mathrm{type-1}}$ is cycle-free locally. 

A special case for the type-1 D-DTI cyclic code $\mathcal{C}_{\mathrm{type-1}}$ is $m = 2s$ with $s \geq 2$. In this case, $2^{2s}-1$ can be factored as the product of $2^s+1$ and $2^s-1$. Set $J = 2^s+1$. Then, the $\mathcal{C}_{\mathrm{type-1}}$ is a 2-dimensional EG-LDPC code whose Tanner graph $\mathcal{G}_{\mathrm{type-1}}$ is free of cycles of length $4$ but has cycles of length $6$. The Tanner graph $\mathcal{G}_{\mathrm{type-1}}$ is composed of $2^{2s}-1$ subgraphs, each is cycle-free.

\subsection{Iterative Decoding of D-DTI Cyclic Codes}

Besides OSML decoding, both type-0 and type-1 D-DTI codes of the DTI cyclic code $\mathcal{C}_h$ can be decoded with an iterative decoding algorithm based on their self-orthogonal and local acyclic structures in the Tanner graphs associated with their overall orthogonal parity-check matrices~\cite{juang2009}. Hence, we may regard these codes as LDPC codes with local cycle-free Tanner graphs.

In each decoding iteration, the reliability of the received code symbol at the location $\alpha^{2^m-2-\ell}$, $0 \leq \ell \leq 2^m-2$, is updated based on its reliability and the reliabilities of the other $2^m-2$ received code symbols computed in the preceding iteration using the parity-check sums orthogonal on it. In each decoding iteration, the local decoders share symbols’ reliability information. Once a decoding iteration is completed, the next decoding iteration is conducted. The decoding process continues until the preset number of decoding iterations is reached. The decoding can be first carried out with OSML decoding. If decoding is successful, stop the decoding. If OSML decoding fails, iterative decoding phase is initiated.

\subsection{Convolutional Codes Constructed Based on D-DTI Cyclic Codes}

Using code chains with D-DTI codes (type-0 or type-1) as the mother codes, we can construct a class of D-DTI-convolutional codes of various rates, constraint lengths and local minimum distances. Each local code of a D-DTI-convolutional code is OSML and iterative decodable.

For even $\ell$, using a D-DTI code, we can construct a rate-$1/\ell$ DTI-LDPC-convolutional code with local minimum distance $J$ (or $J+1$) and free distance lower bounded by $\min\{J, 2(L+1)\}$ (or $\min\{J+1, 2(L+1)\}$). The convolutional code is a self-orthogonal convolutional code.

\subsection*{Example 15}
\begin{example} \label{eg15}
Let $m = 12$. The integer $2^{12} - 1 = 4095$ can be factored as the product of $45$ and $91$. Set $J = 91$ and $L = 45$. The polynomial $X^{4095} + 1$ can be factored into the following product of two polynomials:
\[
X^{4095} + 1 = (1 + X^{91})\bm{\pi}(X)
\]
where
\[
\bm{\pi}(X) = 1 + X^{91} + X^{2 \times 91} + \cdots + X^{44 \times 91}.
\]
Let $\alpha$ be a primitive polynomial in the Galois field $\mathrm{GF}(2^{12})$. Then, the type-1 D-DTI code $\mathcal{C}_{\mathrm{type-1}}$ constructed based on $\bm{\pi}(X)$ is a $(4095, 1649)$ code with minimum distance $91$. The generator polynomial $\bg_{\mathrm{type-1}}(X)$ of $\mathcal{C}_{\mathrm{type-1}}$ has $90$ consecutive powers of $\alpha$, $\alpha, \alpha^2, \ldots, \alpha^{90}$ as roots. The dual code $\mathcal{C}_h$ of $\mathcal{C}_{\mathrm{type-1}}$ is a $(4095, 2446)$ code with minimum distance $46$.

Using the $(4095, 1649)$ type-1 D-DTI code $\mathcal{C}_{\mathrm{type-1}}$, we can construct a rate-$1/2$ $(2, 1, 1223)$ D-DTI-convolutional code of memory order $1223$ with local minimum distance $91$ and free distance lower bounded by $\min\{91, 2 \times 46\} = 91$. The convolutional code is self-orthogonal and can correct $45$ random errors over the BCS with OSML decoding. If we set the shifting factor $t = 4$, we can construct a rate-$1/4$ $(4, 1, 612)$ D-DTI-convolutional code of memory order $612$ with local minimum distance $91$ and free distance lower bounded by $91$.

We can use the $(4095, 1649)$ code $\mathcal{C}_{\mathrm{type-1}}$ as the mother code $\mathcal{C}_0$ to form code chains of various lengths. Suppose we want to construct a 4-fold full rank code chain $\mathcal{C}_0 \supset \mathcal{C}_1 \supset \mathcal{C}_2$ of length $r = 3$. To construct the two descendant codes $\mathcal{C}_1$ and $\mathcal{C}_2$ of $\mathcal{C}_0$, we choose $\mathbf{f}_1(X) =\mathbf{m}_{91}(X)$ and $\mathbf{f}_2(X) = \mathbf{m}_{91}(X) \mathbf{m}_{273}(X)$ as the generator multipliers for the two descendant codes $\mathcal{C}_1$ and $\mathcal{C}_2$ where $\mathbf{m}_{91}(X)$ and $\mathbf{m}_{273}(X)$ are the minimal polynomials of the element $\alpha^{91}$ and $\alpha^{273}$ in $\mathrm{GF}(2^{12})$ which is not a root of the generator polynomial $\bg_0(X)$ of $\mathcal{C}_0$. The degrees of $\mathbf{m}_{91}(X)$ and $\mathbf{m}_{273}(X)$ are $10$ and $4$, respectively. 

With $\mathbf{f}_{1}(X)$ and $\mathbf{f}_{2}(X)$, we form the generator polynomials $\bg_{1}(X) = \bg_{0}(X) \mathbf{f}_{1}(X)$ and $\bg_{2}(X) = \bg_{0}(X) \mathbf{f}_{2}(X)$ of $\mathcal{C}_{1}$ and $\mathcal{C}_{2}$, respectively. The codes $\mathcal{C}_{1}$ and $\mathcal{C}_{2}$ generated by $\bg_{1}(X)$ and $\bg_{2}(X)$ are $(4095, 1637)$ and $(4095, 1633)$ codes, respectively, both with minimum distances lower bounded by $93$. The code chain $\mathcal{C}_0 \supset \mathcal{C}_1 \supset \mathcal{C}_2$ is a $4$-fold full-rank code chain. Using the code chain and $4$-fold decompositions of $\bg_{0}(X)$, $\bg_{1}(X)$, and $\bg_{2}(X)$, we can construct a rate-$3/4$ $(4, 3, 1843)$ D-DTI convolutional code $\mathcal{C}_{\text{convol}}(\bg_{0}, \bg_{1},\bg_{2})$ of constraint length $1843$ with local minimum distance $91$. It is the direct sum of its three rate-$1/4$ convolutional codes $\mathcal{C}_{\text{convol}}(\bg_{0})$, $ \mathcal{C}_{\text{convol}}(\bg_{1})$, $\mathcal{C}_{\text{convol}}(\bg_{2})$ which are $(4, 1, 612)$, $(4, 1, 615)$, and $(4, 1, 616)$ convolutional codes, all with local minimum and free distances lower bounded by $91$.

The $(4095, 1649)$ code $\mathcal{C}_{\text{type-1}}$ and its $(4095, 1648)$ even-weight subcode $\mathcal{C}_{\text{type-0}}$ form a $3$-fold full-rank code chain $\mathcal{C}_{\text{type-1}} \supset C_{\text{type-0}} $ of length $2$. Using this code chain, we can construct a rate-$2/3$ $(3, 2, 1632)$ D-DTI convolutional code of constraint length $1632$ with local minimum distance $91$.
\end{example}

\subsection*{Example 16}
\begin{example} \label{eg16}
Continue Example \ref{eg15}. Suppose we set $J = 45$. We can construct a $(4095, 2074)$ type-1 D-DTI code $\mathcal{C}_{\mathrm{type-1}}$ with a minimum distance $45$. The generator polynomial $\bg_{\mathrm{type-1}}(X)$ of $\mathcal{C}_{\mathrm{type-1}}$ has $\alpha, \alpha^2, \ldots, \alpha^{44}$ of $\mathrm{GF}(2^{12})$ and their conjugates as roots. The elements $1$ and $\alpha^{45}$ are not roots of $\bg(X)$. The minimal polynomial $\mathbf{m}_{45}(X)$ of $\alpha^{45}$ is a polynomial of degree $12$. The code $\mathcal{C}_{\mathrm{type-1}}$ is self-orthogonal and OSML and iterative decodable.

Suppose we set the shifting factor $t = 4$ and want to construct a 4-fold full rank code chain $\mathcal{C}_0 \supset \mathcal{C}_1 \supset \mathcal{C}_2$ of length $3$ with $\mathcal{C}_{\mathrm{type-1}}$ as the mother code $\mathcal{C}_0$ for constructing a rate-$3/4$ D-DTI-convolutional code. To construct such a code chain, we choose $\mathbf{f}_1(X) = X + 1$ and $\mathbf{f}_2(X) = \mathbf{m}_{45}(X)(X + 1)$ as the generator multipliers for the two descendant codes $\mathcal{C}_1$ and $\mathcal{C}_2$ of the mother code $\mathcal{C}_0 = \mathcal{C}_{\mathrm{type-1}}$ of the code chain. The two descendant codes $\mathcal{C}_1$ and $\mathcal{C}_2$ of $\mathcal{C}_0$ generated by $\bg_1(X) = \mathbf{f}_1(X)\bg_0(X)$ and $\bg_2(X) = \mathbf{f}_2(X)\bg_0(X)$ with $\bg_0(X) = \bg_{\mathrm{type-1}}(X)$ are $(4095, 2073)$ and $(4095, 2061)$ with minimum distances $46$ and $48$, respectively. The descendant code $\mathcal{C}_1$ of $\mathcal{C}_0$ is a type-0 D-DTI cyclic code.

Using the 4-fold full rank code chain $\mathcal{C}_0 \supset \mathcal{C}_1 \supset \mathcal{C}_2$, we can construct a rate-$3/4$ $(4, 3, 1520)$ D-DTI-convolutional code $\mathcal{C}_{\mathrm{convol}}(\bg_0, \bg_1, \bg_2)$ with constraint length $1520$ and local minimum distance $45$. It is the direct-sum of its $3$ rate-$1/4$ convolutional codes, $\mathcal{C}_{\mathrm{convol}}(\bg_0)$, $\mathcal{C}_{\mathrm{convol}}(\bg_1)$, and $\mathcal{C}_{\mathrm{convol}}(\bg_2)$, which are $(4, 1, 506)$, $(4, 1, 506)$, and $(4, 1, 508)$ convolutional codes, all with local minimum and free distances lower bounded by $45$.

Using the subchain $\mathcal{C}_0 \supset \mathcal{C}_1$ of the code chain $\mathcal{C}_0 \supset \mathcal{C}_1 \supset \mathcal{C}_2$, we can construct a rate-$2/3$ $(3, 2, 1348)$ D-DTI-convolutional code of constraint length $1348$ and local minimum distance $45$. If we use the subchain $\mathcal{C}_1 \supset \mathcal{C}_2$ of the code chain $\mathcal{C}_0 \supset \mathcal{C}_1 \supset \mathcal{C}_2$, we can construct a rate-$2/3$ $(3, 2, 1352)$ D-DTI-convolutional code of local minimum distance $46$ and constraint length $1352$.

The integer $2^{12} - 1$ can be factored as the product $65 \times 63$. Setting $J = 65$ and $L = 63$, we can construct a $(4095, 3367)$ type-1 D-DTI cyclic code $\mathcal{C}_{\mathrm{type-1}}$ with minimum distance $65$, which is the $(4095, 3367)$ EG-LDPC code given in Example \ref{eg12}.
\end{example}

\section{Conclusion and Remarks}\label{sect9:conclusion}

This paper presents an algebraic method to construct convolutional codes with guaranteed local minimum distances without limit based on cyclic codes of odd lengths. The constructions are simple but effective, and no computer search is needed. For any two positive integers $r$ and $t$ with $1 \leq r < t$, it is shown that a rate-$r/t$ convolutional code $\mathcal{C}_{\text{convol}}$ can be constructed by using a chain of $r$ cyclic codes $\mathcal{C}_0, \mathcal{C}_1, \ldots, \mathcal{C}_{r-1}$ of the same length $n$, which satisfy the inclusion condition, $\mathcal{C}_0 \supset \mathcal{C}_1 \supset \cdots \supset \mathcal{C}_{r-1}$. Such a convolutional code $\mathcal{C}_{\text{convol}}$ is composed of a semi-infinite chain of identical local codes with distinctive algebraic and geometric structures. Each local code $\mathcal{C}_{\text{local}}$ of the convolutional code $\mathcal{C}_{\text{convol}}$ is formed from the $r$ cyclic codes in the code chain and is a specially localized subcode of the mother code $\mathcal{C}_0$ in the code chain. The minimum distance $d_{\text{local}}$ of each local code of $\mathcal{C}_{\text{convol}}$ is lower bounded by the minimum distance $d_0$ of the mother code $\mathcal{C}_0$ in the code chain. The rate-$r/t$ convolutional code $\mathcal{C}_{\text{convol}}$ constructed based on the code chain $\mathcal{C}_0 \supset \mathcal{C}_1 \supset \cdots \supset \mathcal{C}_{r-1}$ is the direct-sum of $r$ rate-$1/t$ convolutional codes, each constructed based on a code in the code chain $\mathcal{C}_0 \supset \mathcal{C}_1 \supset \cdots \supset \mathcal{C}_{r-1}$. For $r > 1$, a rate-$r/t$ convolutional code $\mathcal{C}_{\text{convol}}$ has multi-layer structure. For the case with $t$ as an even integer, the free distance $d_{\text{free}}$ of the rate-$1/t$ convolutional code constructed based on the $i$-th descendant code $\mathcal{C}_i$, $0 \leq i < r$, in the code-chain $\mathcal{C}_0 \supset \mathcal{C}_1 \supset \cdots \supset \mathcal{C}_{r-1}$ is lower bound by $\min\{d_i, 2d_{i,h}\}$ where $d_i$ and $d_{h,i}$ are the minimum distances of $\mathcal{C}_i$ and its dual code $\mathcal{C}_{i,h}$, respectively.

In the paper, convolutional codes with guaranteed local minimum distances have been constructed based on 4 major classes of cyclic codes: primitive BCH codes, cyclic RM codes, two-dimensional Euclidean and projective LDPC codes, and doubly transitive invariant codes. All these convolutional codes have either distinct algebraic structures or geometric structures. A convolutional code constructed based on a geometric LDPC code or a doubly transitive invariant LDPC is a self-orthogonal convolutional code, and its local codes can be decoded iteratively. For even shifting factor $t$, the rate-$1/t$ self-orthogonal convolutional code constructed based on a geometric LDPC code is much better than a known self-orthogonal convolutional code listed in~\cite[Table 13.2]{lin2004}.

Besides the construction of convolutional codes based on the primitive BCH codes, cyclic RM codes, two-dimensional Euclidean and projective LDPC codes, and doubly transitive invariant codes, convolutional codes with guaranteed local minimum distances can also be constructed based on other types of cyclic codes, such as nonprimitive BCH codes~\cite{lin2004}, generalized RM codes~\cite{kasami1968a,weldon1968,delsarte1970}, $m$-dimensional Euclidean geometry cyclic codes $m > 2$~\cite{lin2022,kou2001,lin2004}, polynomial codes~\cite{kasami1968c}, and quadratic residue codes~\cite{gleason1970,blake1975}. Nonbinary convolutional codes with guaranteed local minimum distances can be constructed based on nonbinary cyclic codes, especially Reed-Solomon codes, in a similar manner.

While decoding algorithms for the convolutional codes constructed are not considered here, it is anticipated that hybrid algorithms using convolutional and block decoding algorithms will be effective for decoding convolutional codes with large local minimum distance codes. For a convolutional code constructed based on a code chain with a cyclic LDPC mother code, it can be decoded with a sliding-window soft-decision iterative decoding scheme.

The paper basically presents a wider bridge to connect convolutional codes to cyclic block codes with distinctive algebraic and geometric structures.


\end{document}